\documentclass{article} \newcommand{\ARXIVVERSIONONLY}[1]{#1} \newcommand{\ELSVERSIONONLY}[1]{} 

\usepackage[T1]{fontenc} 
\usepackage{microtype} 
\usepackage[table]{xcolor} 
\usepackage{multirow}
\usepackage{soul} 
\usepackage{amsmath} 
\usepackage{graphicx} 
    \newcommand{\graphicsScale}{1} 
    \graphicspath{{./graphics/}}
\usepackage{subcaption} 
\usepackage{enumitem} 
    \newlist{thmEnumerate}{enumerate}{1}
    \setlist[thmEnumerate, 1]{label=(\alph{thmEnumeratei})}
    \newlist{rules}{enumerate}{1}
    \setlist[rules, 1]{label=(R\arabic{rulesi})}
    \newlist{Cases}{enumerate}{5}
    \setlist[Cases]{wide}
    \setlist[Cases, 1]{label=\textbf{Case~\arabic{Casesi}:},ref=\arabic*}
    \setlist[Cases, 2]{label=\textbf{Case~\arabic{Casesi}.\arabic{Casesii}:},ref=\arabic*}
    \setlist[Cases, 3]{label=\textbf{Case~\arabic{Casesi}.\arabic{Casesii}.\arabic{Casesiii}:},ref=\arabic*}
    \setlist[Cases, 4]{label=\textbf{Case~\arabic{Casesi}.\arabic{Casesii}.\arabic{Casesiii}.\arabic{Casesiv}:},ref=\arabic*}
    \setlist[Cases, 5]{label=\textbf{Case~\arabic{Casesi}.\arabic{Casesii}.\arabic{Casesiii}.\arabic{Casesiv}.\arabic{Casesv}:},ref=\arabic*}
    \newlist{inlineEnum}{enumerate*}{1}
    \setlist[inlineEnum, 1]{label=(\roman{inlineEnumi})}
    \renewlist{enumerate}{enumerate}{10}
    \setlist[enumerate]{label*=\arabic*.}
    \setlistdepth{10} 
\usepackage{mleftright}
\usepackage{amsthm}
    \newtheorem{theorem}{Theorem}[section]
    \newtheorem{lemma}{Lemma}[section]
    \newtheorem{conjecture}{Conjecture}
    
    \newtheorem{observation}{Observation}
\ARXIVVERSIONONLY{
    \usepackage[pdfusetitle,breaklinks]{hyperref} 
    \usepackage{authblk}
    \usepackage{fullpage}
}
\ELSVERSIONONLY{
    \usepackage[breaklinks]{hyperref} 
}

\ELSVERSIONONLY{
    \usepackage{lineno}
        \linenumbers
}
\showboxdepth=\maxdimen 
\showboxbreadth=\maxdimen 

\newcommand{\OO}{O} 
\newcommand{\PP}{\mathcal{P}} 
\newcommand{\NP}{\mathcal{N\!P}} 
\newcommand{\mya}{a} 
\newcommand{\polygon}{P} 
\newcommand{\SATn}{\mathsf{SAT}} 
\newcommand{\SAT}[1]{#1\textnormal{-}\SATn} 
\newcommand{\HORNSAT}{\mathsf{HORNSAT}} 
\newcommand{\STAB}[1]{#1\textnormal{-}\mathsf{STAB}} 
\newcommand{\CSTAB}[1]{#1\textnormal{-}\mathsf{CSTAB}} 
\newcommand{\RPM}{\mathsf{RPM}\textnormal{-}\SAT{3}} 
\newcommand{\partition}{\mathbf{R}} 
\newcommand{\segment}{\mathbf{s}} 
\newcommand{\pixel}{\xi} 
\newcommand{\vertexSetVariable}{V} 
\newcommand{\edgeSetVariable}{E} 
\newcommand{\pixelSetVariable}{\Xi} 
\newcommand{\vertexSetOf}[1]{\vertexSetVariable_{#1}} 
\newcommand{\edgeSetOf}[1]{\edgeSetVariable_{#1}} 
\newcommand{\pixelSetOf}[1]{\pixelSetVariable_{#1}} 
\newcommand{\formula}{\phi} 
\newcommand{\tw}{\mathit{tw}} 
\newcommand{\hor}{{\textnormal{hor}}} 
\newcommand{\ver}{{\textnormal{ver}}} 
\newcommand{\EdgesAndRefSeg}[1]{U} 
\newcommand{\reflexSegmentsOf}[1]{W} 
\newcommand{\sgt}[2]{\overline{#1 #2}} 
\newcommand{\graph}{H} 
\newcommand{\treeDecomposition}{\mathcal{T}} 
\newcommand{\pixelGraph}[1]{G} 
\newcommand{\extGraph}[1]{G'} 
\newcommand{\Equal}{:= \quad} 
\renewcommand{\implies}{\Rightarrow} 
\renewcommand{\iff}{\Leftrightarrow} 
\newcommand{\MSOFormulaConf}[1]{\formula^{\mathit{conf}}_{#1}} 
\newcommand{\MSOFormula}[1]{\formula _{ #1 }} 
\newcommand{\isPartition}{\mathit{isPartition}} 
\newcommand{\isConforming}{\mathit{isConforming}} 
\newcommand{\hasStabbingNumberAtMost}[1]{\mathit{hasStabbingNumberAtMost}_{#1}} 
\newcommand{\isOnBoundary}{\mathit{isOnBoundary}} 
\newcommand{\yieldsOnlyRectangles}{\mathit{yieldsOnlyRectangles}} 
\newcommand{\containsOnlyEdgesOfThePixelGraph}{\mathit{containsOnlyEdgesOfThePixelGraph}} 
\newcommand{\coversAllEdgesOfThePolygon}{\mathit{coversAllEdgesOfThePolygon}} 
\newcommand{\isAdjacent}[1]{\mathit{isAdjacent}_{#1}} 
\newcommand{\areAdjacent}{\mathit{areAdjacent}} 
\newcommand{\coversAllPixels}{\mathit{coversAllPixels}} 
\newcommand{\isCountIncrementedAtEachStab}{\mathit{isCountIncrementedAtEachStab}} 
\newcommand{\isCountNotDecrementedAtEachNonStab}{\mathit{isCountNotDecrementedAtEachNonStab}} 
\newcommand{\xcoor}{\mathrm{x}} 
\newcommand{\variable}{x} 
\newcommand{\literal}{y} 
\newcommand{\FOvariable}{x} 
\newcommand{\Variable}{X} 
\newcommand{\comb}{P_{\textnormal{comb}}} 
\newcommand{\ors}{\mathbf{S}} 
\newcommand{\Rhor}{\mathbf{R}^{\hor}} 
\newcommand{\Rver}{\mathbf{R}^{\ver}} 
\newcommand{\Shor}{\mathbf{S}^{\hor}} 
\newcommand{\Sver}{\mathbf{S}^{\ver}} 
\newcommand{\Lhor}{\mathbf{L}^{\hor}} 
\newcommand{\Lver}{\mathbf{L}^{\ver}} 
\newcommand{\query}{\kappa} 
\newcommand{\ColorInGreen}{\mathtt{ColorInGreen}} 
\newcommand{\ColorInBlue}{\mathtt{ColorInBlue}} 
\newcommand{\horOf}[1]{\mathbf{h}_{#1}} 
\newcommand{\verOf}[1]{\mathbf{v}_{#1}} 
\newcommand{\function}{f} 
\newcommand{\ray}{\mathbf{r}} 
\newcommand{\gadget}{\Gamma} 
\newcommand{\clause}{c} 
\newcommand{\north}{\mathit{North}} 
\newcommand{\south}{\mathit{South}} 
\newcommand{\est}{\mathit{East}} 
\newcommand{\west}{\mathit{West}} 
\newcommand{\interpreted}[2]{[#2]_{#1}} 

\newcommand{\mytitle}{Computing Conforming Partitions with Low Stabbing Number for Rectilinear Polygons}
\newcommand{\mythanks}{This paper is dedicated to the memory of our friend Saeed, whose work inspired the project. This work is funded in part by the Natural Sciences and Engineering Research Council of Canada (NSERC). A preliminary version of this article appeared in WALCOM 2025.}

\newcommand{\myabstract}{
A \emph{conforming partition} of a rectilinear $ n $-gon $ \polygon $ (possibly with holes) is a partition of $ \polygon $ into rectangles without using Steiner points (i.e., all corners of all rectangles must lie on the boundary of $ \polygon $).
The stabbing number of such a partition is the maximum number of rectangles intersected by an axis-aligned segment lying in the interior of $ \polygon $. 
In this paper, we examine the problem of computing conforming partitions with low stabbing number. We show that computing a conforming partition with stabbing number at most~$ 4 $ is $ \NP $-hard, which strengthens a previously known hardness result [Durocher \& Mehrabi, Theor. Comput. Sci. 689: 157-168 (2017)] and eliminates the possibility for fixed-parameter-tractable algorithms parameterized by the stabbing number unless $ \PP = \NP $. 
In contrast, we give 
(i) an $ \OO ( n \log n ) $-time algorithm to decide whether a conforming partition with stabbing number~$ 2 $ exists, 
(ii) a fixed-parameter-tractable algorithm parameterized by both the stabbing number and treewidth of the pixel graph of the polygon, and 
(iii) a fixed-parameter-tractable algorithm parameterized by the stabbing number for polygons without holes in general position.
}

\ARXIVVERSIONONLY{
    \title{\mytitle\texorpdfstring{\thanks{\mythanks}}{}}
}
\ELSVERSIONONLY{
    \title{\mytitle\tnoteref{t1}}
    \tnotetext[t1]{\mythanks}
}

\ARXIVVERSIONONLY{
    \author[1]{Therese Biedl}
    \author[2]{Stephane Durocher}
    \author[3]{Debajyoti Mondal}
    \author[4]{Rahnuma Islam Nishat}
    \author[4]{Bastien Rivier}
    \affil[1]{University of Waterloo, Waterloo, Canada\\ \protect\href{mailto:biedl@uwaterloo.ca}{\texttt{biedl@uwaterloo.ca}}}
    \affil[2]{University of Manitoba, Winnipeg, Canada\\ \protect\href{mailto:durocher@cs.umanitoba.ca}{\texttt{durocher@cs.umanitoba.ca}}}
    \affil[3]{University of Saskatchewan, Saskatoon, Canada\\ \protect\href{mailto:dmondal@cs.usask.ca}{\texttt{dmondal@cs.usask.ca}}}
    \affil[4]{Brock University, Saint Catharines, Canada\\ \protect\href{mailto:rnishat@brocku.ca}{\texttt{rnishat@brocku.ca}}, \protect\href{mailto:brivier@brocku.ca}{\texttt{brivier@brocku.ca}}}
    \date{}
}
\ELSVERSIONONLY{
    \author[1]{Therese Biedl}
    \ead{biedl@uwaterloo.ca}
    \author[2]{Stephane Durocher}
    \ead{durocher@cs.umanitoba.ca}
    \author[3]{Debajyoti Mondal}
    \ead{dmondal@cs.usask.ca}
    \author[4]{Rahnuma Islam Nishat}
    \ead{rnishat@brocku.ca}
    \author[4]{Bastien Rivier\texorpdfstring{\corref{cor1}}{}}
    \ead{brivier@brocku.ca}
    \affiliation[1]{
        organization={University of Waterloo},
        city={Waterloo},
        country={Canada}
        }
    \affiliation[2]{
        organization={University of Manitoba},
        city={Winnipeg},
        country={Canada}
        }
    \affiliation[3]{
        organization={University of Saskatchewan},
        city={Saskatoon},
        country={Canada}
        }
    \affiliation[4]{
        organization={Brock University},
        city={Saint Catharines},
        country={Canada}
        }
    \cortext[cor1]{Corresponding author}
}

\begin{document}

\ELSVERSIONONLY{
    \begin{abstract}
    \myabstract
    \end{abstract}

    \newpageafter{abstract}
}

\maketitle

\ARXIVVERSIONONLY{
    \begin{abstract}
        \myabstract
    \end{abstract}
}

\section{Introduction} \label{sec:intro}

Partitioning an $ n $-gon $ \polygon $ with nice properties is a fundamental paradigm in computational geometry. 
We are interested in the \emph{stabbing number} of a partition, i.e., the maximum number of elements of the partition that are intersected by a line segment that lies in the interior of $ \polygon $. 
Consider a partition of $ \polygon $ into triangles. 
Such a partition yields a data structure to efficiently process a ray-shooting query starting inside $ \polygon $: compute the first intersection with the boundary of $ \polygon $ by traversing the sequence of triangles that are stabbed by the ray. 
Since the running time is proportional to the number of stabbed triangles, it is desirable to find a triangular partition such that any ray intersects as few triangles as possible, or in other words, such that the stabbing number is minimum. 
Hershberger and  Suri~\cite{HS95} showed that every simple (without self-intersections) $ n $-gon with $ h \geq 0 $ holes has a triangular partition with stabbing number $ \OO ( \sqrt{ h + 1 } \log n ) $ and that there exist polygons without holes where any triangular partition has stabbing number $ \Omega ( \log n ) $. 
The literature also contains an $ \OO ( 1 ) $-approximation algorithm for minimizing the stabbing number of triangular partitions for polygons without holes~\cite{AADK11}.

\begin{figure}[ht]
    \hspace*{\stretch{1}}
    \subcaptionbox{}{\includegraphics[scale=\graphicsScale,page=1]{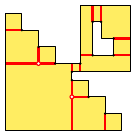}}
    \hspace*{\stretch{2}}
    \subcaptionbox{}{\includegraphics[scale=\graphicsScale,page=2]{partitionWithSteinerPointsVSConformingWithHoles.pdf}}
    \hspace*{\stretch{2}}
    \subcaptionbox{}{\includegraphics[scale=\graphicsScale,page=5]{partitionWithSteinerPointsVSConformingWithHoles.pdf}}
    \hspace*{\stretch{1}}
    \caption{
    (a) An optimal rectangular partition of a polygon $ \polygon _ 1 $ in general position  with one hole, using two Steiner points (tiny hollow circles) with stabbing number~$ 3 $. The portion of the edges of partition rectangles that are not on the boundary of $ \polygon _ 1 $ are plain bold (red).
    (b) An optimal conforming partition of $ \polygon _ 1 $ with stabbing number~$ 4 $.
    (c) The pixel graph of $ \polygon _ 1 $.}
    \label{fig:partitionWithSteinerPointsVSConforming}
\end{figure}

In this paper, we restrict the attention to a problem motivated by orthogonal ray-shooting (see~\cite{GiyoraKaplan} for an example where orthogonal ray-shooting is considered).
We only consider \emph{rectangular} partitions, i.e., partitions of rectilinear polygons (i.e., polygons whose edges are axis-aligned) into rectangles.
Alternatively, a rectangular partition is obtained by cutting $ \polygon $ along disjoint axis-aligned open segments (i.e., segments excluding endpoints) lying in the interior of $ \polygon $ that are of maximal length.
We say that such a set of segments is a \emph{cut set}, and freely consider a rectangular partition as a cut set.
As for the stabbing number, we only consider line segments that are in the interior of the polygon and axis-aligned; we call these \emph{stabbing segments}.
More precisely, we study the following problem for a rectilinear $ n $-gon $ \polygon $, possibly with holes (unless mentioned otherwise): partition $ \polygon $ into rectangles while minimizing the \emph{stabbing number of the partition}, 
that is, the maximum over all stabbing segments $ \segment $ of the number of partition rectangles intersected by $ \segment $. 
We say that such a rectangular partition is \emph{optimal}, and call its stabbing number the \emph{stabbing number} of $ \polygon $.
In fact, it is known (Lemma~1 in~\cite{AADK11}\footnote{Although this lemma is stated and proved in the context of a polygon without holes, its proof holds without modification for polygon with holes.}) that to minimize the stabbing number of a rectangular partition (viewed as a cut set), it suffices to restrict to \emph{anchored} segments, i.e., segments that have at least one endpoint that is a \emph{reflex} vertex of $ \polygon $ (i.e., a vertex of $ \polygon $ with a reflex interior angle).
Figure~\ref{fig:partitionWithSteinerPointsVSConforming}(a) shows an example of an optimal partition; note the use of \emph{Steiner points}, i.e., corners of partition rectangles that lie in the interior of $ \polygon $\footnote{In contrast, in the context of a triangulation of a polygon, a Steiner point is usually defined as a corner of a partition triangle that is not a vertex of the polygon (which includes corners lying on the boundary of the polygon).}.
Similar to triangular partitions, every rectilinear polygon with $ h \geq 0 $ holes has stabbing number $ \OO ( \sqrt{ h { + } 1 } \log n ) $, and there exist polygons (without holes) of arbitrary size with stabbing number $ \Omega ( \log n ) $~\cite{DV94}.  
However, there also exist arbitrary-size polygons (without holes) with stabbing number $ \OO ( 1 ) $. 
This motivates the design of approximation algorithms for computing the stabbing number of simple rectilinear polygons, such as the~$ 3 $-approximation algorithm in~\cite{AADK11}for polygons without holes.
An interesting open problem in this context is to determine the computational complexity of computing the stabbing number for simple polygons with or without holes. 
Although this question remains open in general, there has been some progress on a variant of rectangular partition called conforming partition.

A \emph{conforming partition}\footnote{The term ``conforming'' in this context appears in~\cite{DM17}.} of a rectilinear polygon $ \polygon $ is a rectangular partition without Steiner points. 
Again, we say that a conforming partition is \emph{optimal} if its stabbing number is minimum among all the conforming partitions, and we call this stabbing number the \emph{conforming stabbing number of $ \polygon $}.
To minimize the stabbing number in this case, it suffices to restrict the (cut sets of) conforming partitions to segments with \emph{both} endpoints on the boundary of $ \polygon $, at least one of which being a reflex vertex of $ \polygon $. 
We say that such segments are \emph{reflex} (Figure~\ref{fig:partitionWithSteinerPointsVSConforming}(b)).
Durocher and Mehrabi~\cite{DM17,DM17erratum} showed that computing an optimal conforming partition is $ \NP $-hard for rectilinear polygons with holes,
and gave a~$ 2 $-approximation algorithm to compute their conforming stabbing number (see also~\cite{piva2017minimum} for experimental results). 
However, the complexity of the problem remains open for simple polygons without holes. 
The state of the art covered in this introduction is summarized in Table~\ref{tab:stateOfTheArt}.

\begin{table}[htb]
    \newcommand{\mr}[2]{\multirow{#1}{5.2em}{#2}}
    \newcommand{\highlight}{\cellcolor{yellow!50}}
    \setlength{\tabcolsep}{0.1em}
    \renewcommand{\arraystretch}{1.1}
    \footnotesize
    \centering
    \begin{tabular}{@{}l@{}c@{}ccc@{}}
       Partition:   & Holes: & Stabbing Number: & Complexity Class: & Approx. Factor:\\\hline

\mr{2}{Triangular}  & $ h $  & $ \Omega ( \log n )$~\cite{HS95}, $ \OO ( \sqrt{ h { + } 1 } \log n ) $~\cite{HS95} & - & $ \OO ( \sqrt{ h { + } 1 } \log n ) $~\cite{HS95} \\

                    & $ 0 $  & $ \Theta ( \log n ) $~\cite{HS95} & - & $ \OO ( 1 ) $~\cite{AADK11} \\\hline

\mr{2}{Rectangular} & $ h $  & $ \Omega ( \sqrt{ h } { + } \log n ) $~\cite{DV94}, $ \OO ( \sqrt{ h { + } 1 } \log n ) $~\cite{DV94} & $ \NP $-c. (\ref{thm:NPHardStabbingN4+})\highlight & $ \OO ( \sqrt{ h { + } 1 } \log n ) $~\cite{DV94} \\ 

                    & $ 0 $  & $ \Theta ( \log n ) $~\cite{DV94} & $ \NP $ & $ 3 $~\cite{AADK11} \\\hline

\mr{2}{Conforming\\Rectangular}  & $ h $  & $ \Omega ( \sqrt{ h } )$~\cite{DV94}, $ \Theta ( n ) $ & $ \NP $-c.~\cite{DM17,DM17erratum}, (\ref{thm:NPHardStabbingN4+})\highlight & $ 2 $~\cite{DM17,DM17erratum} \\ 

                    & $ 0 $  & $ \Theta ( n ) $ & $ \NP $ & $ 2 $~\cite{DM17,DM17erratum}  
    \end{tabular}
    \caption{
    Summary of the state of the art for computing the stabbing number of a polygon $ \polygon $ with $ n $ vertices (highlighted cells signal contributions and the theorem number is given in parenthesis).
    The first column gives the partition type of $ \polygon $.
    The second column gives the number of holes of $ \polygon $, ``$ h $'' meaning any non-negative integer $ h $.
    The third column gives known bounds on the minimum upper bound on the stabbing number of $ \polygon $. Bounds without citation are trivial.
    The fourth column gives the time complexity of deciding whether the stabbing number of $ \polygon $ is at most a given bound. 
    The fifth column gives the best approximation factor known for polynomial time algorithms approximating the stabbing number of $ \polygon $.
    }
    \label{tab:stateOfTheArt}
\end{table}

\paragraph{Contributions\ARXIVVERSIONONLY{.}} In this paper, we investigate the problem of computing an optimal conforming partition of an arbitrary rectilinear $ n $-gon $ \polygon $ (possibly with holes) from the perspective of designing fixed-parameter tractable (FPT) algorithms. (Recall that an algorithm is FTP if its running time is of the form $ f ( t ) n ^{ \OO ( 1 ) } $ for some parameter $ t $ and some computable function $ f $ that is independent of $ n $.)
The choice of a suitable parameter is critical.
To decide whether $ \polygon $ admits conforming partitions with stabbing number at most $ k $ using an FPT algorithm, the parameter $ t = k $ is a natural choice.
We show that an FPT algorithm with such a parameter does not exist unless $ \PP = \NP $. 
Specifically, we define $ \STAB{ k } $ (respectively, $ \CSTAB{ k } $) as the problem of deciding whether the stabbing number (respectively, conforming stabbing number) of $ \polygon $ is at most~$ k $.
We show that both $ \CSTAB{ k } $ and $ \STAB{ k } $ remain $ \NP $-hard when  $ k \geq 4 $ (Section~\ref{sec:4+intratable} and Section~\ref{sec:intractableProof}). 
This strengthens the result of Durocher and Mehrabi~\cite{DM17} stating that it is $ \NP $-hard to determine whether the conforming stabbing number is $ \Theta ( \sqrt{ n } ) $.

Our strengthened hardness result raises two interesting questions. 
First, is it decidable whether $ \polygon $ admits a conforming partition with stabbing number at most~$ 2 $ or~$ 3 $ in polynomial time? 
Second, are there other natural parameters for FPT algorithms to compute optimal conforming partitions of $ \polygon $?
For the former, we give an $ \OO ( n \log n ) $-time algorithm to decide whether a conforming partition with stabbing number~$ 2 $ exists (Section~\ref{sec:2tractable}); this leaves the case of stabbing number~$ 3 $ open.
For the latter, we give two FPT algorithms to test whether $ \polygon $ has conforming stabbing number at most $ k $ (Section~\ref{sec:tw} and Section~\ref{sec:gate-free-hole-free}). 
One is parameterized by $ t = k + \tw ( \polygon ) $, the tuple consisting of $ k $ and the \emph{treewidth} of $ \polygon $ (defined as the treewidth of the \emph{pixel graph} of $ \polygon $ illustrated in Figure~\ref{fig:partitionWithSteinerPointsVSConforming}(c)), the other is specific to polygons without holes in \emph{general position} and is parameterized by $ t = k $ alone (missing definitions are in Section~\ref{sec:preliminaries}). 
The treewidth has frequently been used as a parameter for FPT algorithms for graph problems but also for solving problems on polygons, see e.g.~\cite{BiedlMehrabi}.
The contributions are summarized in Table~\ref{tab:stateOfTheArt} and~\ref{tab:contributions}.

\begin{table}[htb]
    \newcommand{\mr}[2]{\multirow{#1}{5.2em}{#2}}
    \newcommand{\highlight}{\cellcolor{yellow!50}}
    \renewcommand{\arraystretch}{1.1}
    \footnotesize
    \centering
    \begin{tabular}{@{}l@{\hspace{-1em}}|@{\hspace{.3em}}c@{~}c@{~}c@{\hspace{.3em}}|cccc}
Problem: & \multicolumn{3}{@{}c|}{Restrictions:} & $ k = 1 $: & $ k = 2 $: & $ k = 3 $: & $ k \geq 4 $: \\\hline

\mr{8}{$ \STAB{ k }  $} &            &          &      & $ \OO ( 1 ) $ & $ \OO ( n \log n ) $ (\ref{thm:decideStabbing2})\highlight & FPT\highlight & $ \NP $-c. (\ref{thm:NPHardStabbingN4+}), FPT\highlight \\ 

                         &            &          & thin & $ \OO ( 1 ) $ & $ \OO ( n \log n ) $ (\ref{thm:decideStabbing2})\highlight & FPT\highlight & $ \NP $-c. (\ref{thm:NPHardStabbingN4+}), FPT\highlight \\ 

                         &            & gen.pos. &      & $ \OO ( 1 ) $ & $ \OO ( n \log n ) $ (\ref{thm:decideStabbing2})\highlight & FPT\highlight & $ \NP $-c. (\ref{thm:NPHardStabbingN4+}), FPT\highlight \\ 

                         &            & gen.pos. & thin & $ \OO ( 1 ) $ & $ \OO ( 1 ) $ (\ref{lem:thinGeneral})\highlight & $ \OO ( 1 ) $ (\ref{lem:thinGeneral})\highlight & $ \OO ( 1 ) $ (\ref{lem:thinGeneral})\highlight \\ 

                         & \st{holes} &          &      & $ \OO ( 1 ) $ & $ \OO ( n \log n ) $ (\ref{thm:decideStabbing2})\highlight & FPT\highlight & FPT\highlight \\ 

                         & \st{holes} &          & thin & $ \OO ( 1 ) $ & $ \OO ( n \log n ) $ (\ref{thm:decideStabbing2})\highlight & FPT\highlight & FPT\highlight \\ 

                         & \st{holes} & gen.pos. &      & $ \OO ( 1 ) $ & $ \OO ( n \log n ) $ (\ref{thm:decideStabbing2})\highlight & FPT\highlight & FPT\highlight \\ 

                         & \st{holes} & gen.pos. & thin & $ \OO ( 1 ) $ & $ \OO ( 1 ) $ (\ref{lem:thinGeneral})\highlight & $ \OO ( 1 ) $ (\ref{lem:thinGeneral})\highlight & $ \OO ( 1 ) $ (\ref{lem:thinGeneral})\highlight \\ \hline

\mr{8}{$ \CSTAB{ k } $} &            &          &      & $ \OO ( 1 ) $ & $ \OO ( n \log n ) $ (\ref{thm:decideStabbing2})\highlight & FPT\highlight & $ \NP $-c. (\ref{thm:NPHardStabbingN4+}), FPT\highlight \\ 

                         &            &          & thin & $ \OO ( 1 ) $ & $ \OO ( n \log n ) $ (\ref{thm:decideStabbing2})\highlight & FPT\highlight & $ \NP $-c. (\ref{thm:NPHardStabbingN4+}), FPT\highlight \\ 

                         &            & gen.pos. &      & $ \OO ( 1 ) $ & $ \OO ( n \log n ) $ (\ref{thm:decideStabbing2})\highlight & FPT\highlight & $ \NP $-c. (\ref{thm:NPHardStabbingN4+}), FPT\highlight \\ 

                         &            & gen.pos. & thin & $ \OO ( 1 ) $ & $ \OO ( 1 ) $ (\ref{lem:thinGeneral})\highlight & $ \OO ( 1 ) $ (\ref{lem:thinGeneral})\highlight & $ \OO ( 1 ) $ (\ref{lem:thinGeneral})\highlight \\ 

                         & \st{holes} &          &      & $ \OO ( 1 ) $ & $ \OO ( n \log n ) $ (\ref{thm:decideStabbing2})\highlight & FPT\highlight & FPT\highlight \\ 

                         & \st{holes} &          & thin & $ \OO ( 1 ) $ & $ \OO ( n \log n ) $ (\ref{thm:decideStabbing2})\highlight & FPT\highlight & FPT\highlight \\ 

                         & \st{holes} & gen.pos. &      & $ \OO ( 1 ) $ & $ \OO ( n \log n ) $ (\ref{thm:decideStabbing2})\highlight & FPT$'$\highlight & FPT$'$\highlight \\

                         & \st{holes} & gen.pos. & thin & $ \OO ( 1 ) $ & $ \OO ( 1 ) $ (\ref{lem:thinGeneral})\highlight & $ \OO ( 1 ) $ (\ref{lem:thinGeneral})\highlight & $ \OO ( 1 ) $ (\ref{lem:thinGeneral})\highlight 
    \end{tabular}
    \caption{Summary of the contributions (using the same conventions as in Table~\ref{tab:stateOfTheArt}).
    The first column gives the decision problem: $ \STAB{ k } $ or $ \CSTAB{ k } $.
    The second column gives the restrictions on the input polygon $ \polygon $, ``\st{holes}'' standing for ``without holes'' and ``gen.pos.'' for ``general position''.
    In the remaining, ``FPT'' stands for ``$ \OO ( f ( k { + } \ell ) n ^ 2 )$ (\ref{thm:treewidth})'' and ``FPT$'$'' stands for ``$ \OO ( f '( k ) n ^ 2 ) $ (\ref{thm:simpleGeneral})'', where $ \ell $ is the treewidth of $ \polygon $, and $ f $ and $ f ' $ are two functions that do not depend on $ n $.
    Every decision algorithm mentioned also provides a solution if it exists.
    }
    \label{tab:contributions}
\end{table}

\section{Preliminaries} \label{sec:preliminaries}

Throughout the article, the polygons we consider are all simple (i.e., without self-intersections), rectilinear (i.e., the edges are axis-aligned) and may have holes (unless mentioned otherwise).
Although we do not make it a general assumption, we say that a polygon $ \polygon $ is in \emph{general position} if no three vertices of $ \polygon $ lie on one axis-aligned line (Figure~\ref{fig:partitionWithSteinerPointsVSConforming}(a-c) and Figure~\ref{fig:preliminaries}(a)).
Moreover, we use the terminology of the Euclidean topology of the plane: a set of points in the plane may be \emph{open}, \emph{closed}, has an \emph{interior} and a \emph{boundary}. A standard exception is when a line segment $ \segment $ is considered: although $ \segment $ is never open in the Euclidean topology of the plane, we say that $ \segment $ is \emph{open} if $ \segment $ does not contain its endpoints (the underlying topology is then the Euclidean topology of a line containing $ \segment $).

Next, we recall the definitions leading to the notion of \emph{stabbing number of a polygon} (already in Section~\ref{sec:intro}):
Given a polygon $ \polygon $, a finite partition of $ \polygon $ is \emph{rectangular} if all the elements of the partition are rectangles with non-empty interior (an unimportant detail: an edge or a vertex shared by more than one partition rectangle is assigned to any one of them).
Alternatively, a rectangular partition is obtained by cutting $ \polygon $ along disjoint axis-aligned open segments (i.e., segments excluding endpoints) lying in the interior of $ \polygon $ that are of maximal length.
We say that such a set of segments is a \emph{cut set}, and freely consider a rectangular partition as a cut set.
Given a polygon $ \polygon $, a \emph{stabbing segment} of $ \polygon $ is an axis-aligned line segment that lies in the interior of $ \polygon $.
Given a polygon $ \polygon $ and a rectangular partition $ \partition $ of $ \polygon $, the \emph{stabbing number of $ \partition $} is the maximum, over all stabbing segments $ \segment $ of $ \polygon $, of the number of partition rectangles of $ \partition $ intersected by $ \segment $.
A rectangular partition of $ \polygon $ with minimum stabbing number $ k $ is \emph{optimal}; $ k $ is then the \emph{stabbing number of $ \polygon $}.
Figure~\ref{fig:partitionWithSteinerPointsVSConforming}(a) shows an example of an optimal partition; note the use of \emph{Steiner points}, i.e., corners of partition rectangles that lie in the interior of $ \polygon $.
In fact, it is known (Lemma~1 in~\cite{AADK11}) that to minimize the stabbing number of a rectangular partition (viewed as a cut set), it suffices to restrict to \emph{anchored} segments, i.e., segments that have at least one endpoint that is a \emph{reflex} vertex of $ \polygon $ (i.e., a vertex of $ \polygon $ with a reflex interior angle).

We now recall the definitions leading to the notion of \emph{conforming} stabbing number of a polygon (also in Section~\ref{sec:intro}):
A \emph{conforming partition} of a polygon $ \polygon $ is a rectangular partition without Steiner points. 
To minimize the stabbing number in this case, it suffices to restrict the (cut sets of) conforming partitions to segments with \emph{both} endpoints on the boundary of $ \polygon $, at least one of which being a reflex vertex of $ \polygon $.
We say that such segments are \emph{reflex} (Figure~\ref{fig:partitionWithSteinerPointsVSConforming}(b)).
Again, we say that a conforming partition is \emph{optimal} if its stabbing number is minimum among all the conforming partitions, and we call this stabbing number the \emph{conforming stabbing number of $ \polygon $}.
Note that, if no pair of reflex segments of $ \polygon $ intersect, then any optimal partition of $ \polygon $ is conforming.
We say that such a polygon is \emph{thin} (Figure~\ref{fig:preliminaries}(b-c)).

Next, we define the treewidth of a polygon $ \polygon $.
We first recall the well-known definition of the treewidth of a graph and the related concepts.
A \emph{tree decomposition} of a graph $ \graph = ( V , E ) $ is a tree $ \treeDecomposition $ whose nodes (called \emph{bags}) are subsets of $ V $ with the following properties:
\begin{inlineEnum}
    \item\label{connectedSubtree} For every vertex $ v $ of $ \graph $, the bags of $ \treeDecomposition $ containing $ v $ form a non-empty connected subtree of $ \treeDecomposition $.
    \item\label{edgeInBag} For every edge $ e $ of $ \graph $, there exists a bag of $ \treeDecomposition $ that contains both vertices of $ e $.
\end{inlineEnum}
The \emph{width} of a tree decomposition is the maximum bag size minus one, and the \emph{treewidth} $ \tw ( \graph ) $ of $ \graph $ is the minimum width of a tree decomposition of $ \graph $.

\begin{figure}[ht]
    \hspace*{\stretch{1}}
    \subcaptionbox{}{\includegraphics[scale=\graphicsScale,page=4]{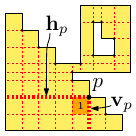}}
    \hspace*{\stretch{2}}
    \subcaptionbox{}{\includegraphics[scale=\graphicsScale,page=2]{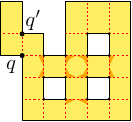}}
    \hspace*{\stretch{2}}
    \subcaptionbox{}{\includegraphics[scale=\graphicsScale,page=3]{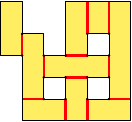}}
    \hspace*{\stretch{1}}
    \caption{
    (a) The reflex segments of $ \polygon _ 1 $ from Figure~\ref{fig:partitionWithSteinerPointsVSConforming} are dotted (red) and the reflex vertices are tiny (black) discs. The horizontal and vertical reflex segments $ \horOf{ p } $ and $ \verOf{ p } $ from the reflex vertex $ p $ are bold. The wedge-pixel of $ p $ is shaded (in orange) and labeled~$ 1 $.
    (b) The reflex segments of a thin polygon $ \polygon _ 2 $ (not in general position) with three holes and seven gates. The gates are the reflex segments with two small orange triangles on their side. (Note that the reflex segment $ \sgt{ q }{ q ' } $ is not a gate.)
    (c) An optimal conforming partition of $ \polygon _ 2 $ with stabbing number~$ 3 $.
    }
    \label{fig:preliminaries}
\end{figure}

In our case, we consider the treewidth of the \emph{pixel graph} $ \pixelGraph{ \polygon } $ of $ \polygon $, defined as the undirected simple graph with the following vertex set and edge set (and illustrated in Figure~\ref{fig:partitionWithSteinerPointsVSConforming}(c)). To describe these two sets, consider the union $ \EdgesAndRefSeg{ \polygon } $ of the edges of $ \polygon $ and the reflex segments of $ \polygon $.
\begin{itemize}
    \item The vertex set of $ \pixelGraph{ \polygon } $, denoted by $ \vertexSetOf{ \pixelGraph{ \polygon } } $, consists of all the endpoints of the segments in $ \EdgesAndRefSeg{ \polygon } $ together with all the intersections between any pair of segments in $ \EdgesAndRefSeg{ \polygon } $.
    \item The edge set of $ \pixelGraph{ \polygon } $, denoted by $ \edgeSetOf{ \pixelGraph{ \polygon } } $, consists of all the pairs of vertices $ v _ 1 , v _ 2 \in \vertexSetOf{ \pixelGraph{ \polygon } } $ such that the open segment $ \sgt{ v _ 1 }{ v _ 2 } $ is contained in some segment of $ \EdgesAndRefSeg{ \polygon } $ and does not contain any vertex of $ \vertexSetOf{ \pixelGraph{ \polygon } } $.
\end{itemize}
We then define the \emph{treewidth} of $ \polygon $ as the treewidth $ \tw ( \pixelGraph{ \polygon } ) $ of the pixel graph.

We also define the \emph{pixels} of a polygon $ \polygon $ as the inner faces of the pixel graph $ \pixelGraph{ \polygon } $ (which is planar by definition), and their set is denoted $ \pixelSetOf{ \pixelGraph{ \polygon } } $.
For each reflex vertex $ p $ of a polygon $ \polygon $, let $ \horOf{ p } $ and $ \verOf{ p } $ denote the horizontal and vertical reflex segments incident to $ p $, and let the \emph{wedge-pixel of $ p $} be the pixel of $ \polygon $ that is incident to $ \horOf{ p } $ and $ \verOf{ p } $ (Figure~\ref{fig:preliminaries}(a)).
A \emph{gate} of $ \polygon $ is a reflex segment $ \sgt{ p }{ q } $ that connects two reflex vertices $ p , q $ of $ \polygon $ such that the wedge-pixels of $ p $ and $ q $ lie on the same side of $ \sgt{ p }{ q } $.
Figure~\ref{fig:preliminaries}(b) shows all the gates of a polygon; note that segment $ \sgt{ q }{ q ' } $ is not a gate.
Note that if $ \polygon $ is in general position, then $ \polygon $ has no gates.

Recall that a \emph{stabbing segment} of a rectilinear polygon $ \polygon $ is an axis-aligned line segment that lies in the interior of $ \polygon $; for purposes of the stabbing number we only need to consider segments of maximal length, and we consider them to be open segments.
We say that two stabbing segments are \emph{equivalent} if they intersect the same set of pixels; there are $ \OO ( n ) $ equivalence classes of stabbing segments.
For instance, in Figure~\ref{fig:preliminaries}(a), there are~$ 52 $ equivalence classes: One for each of the~$ 26 $ reflex segments, one for each of the~$ 13 $ maximal open regions that do not contain a horizontal reflex segment, and one for each of the~$ 13 $ maximal open regions that do not contain a vertical reflex segment.

Finally, we denote by $ \STAB{ k } $ (respectively $ \CSTAB{ k } $) the set of rectilinear polygons $ \polygon $ (possibly with holes) that admit a partition (respectively conforming partition) into rectangles such that all stabbing segments of $ \polygon $ intersect at most $ k $ of these rectangles.
As is usual, we use the same notation $ \STAB{ k } $ (respectively $ \CSTAB{ k } $) for the associated problem of deciding whether a polygon is in $ \STAB{ k } $ (respectively $ \CSTAB{ k } $).
Moreover, we say that a conforming partition $ \partition $ of $ \polygon $ (seen as a set of reflex segment) is \emph{minimal} if it is minimal for inclusion, that is, if any set obtained by removing a reflex segment from $ \partition $ is not a conforming partition anymore.
This definition leads to the following equivalence: $ \polygon \in \CSTAB{ k } $ if and only if $ \polygon $ admits a minimal conforming partition with stabbing number at most $ k $.

\section{Intractability of Stabbing Number \texorpdfstring{$ 4 $}{4} or More} \label{sec:4+intratable}

In this section, we prove the following $ \NP $-completeness result.

\begin{theorem} \label{thm:NPHardStabbingN4+}
    For all integers $ k \geq 4 $, the decision problems $ \STAB{ k } $ and $ \CSTAB{ k } $ are $ \NP $-complete. 
    Moreover, $ \STAB{ k } $ and $ \CSTAB{ k } $ remain $ \NP $-complete even for thin polygons (possibly with holes).
    Furthermore, $ \CSTAB{ k } $ remains $ \NP $-complete even for rectilinear polygons (possibly with holes) in general position.
\end{theorem}

Before proving this result, we point out that Theorem~\ref{thm:NPHardStabbingN4+} does not state anything about rectilinear polygons that are both thin and in general position.
Indeed, for a rectilinear polygon that is both thin and in general position, $ \STAB{ k } $ and $ \CSTAB{ k } $ are solvable in polynomial time;
see Lemma~\ref{lem:thinGeneral}. 
Further note that Theorem~\ref{thm:NPHardStabbingN4+} does not state that $ \STAB{ k } $ is $ \NP $-complete for polygons in general position as this remains an open problem.

\paragraph{Proof structure\ARXIVVERSIONONLY{.}}
Thanks to Lemma~1 in~\cite{AADK11}, we know that, to minimize the stabbing number of a rectangular partition, it suffices to restrict to anchored segments.
In other words, it suffices to restrict to the $ \OO ( n ^ 2 ) $ vertices of the pixel graph for the Steiner points.
It is then straightforward to encode a polynomial-length certificate for $ \STAB{ k } $ and $ \CSTAB{ k } $, and verify its validity in polynomial time. This settles that $ \STAB{ k } $ and $ \CSTAB{ k } $ are in $ \NP $ for any integer $ k $.
We thus concentrate on proving $ \NP $-hardness.

First, we prove that $ \CSTAB{ 4 } $ is $ \NP $-hard in thin polygons.
In a thin polygon, any optimal partition is conforming.
Consequently, $ \STAB{ 4 } $ is also $ \NP $-hard. 
However, the gadgets take advantage of not being in general position.

Second, we provide an alternative version of this proof, this time for polygons in general position (but the gadgets take advantage of not being thin).

Third, using a similar approach and with a similar alternative version for polygons in general position, we prove that $ \CSTAB{ 5 } $, and thus $ \STAB{ 5 } $, is $ \NP $-hard.

Finally, we show how to modify our constructions for $ \CSTAB{ 4 } $ and $ \CSTAB{ 5 } $ to work for $ \CSTAB{ ( 4 + 2 m ) } $ and $ \CSTAB{ ( 5 + 2 m ) } $ for any $ m \geq 1 $.
Therefore, $ \CSTAB{ k } $ is $ \NP $-hard for thin polygons for all $ k \geq 4 $, which implies hardness for $ \STAB{ k } $.
Similarly, we show that, for all $ k \geq 4 $, $ \CSTAB{ k } $ is $ \NP $-hard for polygons in general position.

\paragraph{Problem used for the reductions: $ \RPM $\ARXIVVERSIONONLY{.}}
We use an $ \NP $-hard problem called \emph{rectilinear planar monotone $ \SAT{ 3 } $} ($ \RPM $)~\cite{BK12} to prove the hardness results by reduction.
The $ \RPM $ problem is a version of $ \SAT{ 3 } $ where every clause is either negative or positive, i.e., consists of either three positive literals or three negative literals.
Furthermore, the 
bipartite graph $ \graph _ \formula $ (defined by creating a vertex for every variable and clause of a given a $ 3 $-CNF formula $ \formula $ and adding an edge $ ( \clause , \variable ) $ whenever $ \variable $ is a variable of a clause $ \clause $)
has a drawing where the following holds (Figure~\ref{fig:RPMStabbingN4}).

\begin{figure}[hpt]
    \centering
    \includegraphics[scale=\graphicsScale,page=1]{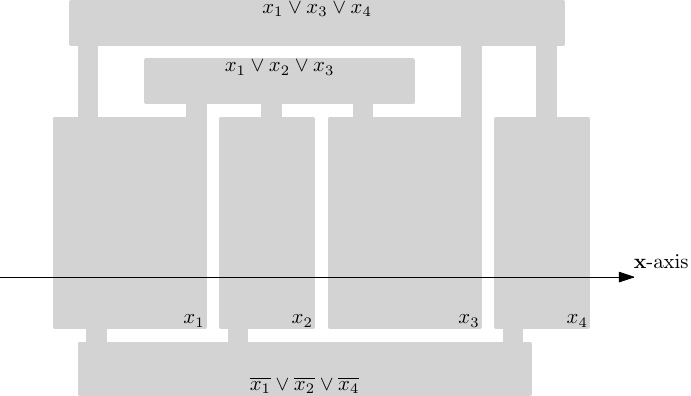}
    \caption{An $ \RPM $ drawing of $ \formula = ( \variable _ 1 \lor \variable _ 3 \lor \variable _ 4 ) \land  ( \variable _ 1 \lor \variable _ 2 \lor \variable _ 3 ) \land  ( \overline{ \variable _ 1 } \lor \overline{ \variable _ 2 } \lor \overline{ \variable _ 4 } ) $.}
    \label{fig:RPMStabbingN4}
\end{figure}

\begin{itemize}
    \item Each variable is represented by a \emph{variable rectangle},
    an axis-aligned rectangle that intersects the $ \xcoor $-axis. 
    \item Each clause is represented by a \emph{clause rectangle}, an
        axis-aligned rectangle that is above the $ \xcoor $-axis for a positive clause and below the $ \xcoor $-axis for a negative clause.
    \item The edges are drawn as \emph{channels}, i.e., as open axis-aligned rectangles whose top and bottom sides touch the corresponding variable and clause rectangles.
    \item All of the variable rectangles, clause rectangles, and edge channels are pairwise interior disjoint.
\end{itemize}

In the following, $ \formula $ is an arbitrary instance of $ \RPM $ and implicitly comes with an $ \RPM $ drawing.

\paragraph{Proof of the $ \NP $-hardness of $ \CSTAB{ 4 } $ for thin polygons\ARXIVVERSIONONLY{.}}
We reduce $ \RPM $ to $ \STAB{ 4 } $ in polynomial time.
We illustrate the reduction on the instance $ \formula $ of $ \RPM $ in Figure~\ref{fig:RPMStabbingN4}, see the resulting polygon $ \polygon ( \formula ) $ in Figure~\ref{fig:fullGadgetStabbingN4}(a).

\subparagraph{Gadgets.}
Our reduction uses multiple gadgets, that is, polygon pieces from which we build the polygon $ \polygon ( \formula ) $.
Specifically we use \emph{variable gadgets}, \emph{split gadgets} and \emph{clause gadgets}, and the former two in turn use \emph{forcer gadgets}.
In this section, we define these gadgets mostly via figure; precise descriptions with coordinates are given in the appendix.
We also give a local description (gadget per gadget) of key properties
of any conforming partition $ \partition $ of $ \polygon ( \formula ) $ that has stabbing number at most~$ 4 $ and that is \emph{minimal} (i.e., no reflex segment can be removed while retaining a conforming partition).
Besides, we exhibit key instances of such a partition in the figures.
Again, we refer to the appendix for the formal statements and their proofs.

\begin{figure}[hpt]
    \centering
    \includegraphics[scale=\graphicsScale,page=6]{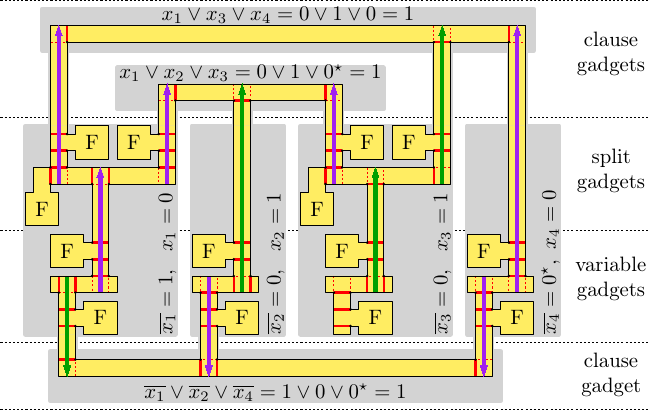}
    \caption{
    The polygon $ \polygon ( \formula ) $ (not to scale) of the $ \RPM $ drawing of $ \formula = ( \variable _ 1 \lor \variable _ 3 \lor \variable _ 4 ) \land  ( \variable _ 1 \lor \variable _ 2 \lor \variable _ 3 ) \land  ( \overline{ \variable _ 1 } \lor \overline{ \variable _ 2 } \lor \overline{ \variable _ 4 } ) $. Forcer gadgets are represented by squares labeled F. 
    The reflex segments of a partition with stabbing number~$ 4 $ are solid bold (in red), the other reflex segments are dotted (in red).
    Vertical stabbing segments are bold, and propagate~$ 0 $ if they are green (lightly shaded) and~$ 0 $ if they are purple (darkly shaded).
    We use $ 0 ^ \star $ for a value that is~$ 1 $ in the variable assignment but that has been decreased by a variable gadget or by a split gadget and is propagated as~$ 0 $.
    }
    \label{fig:fullGadgetStabbingN4}
\end{figure}

\begin{itemize}
\item
A {\bf forcer gadget} is shown in Figure~\ref{fig:forcerGadgetOverview}(a).
This is a thin polygon similar to a $ 4 { \times } 4 $ grid, with an extension to attach it to the rest of the polygon through a pixel called \emph{connection pixel}.
Let the \emph{force-stab} be the vertical stabbing segment through the connection pixel.
The grid-like structure and the seven pixels in each row and column imply the following:
\begin{quotation}
In any partition $ \partition $, the force-stab intersects three reflex segments within the $ 4 { \times } 4 $ grid.
Besides, such a partition exists.
\end{quotation}
See  Appendix~\ref{sec:forcer-gadget}  and Lemma~\ref{lem:forcerGadgetStabbingN4}  for details. 
In consequence, the force-stab is not allowed to intersect reflex segments outside the forcer gadget, and in particular, enforces the presence  of the indicated parallel pair of reflex segments in the pixel to which it is attached. 
In our later construction, we indicate a forcer gadget by a square labeled F; see Figure~\ref{fig:forcerGadgetOverview}(c).

\begin{figure}
    \hspace*{\stretch{1}}
    \subcaptionbox{}{\includegraphics[scale=\graphicsScale,page=14]{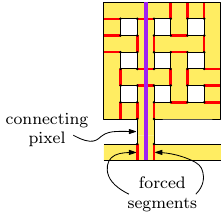}}
    \hspace*{\stretch{2}}
    \subcaptionbox{}{\includegraphics[scale=\graphicsScale,page=5]{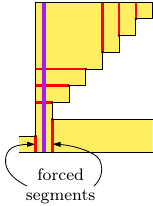}}
    \hspace*{\stretch{2}}
    \subcaptionbox{}{\includegraphics[scale=\graphicsScale,page=13]{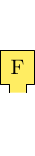}}
    \hspace*{\stretch{1}}
    \caption{Forcer gadgets (with their force-stab in bold purple, i.e., darkly shaded), for (a) thin polygons and (b) polygons in general position. (c) A schematic drawing of the forcer gadget.}
    \label{fig:forcerGadgetOverview}
\end{figure}

\item
Before describing the other gadgets, we first need to discuss how
information is propagated between them.
This happens via vertical stabbing segments that intersects two gadgets.
(In Figure~\ref{fig:fullGadgetStabbingN4}, such segments are bold.)
We consider such a stabbing segment $\segment$ as directed away from the $ \xcoor $-axis and call it an \emph{out-stab} of the gadget containing its tail, and an \emph{in-stab} of the gadget containing its head.
When considering a partition $ \partition $, we say that $ \segment $ \emph{propagates}~$ 0 $ (standing for ``false'') if $ \segment $ intersects three segments of $ \partition $ within the gadget $ \gadget $ containing its tail, and that $ \segment $ propagates~$ 1 $ (standing for ``true'') if $ \segment $ intersects at most two segments of $ \partition $ within $ \gadget $.

\item The {\bf variable gadget} for a variable $ \variable $ is placed inside the variable rectangle of $ \variable $.
It consists of a height-$ 1 $ axis-aligned rectangle that straddles the $ \xcoor $-axis, with two width-$ 1 $ axis-aligned rectangles attached above and below at different $ \xcoor $-coordinates. 
See Figure~\ref{fig:otherGadgetsOverviewN4}(a).
The vertical stabbing segments through these latter rectangles define the two out-stabs of the variable gadgets; the upward one is assigned to the positive literal $ \variable $ and the downward one to the negative literal $ \overline{ \variable } $. 
We attach a forcer gadget at each width-$ 1 $ rectangle such that, in $ \partition $, the corresponding out-stab intersects two reflex segments within the variable gadget. 
As for determining whether the remaining reflex segments of the variable gadget belong to $ \partition $, there are exactly three possibilities summarized as follows:
\begin{quotation}
    In any partition $ \partition $, not both out-stabs propagates~$ 1 $, but all other combinations of propagated values exist.
\end{quotation}
See  Appendix~\ref{sec:variable-gadget}  and Lemma~\ref{lem:variableGadgetStabbingN4}  for details. 
For each variable $ \variable $, we can therefore read
from the partition $ \partition $ of its variable gadget a value for $ \variable $, namely, $ \variable $ is assigned the value propagated by the out-stab of the literal $ \variable $.
Note that we set $ \variable = 0 $ (by convention) if both out-stabs propagate~$ 0 $ (the convention $ \variable = 1 $ would have worked as well).
In the partition of $ \polygon ( \formula ) $ in Figure~\ref{fig:fullGadgetStabbingN4}(a), $ \variable _ 1 , \variable _ 2 , \variable _ 3 , \variable _ 4 = 0 , 1 , 1 , 0 $ (even though the out-stab of the literal $ \overline{ \variable _ 4 } $ propagates~$ 0 $).

\begin{figure}[ht]
    \hspace*{\stretch{1}}
    \subcaptionbox{}{\includegraphics[scale=\graphicsScale,page=3]{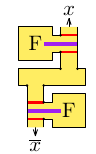}}
    \hspace*{\stretch{2}}
    \subcaptionbox{}{\includegraphics[scale=\graphicsScale,page=3]{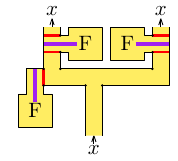}}
    \hspace*{\stretch{2}}
    \subcaptionbox{}{\includegraphics[scale=\graphicsScale,page=11]{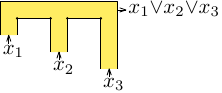}}
    \hspace*{\stretch{1}}
    \caption{(a) Variable gadget, (b) split gadget, and (c) clause gadget for thin polygons and $\CSTAB{4}$.}
    \label{fig:otherGadgetsOverviewN4}
\end{figure}

\item
The {\bf split gadgets} for a variable $ \variable $ are also drawn inside the variable rectangle of $ \variable $.
In Figure~\ref{fig:fullGadgetStabbingN4}, we only need split gadgets for the positive literals $ \variable _ 1 $ and $ \variable _ 3 $; these split gadgets are placed above the variable gadgets.
In general there could also be split gadgets for negative literals (which are placed below the variable gadgets), and there could be nested split gadgets in case a literal is used in more than two clause.
A \emph{positive split gadget} consists of a height-$ 1 $ axis-aligned rectangle, with one width-$ 1 $ axis-aligned rectangle attached below and two width-$ 1 $ axis-aligned rectangles attached above, see also Figure~\ref{fig:otherGadgetsOverviewN4}(b).
A \emph{negative split gadget} is horizontally symmetric.
The vertical stabbing segments through the width-$ 1 $ rectangles define one in-stab and two outer-stabs of the split gadget.
We attach forcer gadgets to these rectangles that force both out-stabs to intersect two reflex segments, and the stabbing segment of the height-$ 1 $ rectangle to intersect one reflex segment.
A case analysis shows that the following holds:
\begin{quotation}
    In any partition $ \partition $, the value propagated by each out-stab is at most the value propagated by the in-stab.
    Besides, such a partition exists for all such combinations of propagated values.
\end{quotation}
See Appendix~\ref{sec:split-gadget} and Lemma~\ref{lem:splitGadgetStabbingN4}  for details.    
In Figure~\ref{fig:fullGadgetStabbingN4}, the split gadget of $ \variable _ 1 $ splits the in-stab's value~$ 0 $ into two out-stabs propagating~$ 0 $ as well, whereas the split gadget of $ \variable _ 3 $ splits the in-stab's value~$ 1 $ into a left out-stab propagating~$ 0 $ and a right one propagating~$ 1 $. 

\item 
A {\bf clause gadget} of clause $ \clause $ consists of a unit-height axis-aligned rectangle that is placed inside clause rectangle of $ \clause $ and spans its width.  We attach unit-width axis-aligned rectangles to connect the unit-height rectangle to the edge channels that lead to its three literals of $ \clause $; these define three in-stabs of the gadget.
It is easy to verify the following:
\begin{quotation}
Partition $ \partition $ exists only if and only if at least one of the three in-stabs propagates~$ 1 $.
\end{quotation}
See  Appendix~\ref{sec:clause-gadget} and Lemma~\ref{lem:clauseGadgetStabbingN4}  for details.    

\end{itemize}

\subparagraph{Combining gadgets.} 
We build the polygon $ \polygon ( \formula ) $ by combining the gadgets as follows. 
For each variable $ \variable $, if the positive literal of $ \variable $ appears $ t $ times in a clause, then we add $ t - 1 $ split gadgets above the variable gadget of $ \variable $. 
We translate these split gadgets such that the right out-stab of one gadget (variable or split) exactly becomes the in-stab of the next gadget. 
This leaves $ t $ out-stabs that are not yet used as in-stabs. 
We call these the \emph{out-stabs of literal $ \variable $} and extend the corresponding width-$ 1 $ rectangles upward until they reach the top side of the variable rectangle of $ \variable $.  
We assume that the variable rectangle of $ \variable $ has been increased suitably in width and height (by inserting empty rows/columns into the drawing) so that all these split gadgets fit within it. 
We also assume, after stretching the drawing within the variable rectangle horizontally as needed, that the out-stabs of the literals exactly hit the boundaries of the variable rectangles at the place where the edge channels attach. 
We include the edge channels in $ \polygon ( \formula ) $, and these in turn attach to the clause gadgets of each clause.
As such, the three in-stabs of a clause gadget of clause $ \clause $ now correspond to three out-stabs from the three literals in $ \clause $.

\subparagraph{Proof of equivalence.}
We prove that there exists a conforming partition $ \partition $ of $ \polygon ( \formula ) $ with stabbing number at most~$ 4 $ if and only if there exists a satisfying assignment for $ \formula $.
Without loss of generality, we assume $ \partition $ to be minimal.

For the first implication, we are given the partition $ \partition $.
We have already described how to read from $ \partition $ an assignment for each variable of $ \formula $ in the description of the variable gadget: use the value propagated by the upward out-stab.
We now prove that this assignment satisfies $ \formula $, by proving that each clause $ \clause $ is satisfied.
Because the stabbing number of $ \partition $ is at most~$ 4 $, the clause gadget of $ \clause $ has at least one in-stab $ \segment $ propagating~$ 1 $.
This in-stab $ \segment $ propagates~$ 1 $ as a result of a chained propagation, possibly through a number of split gadgets, and that originates from one side of a variable gadget associated with some literal $ \literal \in \{ \variable , \overline{ \variable } \} $.
Recall that, for split gadgets, the propagated value never increases.
Consequently, the out-stab of $ \literal $ also propagates $ 1 $.
Because the values propagated from a variable gadget never increase compared to the value of the corresponding literals in the assignment, $ \literal $ has value $ 1 $ in the assignment.
Therefore, the clause $ \clause $ is satisfied in the assignment.
(Lemma~\ref{lem:variableGadgetStabbingN4}, Lemma~\ref{lem:splitGadgetStabbingN4}, and Lemma~\ref{lem:clauseGadgetStabbingN4} in the appendix show the non-increasing property for the variable gadget and the split gadget, and the property that a clause gadget admitting a partition with stabbing number at most~$ 4 $ has at least one in-stab $ \segment $ propagating~$ 1 $.)

For the converse implication, we are given an assignment satisfying $ \formula $.
From this assignment, we construct a conforming partition $ \partition $ of $ \polygon ( \formula ) $ with stabbing number at most~$ 4 $ as follows.
We partition each variable gadget such that the values propagated by its out-stabs equal the values of the corresponding literals in the assignment.
We partition each split gadget such that its out-stabs propagate the same value as its in-stab.
By construction, every propagated value throughout every gadget equals the value of the corresponding literal in the assignment.
Since the assignment satisfies $ \formula $, at least one in-stab per clause gadget propagates~$ 1 $.
It is then possible to partition each clause gadget with stabbing number at most~$ 4 $.
(Lemma~\ref{lem:variableGadgetStabbingN4}, Lemma~\ref{lem:splitGadgetStabbingN4}, and Lemma~\ref{lem:clauseGadgetStabbingN4} in the appendix show that all these partitions of gadgets exist.)

\paragraph{Proof of the $ \NP $-hardness of $ \CSTAB{ 4 } $ for polygons in general position\ARXIVVERSIONONLY{.}}
In the previous reduction, we use aligned reflex vertices most notably in the forcer gadget.
To achieve the reduction for polygons in general position, we design a completely different forcer gadget based on a staircase with~$ 6 $ reflex vertices; see Figure~\ref{fig:forcerGadgetOverview}(b) for a picture and Section~\ref{sec:4intractableGeneralPosition} for details.
For all other gadgets, having aligned reflex vertices turns out to be not important if we need not have a thin polygon.   Specifically, we show that by shifting all some reflex vertices in these gadgets suitably, we obtain gadgets in general position which have the same properties with respect to the values of their in-stabs and out-stabs.
Figure~\ref{fig:GadgetsGeneralN4} shows the modified variable gadget, split gadget and clause gadget, and Section~\ref{sec:4intractableGeneralPosition} describes their construction in detail.
The proof of NP-hardness then carries over verbatim.

\begin{figure}[ht]
    \hspace*{\stretch{1}}
    \subcaptionbox{}{\includegraphics[scale=\graphicsScale,page=1]{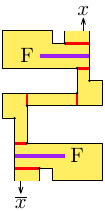}}
    \hspace*{\stretch{2}}
    \subcaptionbox{}{\includegraphics[scale=\graphicsScale,page=1]{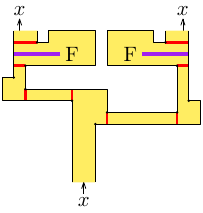}}
    \hspace*{\stretch{2}}
    \subcaptionbox{}{\includegraphics[scale=\graphicsScale,page=1]{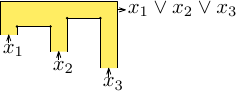}}
    \hspace*{\stretch{1}}
    \caption{The modified (a) variable gadget, (b) split gadget, and (c) clause gadget for polygons in general position and $\CSTAB{4}$.}
    \label{fig:GadgetsGeneralN4}
\end{figure}

\paragraph{Proof of the $ \NP $-hardness of $ \CSTAB{ 5 } $ for thin polygons and for polygons in general position\ARXIVVERSIONONLY{.}}
We prove that $ \CSTAB{ 5 } $ is $ \NP $-hard using the exact same reduction idea with slightly modified gadgets, be it for thin polygons or for polygons in general position.

Informally, the gadgets for $ \CSTAB{ 5 } $ for thin polygons are obtained from the gadgets for $ \CSTAB{ 4 } $ for thin polygons by adding one pixel at both ends of each row or column of adjacent pixels (Figure~\ref{fig:fullGadgetStabbingN5_main}).
In particular, we use this technique to generate the forcer gadget for $ \CSTAB{ 5 } $ for thin polygons (Figure~\ref{fig:forcerGadgetStabbingN5+}(a)).

As for $ \CSTAB{ 5 } $ and polygons in general position, we generalize the staircase forcer gadget for $ \CSTAB{ 4 } $ with~$ 6 $ reflex vertices to a staircase with~$ 8 $ reflex vertices.
We actually do not use this forcer gadget as is but integrate its principle into the other gadgets, as shown in Figure~\ref{fig:generalGadgetsStabbingN5}.

Details can be found in Section~\ref{sec:5+intractable}.

\begin{figure}[hpt]
    \centering
    \includegraphics[scale=\graphicsScale,page=2]{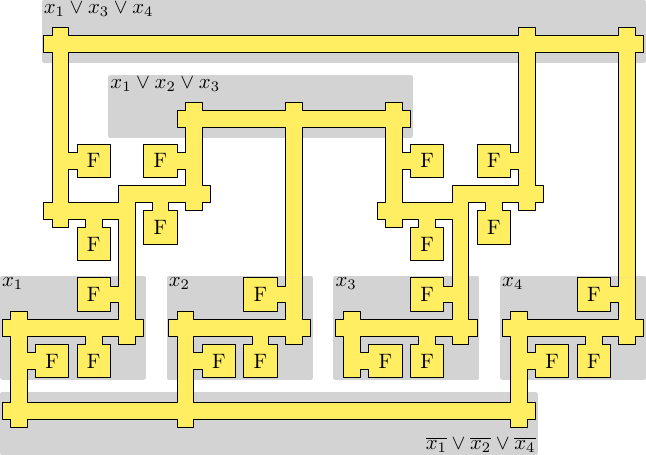}
    \caption{The polygon used to reduce $ \RPM $ to $ \STAB{ 5 } $. (The figure is not to scale, and the split gadgets are drawn outside the variable rectangles only to fit the figure in the page.)}
    \label{fig:fullGadgetStabbingN5_main}
\end{figure}

\begin{figure}[ht]
    \hspace*{\stretch{1}}
    \subcaptionbox{}{\includegraphics[scale=\graphicsScale,page=1]{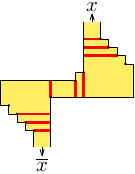}}
    \hspace*{\stretch{2}}
    \subcaptionbox{}{\includegraphics[scale=\graphicsScale,page=1]{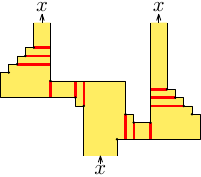}}
    \hspace*{\stretch{2}}
    \subcaptionbox{}{\includegraphics[scale=\graphicsScale,page=1]{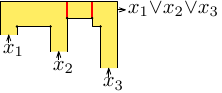}}
    \hspace*{\stretch{1}}
    \caption{In the context of $ \STAB{ 5 } $ for polygons in general position, (a) the variable gadget, (b) the split gadget, and (c) the clause gadget. The reflex segments that are forced by the gadgets themselves (i.e., the reflex segments that must be in any conforming partition with stabbing number at most~$ 5 $) are drawn in bold red. (To gain space, the gadgets are drawn with some edges on a common line which implies that they are not in general position. But since every reflex segment has exactly one endpoint which is a reflex vertex, it is straightforward to regain general position with small perturbations of the vertices of the polygons.)}
    \label{fig:generalGadgetsStabbingN5}
\end{figure}

\paragraph{Proof of the $ \NP $-hardness of $ \CSTAB{ k } $, $ k > 5 $, for thin polygons and for polygons in general position\ARXIVVERSIONONLY{.}}
The proof has two cases based on the parity of $ k $. In both cases, the two main ideas are the following.
First, the forcer gadgets for stabbing number $ 4 $ and $ 5 $ are straightforward to generalize for stabbing number $ k = 4 + 2 m $ and $ k = 5 + 2 m $, where $ m $ is a positive integer (Figure~\ref{fig:forcerGadgetStabbingN5+}).
Second, it is enough to attach $ m $ forcer gadgets for stabbing number $ k $ to the middle of each row or column of adjacent pixels of $ \polygon ( \formula ) $ (Figure~\ref{fig:fullGadgetStabbingN6}).
To avoid interference between gadgets, the new forcer gadgets are placed away from the other gadgets in such a way that they do not share a stabbing segment with other gadgets.
Doing so, the other gadgets for stabbing number $ 4 $ and $ 5 $ do not need to be modified to work for stabbing number $ k = 4 + 2 m $ and $ k = 5 + 2 m $.
See Section~\ref{sec:forcer-gadget} for details.

\begin{figure}[ht]
    \hspace*{\stretch{1}}
    \subcaptionbox{}{\includegraphics[scale=\graphicsScale,page=1]{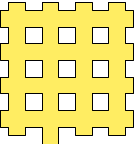}}
    \hspace*{\stretch{2}}
    \subcaptionbox{}{\includegraphics[scale=\graphicsScale,page=1]{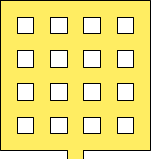}}
    \hspace*{\stretch{2}}
    \subcaptionbox{}{\includegraphics[scale=\graphicsScale,page=1]{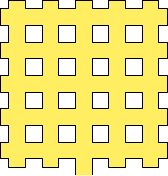}}
    \hspace*{\stretch{2}}
    \subcaptionbox{}{\includegraphics[scale=\graphicsScale,page=1]{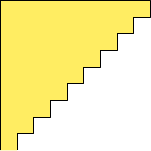}}
    \hspace*{\stretch{1}}
    \caption{(a-c) The thin forcer gadget in the context of $ \STAB{ k } $ for $ k = 5 , 6 , 7 $.
    (d) The forcer gadget in general position in the context of $ \STAB{ 5 } $.}
    \label{fig:forcerGadgetStabbingN5+}
\end{figure}

\begin{figure}[hpt]
    \centering
    \includegraphics[scale=\graphicsScale,page=7]{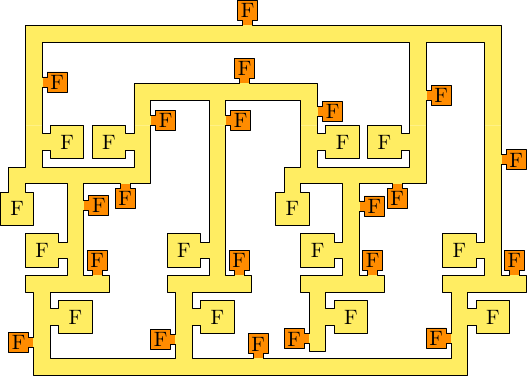}
    \caption{The polygon used to reduce $ \RPM $ to $ \STAB{ 6 } $ obtained from the polygon of Figure~\ref{fig:fullGadgetStabbingN4} by adding a forcer gadget to each row and each column of adjacent pixels. Additional forcer gadgets are orange (and smaller for space reasons).}
    \label{fig:fullGadgetStabbingN6}
\end{figure}

\section{Tractability of Conforming Stabbing Number \texorpdfstring{$ 2 $}{2}}
\label{sec:2tractable}

In this section, we consider $ \CSTAB{ 2 } $, the problem of deciding whether a rectilinear $ n $-gon $ \polygon $ (possibly with holes) has conforming stabbing number at most~$ 2 $.
We present two different proofs for the tractability of $ \CSTAB{ 2 } $.
The first proof, the proof of Lemma~\ref{lem:2SAT} hereafter, is an easy reduction of $ \CSTAB{ 2 } $ to $ \SAT{ 2 } $ and yields a quadratic running time decision algorithm.
The second proof, the proof of Theorem~\ref{thm:decideStabbing2} which spans over the entire section, provides a decision algorithm with an improved running time of $ \OO ( n \log n ) $.
Interestingly, in this second proof, we also prove that a special class of polygons, that is, the thin polygons in general position always have stabbing number at most~$ 2 $ (see Lemma~\ref{lem:thinGeneral}).

\paragraph{First proof\ARXIVVERSIONONLY{.}}
We start with a reduction from $\CSTAB{2}$ to $\SAT{2}$.

\begin{lemma}
    \label{lem:2SAT}
    Let $ \polygon $ be a rectilinear $ n $-gon (possibly with holes).
    Suppose that every stabbing segment intersects at most $ \ell $ reflex segments.
    There exists an algorithm that decides $ \CSTAB{ 2 } $ for $ \polygon $ and provides a solution (if any) in $ \OO ( \ell n ) $ time.
\end{lemma}

\begin{proof}
We reduce $ \CSTAB{ 2 } $ to $ \SAT{ 2 } $ as follows.
We declare a Boolean variable $ \variable ( \segment ) $ for every reflex segment $ \segment $. A reflex segment $ \segment $ is used in the solution if and only if $ \variable ( \segment ) $ is true. 
To ensure that we have encoded a conforming partition, we require
\begin{thmEnumerate}
    \item  $ \variable ( \horOf{ p } ) \lor \variable ( \verOf{ p } ) $ for every reflex vertex $ p $,  as well as
    \item $ \lnot \variable ( \horOf{ p } ) \lor \lnot \variable ( \verOf{ q } ) $ for any two intersecting reflex segments $ \horOf{ p } , \verOf{ q } $ (recall that reflex segments are open segments and thus do not intersect at their endpoints).
\end{thmEnumerate}
To ensure that the conforming stabbing number is at most~$ 2 $, we require that every stabbing segment intersects at most one chosen reflex segment. 
In other words, we additionally require 
\begin{thmEnumerate}\setcounter{thmEnumeratei}{2}
    \item $ \lnot \variable ( \segment _ 1 ) \lor \lnot \variable ( \segment _ 2 ) $ for any two reflex segments $ \segment _ 1 , \segment _ 2 $ intersected by a common stabbing segment.
\end{thmEnumerate}

Each of these restrictions only involve two variables, so this gives a $ \SAT{ 2 } $ instance that has $ \OO ( n ) $ variables.
For every reflex segment $ \segment $, variable $ \variable ( \segment ) $ belongs to at most two clauses of the first kind,
and at most $ \ell $ clauses each of the second and the third kind. 
Thus, the number of clauses is $ \OO ( \ell n ) $.
Since $ \SAT{ 2 } $ is decidable in linear time~\cite{aspvall1979linear}, the lemma follows.
\end{proof}

In an arbitrary polygon there could be stabbing segments that intersect $ \Theta ( n ) $ reflex segments, so the running time of the $ \SAT{ 2 } $ approach is $ \OO ( n ^ 2 ) $ in the worst case.
Our main contribution in this section is to give a faster algorithm, with $ \OO ( n \log n ) $ worst-case running time.

\paragraph{Second proof idea\ARXIVVERSIONONLY{.}}

We call a reflex segment of $ \polygon $ \emph{impossible} if no conforming partition of $ \polygon $ with stabbing number~$ 2 $ contains it, and \emph{fixed} if any such conforming partition must contain it.
The idea of the second proof of Theorem~\ref{thm:decideStabbing2} is an algorithm that determines via some rules that some segments are impossible and some are fixed, from which it deduces other segments to be impossible or fixed, and so on (Figure~\ref{fig:algo}).
We apply these rules recursively
\begin{itemize}
\item until we find a reflex segment that is both fixed and impossible, meaning that $ \polygon \not \in \CSTAB{ 2 } $, or 
\item until \emph{saturation} (i.e. until no more rule applies), in which case we show that $ \polygon \in \CSTAB{ 2 } $.
\end{itemize}
The next lemma lists these rules.

\begin{figure}[ht]
    \hspace*{\stretch{1}}
    \subcaptionbox{}{\includegraphics[scale=\graphicsScale,page=1]{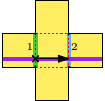}}
    \hspace*{\stretch{2}}
    \subcaptionbox{}{\includegraphics[scale=\graphicsScale,page=2]{algo.pdf}}
    \hspace*{\stretch{2}}
    \subcaptionbox{}{\includegraphics[scale=\graphicsScale,page=3]{algo.pdf}}
    \hspace*{\stretch{2}}
    \subcaptionbox{}{\includegraphics[scale=\graphicsScale,page=4]{algo.pdf}}
    \hspace*{\stretch{1}}
    \caption{
    (a) An example where the algorithm rejects. The gate at label~1 is marked as fixed (i.e., colored green) by~\ref{R4}. As a result, all segments intersecting a common stabbing segment with the segment at~1 are marked as impossible (i.e., colored blue) by~\ref{R2}, which is the case of the segment at~2. But then, this segment at~2 is also a gate and thus is fixed, at which point the algorithm rejects.
    (b) An example where the algorithm marks all the reflex segments and accepts. Both segments at~1 are marked as impossible by~\ref{R3}. This triggers to mark as fixed the two segments at~2 by~\ref{R1}. Consequently, the segment at~3 is marked as impossible by~\ref{R2}, which in turn triggers to mark the segment at~4 as fixed by~\ref{R1}. Furthermore, the gate segment at~5 is marked as fixed by~\ref{R4}. 
    (c) An example where none of the rules apply. The algorithm does not mark any reflex segment and accepts. (Note that $ \sgt{ p _ 1 }{ p _ 2 } $ and $ \sgt{ p _ 1 }{ p _ 3 } $ are not gates.) 
    (d) An example where the algorithm marks only some reflex segments and accepts. 
    } 
    \label{fig:algo}
\end{figure}

\paragraph{Rules\ARXIVVERSIONONLY{.}}

In the following lemma, the first two rules are obvious and stated without proof, whereas the last two rules, which initialize the entire process, are proved next.

\begin{lemma}[Rules]\label{lem:rules}
    Let $ \polygon $ be a rectilinear $ n $-gon (possibly with holes).
    The following holds.
    \begin{rules}
        \item\label{R1} If, at some reflex vertex $ p $ of $ \polygon $, one of $ \horOf{ p } , \verOf{ p } $ is impossible, then the other one is fixed (Figure~\ref{fig:Rules}(a)).
        \item\label{R2} If a stabbing segment $ \segment $ of $ \polygon $ intersects a fixed segment, then all other reflex segments of $ \polygon $ intersected by $ \segment $ are impossible (Figure~\ref{fig:Rules}(b)).
        \item\label{R3} If two reflex segments of $ \polygon $ (open by definition) intersect, then both are impossible (Figure~\ref{fig:Rules}(c)).
        \item\label{R4} Any gate of $ \polygon $ is fixed (Figure~\ref{fig:Rules}(d)).
    \end{rules}
\end{lemma}

\begin{proof}[Proof of~\ref{R3}]
Let $ \horOf{ p } $ and $ \verOf{ q } $ be a pair of horizontal and vertical reflex segments that intersect.
Assume for contradiction that some conforming partition $ \partition $ includes $ \horOf{ p } $ (the argument is similar for $ \verOf{ q } $). 
Then $ \partition $ does not include $ \verOf{ q } $ (segments of $ \partition $ must not intersect), so we must include $ \horOf{ q } $.
Let $ \chi $ be the common point of $ \horOf{ p } $ and $ \verOf{ q } $.
Up to symmetry, the wedge-pixel of $ q $ lies to the right of $ \verOf{ q } $. 
Then for a small enough $ \varepsilon $, the vertical stabbing segment through $ \chi + ( \varepsilon , \varepsilon ) $ intersects both $ \horOf{ p } $ and $ \horOf{ q }$ (Figure~\ref{fig:Rules}(c)) and the conforming partition has stabbing number~$ 3 $ or more.
\end{proof}

\begin{proof}[Proof of~\ref{R4}]
Let $ \sgt{ p }{ q } $ be a gate of $ \polygon $.
Without loss of generality, assume that the segment $ \sgt{ p }{ q } $ is horizontal and that both $ \verOf{ p } $ and $ \verOf{ q } $ lie above $ \sgt{ p }{ q } $.
The horizontal stabbing segment through $ p + ( 0 , \varepsilon ) $ (for a small enough $ \varepsilon $) then intersects both $ \verOf{ p } $ and $ \verOf{ q } $ (Figure~\ref{fig:Rules}(d)).
Thus, any conforming partition with stabbing number~$ 2 $ does not include both $ \verOf{ p } $ and $ \verOf{ q } $, which means by~\ref{R1} that the segment $ \sgt{ p }{ q } $ is included instead.
\end{proof}

\begin{figure}[ht]
    \hspace*{\stretch{1}}
    \subcaptionbox{}{\includegraphics[scale=\graphicsScale,page=1]{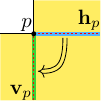}}
    \hspace*{\stretch{2}}
    \subcaptionbox{}{\includegraphics[scale=\graphicsScale,page=2]{Rules.pdf}}
    \hspace*{\stretch{2}}
    \subcaptionbox{}{\includegraphics[scale=\graphicsScale,page=3]{Rules.pdf}}
    \hspace*{\stretch{2}}
    \subcaptionbox{}{\includegraphics[scale=\graphicsScale,page=4]{Rules.pdf}}
    \hspace*{\stretch{1}}
    \caption{Illustrations of (a)~\ref{R1}, (b)~\ref{R2}, (c)~\ref{R3}, and (d)~\ref{R4}.} 
    \label{fig:Rules}
\end{figure}

\paragraph{Orthogonal ray-shooting\ARXIVVERSIONONLY{.}}
Before we describe the $ \OO ( n \log n ) $-time decision algorithm for $ \CSTAB{ 2 } $, note that we cannot afford to compute all intersections between reflex segments (or equivalently, the entire pixel graph), since there may be $ \Omega ( n ^ 2 ) $ such intersections 
    (as in Figure~\ref{fig:partitionWithSteinerPointsVSConforming} and Figure~\ref{fig:comb} for example).

Instead, it is possible to derive all required information via \emph{orthogonal ray-shooting} queries, which we define first. 
For such queries, consider a set $ \ors $ of disjoint parallel line segments stored in some data structure that we call \emph{orthogonal ray-shooting} data structure.
An orthogonal ray-shooting query receives as input an open ray that is perpendicular to the segments, and it reports the first segment in $ \ors $ that is hit by the ray, or that there is no such segment. 
(Here rays are considered open at their origin, i.e., we do not report a segment that only intersects the ray at its origin.) 
Giyora and Kaplan~\cite{GiyoraKaplan} gave an implementation that performs such a query in $ \OO ( \log | \ors | ) $ time and that also permits (within the same running time) to delete a segment of $ \ors $ from the data structure.
In the algorithm described next, we use orthogonal ray-shooting data structure as a black box.

\paragraph{Algorithm\ARXIVVERSIONONLY{.}}
We now describe the decision algorithm using a dynamic programming approach.
The idea is that the algorithm determines the status (i.e., fixed or impossible) for increasingly many segments, until it ends with an impossibility or a partition.
We express whether the status of a segment has been determined with a color scheme.
Initially all reflex segments of $ \polygon $ are uncolored.
During the execution we color some of the sections in green or in blue while maintaining the following invariants:
\begin{enumerate}
   \item \label{invariantGreen} The green reflex segments are fixed.
   \item \label{invariantBlue} The blue reflex segments are impossible.
\end{enumerate}
(Effectively, ``green'' means ``the segment is fixed, and the algorithm has determined that the segments is fixed'', and similarly for ``blue''.)
At any time during the execution, if the algorithm tries to color a blue segment in green (respectively a green segment in blue), then the execution is interrupted and the algorithm rejects.

\subparagraph{Input.}
The input is an axis-aligned polygon $ \polygon $ possibly with holes, stored as a list of directed cycles of vertices (i.e., the coordinates and some room for additional pointers and markers), where each directed cycle corresponds to a connected component of the boundary of $ \polygon $ with the convention that the interior of $ \polygon $ lies on the left (relatively to the local direction of the cycle).
In the following, we assume that every coordinate is a rational number (not a floating-point number), and that computations involving coordinates are exact. In particular, the coordinates of an intersection of two segments are computed exactly, and it is reliable to test whether this intersection is equal to a point.

\subparagraph{Preprocessing.}
The challenge of the algorithm is to avoid computing all the intersections between pairs of reflex segments and all the intersections between a stabbing segment (up to equivalence) and a reflex segment.
To achieve $ \OO ( n \log n ) $ time, we preprocess the input polygon $ \polygon $ to initialize and populate the following two pairs of orthogonal ray-shooting data structures (each pair having one data structure for horizontal segments and one for vertical segments).
\begin{description}
    \item[$ \Rhor , \Rver $:] Contains the reflex segments of $ \polygon $ which are either not colored (as is the case initially), or red (as we will remove a blue reflex segment upon coloring) to allow querying what is the first relevant reflex segment intersected by a stabbing segment.
    \item[$ \Shor , \Sver $:] Contains the stabbing segments of $ \polygon $ (more specifically, one representative per equivalent class of stabbing segment) which have not been processed (as we remove a stabbing segment after having colored all the intersected reflex segments) to allow querying what is the first relevant stabbing segment intersected by a reflex segment that we just colored red.
\end{description}
In the preprocessing, we also build a pair of lists $ \Lhor , \Lver $ of all horizontal and vertical reflex segments of $ \polygon $ that will be unaffected by removals performed on $ \Rhor , \Rver $.
In the following, we only detail how to compute the vertical reflex segments and vertical stabbing segments of $ \polygon $, and to populate $ \Lver , \Rver , \Sver $. The horizontal ones and $ \Lhor , \Rhor , \Shor $ are computed similarly.

We compute the trapezoidal map of the horizontal edges of $ \polygon $.
The set of vertical edges of the trapezoidal map that are not edges of $ \polygon $ is equal to the set of vertical reflex segments of $ \polygon $.
We use this fact to populate $ \Lver $ and $ \Rver $ with all the vertical reflex segments of $ \polygon $.
Moreover, by construction, there is a one-to-one correspondence between the trapezes (which are in fact rectangles that are columns of adjacent pixels) and the equivalent class of vertical stabbing segments.
In order to populate $ \Sver $ with one representative per equivalent class of vertical stabbing segment, we choose a point in the interior of each trapezes and compute the endpoints of the vertical stabbing segment containing this point.

To finalize the preprocessing, for each reflex segment $ \segment $ in $ \Lhor , \Lver $, we perform an orthogonal ray-shooting queries in $ \Rver , \Rhor $ with a ray containing $ \segment $ and with an endpoint of $ \segment $ as origin. If the query reports a reflex segment, i.e., a segment that is not an edge of $ \polygon $, then we mark $ \segment $ as \emph{intersecting some other reflex segment}.
We also determine whether $ \segment $ is a gate and mark $ \segment $ accordingly.

Next, we describe two mutually recursive routines, namely $ \ColorInBlue $ and $ \ColorInGreen $, assuming without loss of generality that the input reflex segment (an element of $ \Lhor \cup \Lver $) is of the form $ \horOf{ p } $ and $ \verOf{ p } $ respectively.

\subparagraph{Routine $  \ColorInBlue ( \horOf{ p } ) $:}
\begin{enumerate}
    \item If $ \horOf{ p } $ is green, then interrupt the algorithm and reject.
    \item If $ \horOf{ p } $ is blue, then do nothing.
    \item If $ \horOf{ p } $ is not colored yet, then do:
    \begin{enumerate}
        \item\label{clr:blue} Color $ \horOf{ p } $ blue.
        \item\label{rm:reflex} Remove $ \horOf{ p } $ from $ \Rhor $.
        \item Call $ \ColorInGreen ( \verOf{ p } ) $ (Figure~\ref{fig:Rules}(a)).
    \end{enumerate}
\end{enumerate}

\subparagraph{Routine $  \ColorInGreen ( \verOf{ p } ) $:}
\begin{enumerate}
    \item If $ \verOf{ p } $ is blue, then interrupt the algorithm and reject.
    \item If $ \verOf{ p } $ is green, then do nothing.
    \item If $ \verOf{ p } $ is not colored yet, then do:
    \begin{enumerate}
        \item\label{clr:green} Color $ \verOf{ p } $ green.
        \item\label{qry:stab} While $ \segment $, the first segment of $ \Shor $ encountered by the ray from $ p $ along $ \verOf{ p } $ (and computed via the corresponding orthogonal ray-shooting query), is not an edge of $ \polygon $, do:
        \begin{enumerate}
            \item\label{rm:stab} Remove $ \segment $ from $ \Shor $.
            \item Let $ \chi $ be the intersection of $ \segment $ and $ \verOf{ p } $.
            \item Let $ \ray _ 1 , \ray _ 2 $ be the rays from $ \chi $ along $ \segment $ on the left and on the right of $ \chi $ respectively (Figure~\ref{fig:Rules}(b)).
            \item For $ i = 1 $ to $ 2 $, do:
            \begin{enumerate}
                \item\label{qry:reflex} While $ \horOf{ q } $, the first segment of $ \Rhor $ encountered by the ray $ \ray _ i $ (and computed via the corresponding orthogonal ray-shooting query), is not an edge of $ \polygon $, do:
                \begin{enumerate}
                    \item Call $ \ColorInBlue ( \horOf{ q } ) $.
                \end{enumerate}
            \end{enumerate}
        \end{enumerate}
    \end{enumerate}
\end{enumerate}

We now describe the initialization of the recursive calls to the routines $ \ColorInBlue $ and $ \ColorInGreen $, which constitute the main algorithm.

\subparagraph{Main:}
\begin{enumerate}
    \item Do the preprocessing.
    \item For each segment $ \segment $ in $ \Lhor , \Lver $, do:
    \begin{enumerate}
        \item If $ \segment $ is marked as intersecting some other reflex segment, then call $ \ColorInBlue ( \segment ) $ (Figure~\ref{fig:Rules}(c)).
        \item If $ \segment $ is marked as a gate, then call $ \ColorInGreen ( \segment ) $ (Figure~\ref{fig:Rules}(d)).
    \end{enumerate}
    \item If all these calls end normally, that is to say, if none of them interrupted the algorithm to reject, then the algorithm accepts.
\end{enumerate}

\paragraph{Running time\ARXIVVERSIONONLY{.}}

\begin{lemma}\label{lem:runningTime}
    The algorithm runs in $ \OO ( n \log n ) $ time.
\end{lemma}

\begin{proof}
We first show with simple verifications that the preprocessing runs in $ \OO ( n \log n ) $ time.
Each step of the preprocessing takes $ \OO ( n \log n ) $ time.
Indeed, the two bottlenecks are the computation of the two trapezoidal maps in $ \OO ( n \log n ) $ time (see, e.g.,~\cite{Mark-book}), and the $ \OO ( n ) $ operations (either insertions or ray-shooting queries) on the orthogonal ray-shooting data structures in $ \OO ( \log n ) $ time each~\cite{GiyoraKaplan}.

We now prove that the rest of the algorithm also runs in $ \OO ( n \log n ) $ time.
Observe that, once the preprocessing is over, we perform only two types of operations that are not in constant time: orthogonal ray-shooting queries and removals from orthogonal ray-shooting data structures.
As we perform no insertion in orthogonal ray-shooting data structures, the number of removals in upper bounded by the maximum size of the orthogonal ray-shooting data structures which is $ \OO ( n ) $.
Besides, each of the $ \OO ( n ) $ reflex segments in $ \Lhor \cup \Lver $ is colored at most once.

We now map each orthogonal ray-shooting query either to a removal from an orthogonal ray-shooting data structure or to a coloring of some reflex segment.
Since each removal and each coloring will have at most a constant number of queries mapped to it, this will prove that the number of such queries is also $ \OO ( n ) $.

Every orthogonal ray-shooting query $ \query $ in $ \Shor , \Sver $ is performed at~\ref{qry:stab} in $ \ColorInGreen $.
If $ \query $ reports an edge of $ \polygon $, then we map $ \query $ to the coloring at~\ref{clr:green} in $ \ColorInGreen $.
Otherwise, $ \query $ reports a stabbing segment and we enter the while body.
We then map $ \query $ to the removal from $ \Shor , \Sver $ performed at~\ref{rm:stab} in $ \ColorInGreen $.

Similarly, every orthogonal ray-shooting query $ \query $ in $ \Rhor , \Rver $ is performed at~\ref{qry:reflex} in $ \ColorInGreen $.
If $ \query $ reports an edge of $ \polygon $, then we map $ \query $ to the coloring at~\ref{rm:stab} in $ \ColorInGreen $ (which is the second time we map a query to this coloring).
Otherwise, $ \query $ reports a reflex segment $ \segment ' $ and we enter the while body to call $ \ColorInBlue ( \segment ' ) $.
If $ \segment ' $ is already green, then we again map $ \query $ to the coloring at~\ref{rm:stab} in $ \ColorInGreen $ (which is still only the second time we map a query to this coloring since the algorithm is interrupted and the while loop ends prematurely without a query reporting an edge of $ \polygon $).
If $ \segment ' $ is not already green, then $ \segment ' $ is not already blue either, since every reflex segment that has been colored in blue (at~\ref{clr:blue} in $ \ColorInBlue $) has also been removed from $ \Rhor , \Rver $ (at~\ref{rm:reflex} in $ \ColorInBlue $).
Thus, $ \segment ' $ is colored in blue at~\ref{clr:blue} in $ \ColorInBlue $.
We then map $ \query $ to this coloring, which concludes the proof that the running time of the algorithm is indeed $ \OO ( n \log n )$.
\end{proof}

\paragraph{Correctness\ARXIVVERSIONONLY{.}}
In the rest of this section, we prove the correctness of the algorithm.

First observe that, by construction, $ \Shor \cup \Sver $ contains at least one stabbing segment for each equivalence class of stabbing segments in $ \polygon $.
A simple verification then shows that the algorithm implements the rules in Lemma~\ref{lem:rules} correctly.
Therefore, Invariants~\ref{invariantGreen} and~\ref{invariantBlue} hold.
Moreover, if the algorithm accepts, then the rules in Lemma~\ref{lem:rules} have indeed been applied until saturation as intended.

Next, we verify the correctness of the algorithm in three cases, of which only the third case is not straightforward.
If the algorithm rejects (Figure~\ref{fig:algo}(a)), then a reflex segment has been found to be both fixed and impossible.
Thus, $ \polygon \not \in \CSTAB{ 2 } $ and the algorithm is correct in this case.
If the algorithm accepts and if all reflex segments have been colored (Figure~\ref{fig:algo}(b)), then, for each reflex vertex $ p $ at least one of $ \horOf{ p } , \verOf{ p } $ is green, and each stabbing segment intersects at most one green segment.
This implies that the green reflex segments yield the only possible conforming partition of $ \polygon $ with stabbing number at most~$ 2 $ (by Invariants~\ref{invariantGreen} and~\ref{invariantBlue}).
In this case, $ \polygon \in \CSTAB{ 2 } $ and the algorithm is correct again.

In the remaining case, the algorithm accepts and some reflex segments have not been colored (while the rules in Lemma~\ref{lem:rules} have been applied until saturation, as in Figure~\ref{fig:algo}(c)).
Next, we show that a conforming partition of $ \polygon $ with stabbing number at most~$ 2 $ always exists, thereby proving that the algorithm is correct in this third case too.

To build a conforming partition of $ \polygon $ with stabbing number at most~$ 2 $, we partition $ \polygon $ along the green segments into rectilinear polygons $ \polygon _ 1 , \dots , \polygon _ \ell $ that we call the \emph{pieces} of $ \polygon $ (Figure~\ref{fig:pices}).
The idea is to prove that:
\begin{itemize}
    \item If we have a conforming partition with stabbing number at most~$ 2 $ for each piece of $ \polygon $, then the induced partition of $ \polygon $ is also conforming and has stabbing number at most~$ 2 $ (Lemma~\ref{lem:solvePieces}).
    \item Each piece of $ \polygon $ indeed has conforming stabbing number at most~$ 2 $ (Lemma~\ref{lem:thinGeneral}).
\end{itemize}

\begin{figure}[ht]
    \hspace*{\stretch{1}}
    \subcaptionbox{}{\includegraphics[scale=\graphicsScale,page=2]{algo.pdf}}
    \hspace*{\stretch{2}}
    \subcaptionbox{}{\includegraphics[scale=\graphicsScale,page=5]{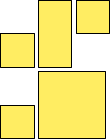}}
    \hspace*{\stretch{2}}
    \subcaptionbox{}{\includegraphics[scale=\graphicsScale,page=4]{algo.pdf}}
    \hspace*{\stretch{2}}
    \subcaptionbox{}{\includegraphics[scale=\graphicsScale,page=7]{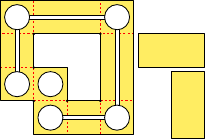}}
    \hspace*{\stretch{1}}
    \caption{
    (a) The polygon of Figure~\ref{fig:algo}(b).
    (b) The pieces of the polygon of Figure~\ref{fig:algo}(b). The constraint graph of each piece is empty. 
    (c) The polygon of Figure~\ref{fig:algo}(d). 
    (d) The pieces of the polygon of Figure~\ref{fig:algo}(d), with the constraint graph of each piece. 
    } 
    \label{fig:pices}
\end{figure}

Next, we make a useful observation that we then use to prove Lemma~\ref{lem:solvePieces}.

\begin{observation}\label{obs:rflxSgtOfPieceIsRflxSgtOfP}
    For every piece $ \polygon _ i $ of a rectilinear polygon $ \polygon $ (possibly with holes), every reflex segment $ \segment $ of $ \polygon _ i $ is a reflex segment of $ \polygon $ that is not colored (Figure~\ref{fig:pices}).
\end{observation}

\begin{proof}
Since $ \segment $ is a reflex segment of $ \polygon _ i $, one endpoint of $ \segment $, say $ p $, is a reflex vertex of $ \polygon _ i $, hence also a reflex vertex of $ \polygon $.
The other endpoint of $ \segment $, say $ q $, lies on the boundary of $ \polygon _ i $. 
If $ q $ is not on the boundary of $ \polygon $, then $ q $ is on an open green reflex segment $ \segment ' $. 
But then, rule~\ref{R3} applies in $ \polygon $ to $ \segment ' $ and the reflex segment of $ \polygon $ containing $ \segment $, a contradiction to the definition of the $ \polygon _ i $ (which implies that the rules have been applied until saturation).
Thus, $ q $ also lies on the boundary of $ \polygon $, and $ \segment $ is a reflex segment of $ \polygon $.

To see that $ \segment $ is not colored, observe first that $ p $ does not have an incident green reflex segment since it is reflex in the piece $ \polygon _ i $.
Thus, $ \segment $ is not green.
Now, $ \segment $ is not blue either, since otherwise the algorithm would have applied rule~\ref{R1} to color in green the other reflex segment incident to $ p $. 
Therefore, $ \segment $ is not colored.
\end{proof}

\begin{lemma}
    \label{lem:solvePieces}
    A rectilinear polygon $ \polygon $ (possibly with holes) is in $ \CSTAB{ 2 } $ if and only if each of the pieces $ \polygon _ 1 , \dots , \polygon _ \ell $ of $ \polygon $ is in $ \CSTAB{ 2 } $.
\end{lemma}

\begin{proof}
Any solution for $ \polygon $ includes all green segments, thereby yielding a solution for each piece.
Vice versa, assume that each piece $ \polygon _ i $ of $ \polygon $ admits a solution $ \partition _ i $ (seen as a set of reflex segments) to $ \CSTAB{ 2 } $. 
We show that $ \partition = \Rhor \cup \Rver \cup ( \bigcup _ i \partition _ i ) $ is a solution for $ \polygon $ ($ \Rhor \cup \Rver $ is equal to the set of all green segments at the end of the algorithm).
To see that $ \partition $ is a conforming partition, observe that it only contains reflex segments of $ \polygon $ by Observation~\ref{obs:rflxSgtOfPieceIsRflxSgtOfP}, and assigns at least one reflex segment to each reflex vertex of $ \polygon $.
Since the pieces $ \polygon _ 1 , \dots , \polygon _ \ell $ are interior-disjoint, the reflex segments in $ \bigcup _ i \partition _ i $ do not intersect each other. They do not intersect a green segment either, by Observation~\ref{obs:rflxSgtOfPieceIsRflxSgtOfP}, so $ \partition $ yields a conforming partition.

To show that $ \partition $ has stabbing number at most~$ 2 $, consider any stabbing segment $ \segment $ of $ \polygon $. 
If $ \segment $ intersects no green segment, then $ \segment $ is also a stabbing segment for one piece $ \polygon _ i $.
Thus, $ \segment $ will intersect at most one segment of $ \partition $.
Now assume that $ \segment $ intersects a green segment.
Since rule~\ref{R2} have been applied until saturation, all other reflex segments of $ \polygon $ intersected by $ \segment $ have been colored blue, thus, are not reflex segments of any pieces, and hence are not used by $ \partition $.
Therefore, stabbing segment $ \segment $ intersects at most one segment of $ \partition $.
\end{proof}

It remains to show that each piece of $ \polygon $ indeed has conforming stabbing number at most~$ 2 $.
Here, we observe that a piece $ \polygon _ i $ of $ \polygon $ is not an arbitrary polygon.
Specifically, since rule~\ref{R3} does not apply, $ \polygon _ i $ has no intersecting reflex segments, so $ \polygon _ i $ is thin. 
Moreover, since rule~\ref{R4} does not apply, $ \polygon _ i $ has no gates.
We now prove a statement that holds for any thin polygon that has no gates, and that we use to prove Lemma~\ref{lem:thinGeneral}.

\begin{observation}
    \label{obs:thinGeneral}
    Let $ \polygon $ be a thin rectilinear polygon (possibly with holes) that has no gates. 
    Then every stabbing segment of $ \polygon $ intersects at most two reflex segments of $ \polygon $ (Figure~\ref{fig:thinGeneral}). 
\end{observation}

\begin{proof}
Let $ \segment $ be a stabbing segment of $ \polygon $, and assume for contradiction that $ \segment $ intersects three reflex segments, say $ \segment $ intersects $ \segment _ 1 , \segment _ 2 , \segment _ 3 $, in this order and with no other reflex segments in between (Figure~\ref{fig:thinGeneral}).   
Up to symmetry, $ \segment $ is horizontal, so $ \segment _ 1 , \segment _ 2 , \segment _ 3 $ are vertical, and up to renaming, $ \segment _ 1 $ is on the left of $ \segment _ 2 $.
Let $ p $ be the reflex vertex of $ \polygon $ with $ \segment _ 2 = \verOf{ p } $.

Up to symmetry, the wedge-pixel $ \pixel $ of $ p $ is to the left of $ \verOf{ p } $ and below $ \horOf{ p } $.
Since $ \polygon $ is thin, pixel $ \pixel $ spans the entire length of $ \verOf{ p } $, and in particular includes the point common to $ \segment $ and $ \verOf{ p } $.
Since there are no vertical reflex segments between $ \segment _ 1 $ and $ \segment _ 2 $ along $ \segment $, pixel $ \pixel $ extends to the point common to $ \segment _ 1 $ and $ \segment $, and therefore spans the entire length of $ \segment _ 1 $. 
It also includes the entire length of $ \horOf{ p } $. 
Thus, the top left corner of $ \pixel $ is a point $ q $ common to $ \segment _ 1 $ and $ \horOf{ p } $, hence $ q $ is a reflex vertex that lies on a horizontal line with $ p $. 
Since both $ \verOf{ q } $ and $ \horOf{ q } $ bound sides of $ \pixel $, this makes $ \sgt{ p }{ q } $ a gate.
\end{proof}

\begin{figure}[ht]
    \centering
    \includegraphics[scale=\graphicsScale,page=7]{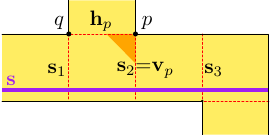}
    \caption{If a stabbing segment $ \segment $ intersects three reflex segments, and no two reflex segments intersect, then the polygon has a gate.   The upper right corner of the wedge-pixel $ \pixel $ of $ p $ is shaded (in orange).}
    \label{fig:thinGeneral}
\end{figure}

In the following, we prove the last lemma, which is rather interesting in itself, using Observation~\ref{obs:thinGeneral}.

\begin{lemma}
    \label{lem:thinGeneral}
    Let $ \polygon $ be a rectilinear $ n $-gon (possibly with holes) that is thin and has no gates (which is the case if $ \polygon $ is in general position).
    Then, $ \polygon $ has conforming stabbing number at most~$ 2 $.
    Moreover, there exists an algorithm that computes a conforming partition of $ \polygon $ with stabbing number at most~$ 2 $ in $ \OO ( n ) $ time.
\end{lemma}

\begin{proof}
We prove that $ \polygon $ admits a conforming partition with stabbing number at most~$ 2 $.
First, we reduce our problem to $ \SAT{ 2 } $ by taking advantage of the fact that $ \polygon $ is thin and has no gates.
Then, we prove that the $ \SAT{ 2 } $ formula we obtain is always satisfiable.

We start with some observations.
A conforming partition may be given by a set $ \partition $ of reflex segments, but it may equivalently be given as a function $ \function $ mapping each reflex vertex $ p $ to a subset of $ \{ \horOf{ p } , \verOf{ p } \} $, with the convention that the union of the image of the reflex vertices under $ \function $ is the set $ \partition $.
By definition of being conforming, each reflex vertex is mapped to at least one reflex segment.
Moreover, if $ \partition $ is minimal (i.e., no reflex segments can be removed while retaining a conforming partition; see Section~\ref{sec:preliminaries}), then each reflex vertex is mapped by $ \function $ to at most one reflex segment.
Thus, without loss of generality, we restrict our search to a conforming partition $ \function $ of $ \polygon $ with stabbing number at most~$ 2 $, where each reflex vertex is mapped by $ \function $ to at most one reflex segment.
(Note that a given reflex segment $ \segment $ may be the image of up to two reflex vertices, even if the conforming partition is minimal. For example, for a given reflex vertex $ p $, it is possible that both $ \horOf{ p } $ and $ \verOf{ p } $ are in a minimal partition $ \partition $ if $ q $, the other endpoint of $ \horOf{ p } $, is also a reflex vertex.
In this case, $ \function $ would map, say, $ p $ to $ \verOf{ p } $ and $ q $ to $ \horOf{ p } $.)

We now declare a Boolean variable $ \variable ( p ) $ for every reflex vertex $ p $ of $ \polygon $.
Given an assignment of all the variables, we build a function $ \function $ from the reflex vertices of $ \polygon $ to the reflex segments of $ \polygon $ as follows.
If $ \variable ( p ) $ is true, then we define $ \function ( p ) $ to be $ \horOf{ p } $.
If $ \variable ( p ) $ is false, then we define $ \function ( p ) $ to be $ \verOf{ p } $.

We now list the clauses which ensure that the function $ \function $ indeed encodes a conforming partition with stabbing number at most~$ 2 $.
By construction, we have ensured that for each reflex vertex $ p $, at least one of $ \horOf{ p } , \verOf{ p } $ is included in the partition $ \function $.
Moreover, since $ \polygon $ is thin, no reflex segments intersect, and in particular, no reflex segments in the partition $ \function $ intersect.
There remains to ensure that the stabbing number of $ \function $ is at most~$ 2 $.
By Observation~\ref{obs:thinGeneral}, every stabbing segment intersects at most two reflex segments.
Thus, for each equivalence class of vertical stabbing segment that intersects exactly two reflex segments of the form $ \horOf{ p } , \horOf{ q } $, we require
\begin{thmEnumerate}
    \item $ \lnot \variable ( p ) \lor \lnot \variable ( q ) $ (which is also equivalent to $ \variable ( p ) \implies \lnot \variable ( q ) $ and to $ \variable ( q ) \implies \lnot \variable ( p ) $).
\end{thmEnumerate}
Similarly, for each equivalence class of horizontal stabbing segment that intersects exactly two reflex segments of the form $ \verOf{ p } , \verOf{ q } $, we require
\begin{thmEnumerate}\setcounter{thmEnumeratei}{1}
    \item $ \variable ( p ) \lor \variable ( q ) $ (which is also equivalent to $ \lnot \variable ( p ) \implies \variable ( q ) $ and to $ \lnot \variable ( q ) \implies \variable ( p ) $).
\end{thmEnumerate}

We finally prove that the conjunction $ \formula $ of these clauses is always satisfiable.
Recall that a \emph{Horn clause} is a clause with at most one positive literal, and that $ \HORNSAT $, the restriction of $ \SATn $ to conjunctions of Horn clauses, is in $ \PP $. 
Our strategy is to show that
\begin{inlineEnum}
    \item these clauses are \emph{renameable} Horn clauses (i.e., they become Horn clauses after a change of variables), and that
    \item after renaming the appropriate variables, each clause contains a negative literal.
\end{inlineEnum}
A conjunction of Horn clauses where each clause contains a negative literal is always satisfiable (by assigning every variable to false), hence the lemma.

To this end, we consider the undirected simple graph $ \graph $, that we call the \emph{constraint} graph of $ \polygon $, defined as follows (Figure~\ref{fig:pices}(d)).
\begin{itemize}
    \item The vertices of $ \graph $ are the pixels of $ \polygon $ that are the wedge-pixel of some reflex vertex $ p $.
    Each vertex of $ \graph $ is thus associated with a unique variable $ \variable ( p ) $.
    \item There is an edge between two pixels $ \pixel , \pixel ' $ whenever there exists a stabbing segment intersecting both $ \pixel $ and $ \pixel ' $.
    Each edge of $ \graph $ is thus associated with a unique clause of $ \formula $.
\end{itemize}
We consider the constraint graph $ \graph $ to be drawn in the plane with straight line segments included in corresponding stabbing segments.
By Observation~\ref{obs:thinGeneral}, a vertex of $ \graph $ is incident to at most one vertical edge and at most one horizontal edge.
Therefore, the maximum degree of $ \graph $ is $ 2 $, and $ \graph $ consists of a set of cycles and paths (which may be isolated vertices).
(In fact, one may show that $ \graph $ is either one cycle or a set of paths, but this is not useful for our purpose.)
Moreover, each cycle of $ \graph $ is even, as it is drawn as a rectilinear polygon (without holes).

We now rename every other variable $ \variable ( p ) $ along each path and each cycle of $ \graph $ as $ \variable ' ( p ) = \lnot \variable ( p ) $.
This is possible because the cycles are all even.
By construction after this change of variables, the clauses are now Horn clauses, and each contains a negative literal.

To conclude the proof, we have declared a linear number of variables and built a linear number of clauses.
The renaming process that we have described is also linear, and so is the assignment computation.
Therefore, we have described an algorithm that computes a conforming partition of $ \polygon $ with stabbing number at most~$ 2 $ in linear time.
\end{proof}

We finally state and prove the main theorem of this section.

\begin{theorem} \label{thm:decideStabbing2}
    There exists an algorithm that, for any rectilinear $ n $-gon $ \polygon $ (possibly with holes), decides $ \CSTAB{ 2 } $ and provides a solution (if any) in $ \OO ( n \log n ) $ time.
\end{theorem}

\begin{proof}
The algorithm described above runs in $ \OO ( n \log n ) $ time by Lemma~\ref{lem:runningTime} and correctly decides $ \CSTAB{ 2 } $ by Lemma~\ref{lem:solvePieces} and Lemma~\ref{lem:thinGeneral}.
Besides, it is possible to compute a conforming partition with stabbing number at most~$ 2 $ of $ \polygon $ in $ \OO ( n ) $ extra time by Lemma~\ref{lem:thinGeneral}.
\end{proof}

\section{Polygons with Small Treewidth} \label{sec:tw}

We now turn towards FPT algorithms, and in particular, study polygons with bounded treewidth. 
The treewidth (the definition is recalled in Section~\ref{sec:preliminaries}) has frequently been used for FPT algorithms for graph problems but can also be used for solving problems on polygons, see e.g.~\cite{BiedlMehrabi}.

In our case, we consider the treewidth of the \emph{pixel graph} $ \pixelGraph{ \polygon } $ of $ \polygon $ (defined in Section~\ref{sec:preliminaries} and illustrated again in Figure~\ref{fig:MSO_ex}(a)).
In fact, we rely on a graph $ \extGraph{ \polygon } $ derived from $ \pixelGraph{ \polygon } $;  this graph $ \extGraph{ \polygon } $ is defined next and illustrated in Figure~\ref{fig:MSO_ex}(b).
\begin{itemize}
    \item The vertex set of $ \extGraph{ \polygon } $, denoted by $ \vertexSetOf{ \extGraph{ \polygon } } $, is $ \vertexSetOf{ \extGraph{ \polygon } } = \vertexSetOf{ \pixelGraph{ \polygon } } \cup \edgeSetOf{ \pixelGraph{ \polygon } } \cup \pixelSetOf{ \pixelGraph{ \polygon } } $, the union of the vertices of $ \pixelGraph{ \polygon } $, the edges of $ \pixelGraph{ \polygon } $, as well as the pixels of $ \polygon $.
    \item The edge set of $ \extGraph{ \polygon } $, denoted by $ \edgeSetOf{ \extGraph{ \polygon } } $, consists of all the pairs of vertices 
    \begin{itemize}
        \item $ e , v \in \edgeSetOf{ \pixelGraph{ \polygon } } \times \vertexSetOf{ \pixelGraph{ \polygon } } $ such that $ v $ is an endpoint of the segment $ e $, and 
        \item $ e , \pixel \in \edgeSetOf{ \pixelGraph{ \polygon } } \times \pixelSetOf{ \pixelGraph{ \polygon } } $ such that $ e $ is an edge of the pixel $ \pixel $.
    \end{itemize}
\end{itemize}
This new graph $ \extGraph{ \polygon } $ has treewidth $ \OO ( \tw ( \pixelGraph{ \polygon } ) ) $.
Indeed, we construct a tree decomposition $ \treeDecomposition ' $ of $ \extGraph{ \polygon } $ from a tree decomposition $ \treeDecomposition $ of $ \pixelGraph{ \polygon } $ as follows.
For each $ v \in \vertexSetOf{ \pixelGraph{ \polygon } } $ and for each bag of $ \treeDecomposition $, we construct a bag of $ \treeDecomposition ' $ by adding the (up to) four edges $ e \in \edgeSetOf{ \pixelGraph{ \polygon } } $ which are adjacent to $ v $, together with the (up to) four pixels $ \pixel \in \pixelSetOf{ \pixelGraph{ \polygon } } $ which share the vertex $ v $.
By construction, for every edge $ v _ 1 , v _ 2 $ of $ \extGraph{ \polygon } $ that is not an edge in $ \pixelGraph{ \polygon } $, we have added the vertices $ v _ 1 , v _ 2 $ to a common bag.
Thus, the tree $ \treeDecomposition ' $ satisfies the property~\ref{edgeInBag} of the definition of a tree decomposition.
As for the property~\ref{connectedSubtree}, consider a vertex $ v ' $ of $ \extGraph{ \polygon } $, and the set $ T $ of bags of $ \treeDecomposition ' $ that contain $ v ' $.
If $ v ' \in \vertexSetOf{ \pixelGraph{ \polygon } } $, then $ T $ is indeed a non-empty connected subtree of $ \treeDecomposition ' $ since~\ref{connectedSubtree} holds for $ \treeDecomposition $ and since $ \treeDecomposition ' $ is just a relabeling of $ \treeDecomposition $.
If $ v ' \in \edgeSetOf{ \pixelGraph{ \polygon } } $, say $ v ' $ is the edge adjacent to $ v _ 1 , v _ 2 $ in $ \extGraph{ \polygon } $, then $ T $ is the union of the non-empty connected subtrees corresponding to $ v _ 1 $ and $ v _ 2 $ in $ \treeDecomposition $ by~\ref{connectedSubtree}.
These two connected subtrees share a common bag by~\ref{edgeInBag}, which ensures that their union, $ T $, is indeed a (non-empty) connected subtree of $ \treeDecomposition ' $.
If $ v ' \in \pixelSetOf{ \pixelGraph{ \polygon } } $, say $ v ' $ is the pixel with corners $ v _ 1 , v _ 2 , v _ 3 , v _ 4 $ in counterclockwise order, then a similar argument shows that $ T $, as the union of four non-empty connected subtrees corresponding to $ v _ 1 , v _ 2 , v _ 3 , v _ 4 $, is again a (non-empty) connected subtree of $ \treeDecomposition ' $.

\begin{figure}[ht]
    \hspace*{\stretch{1}}
    \subcaptionbox{}{\includegraphics[scale=\graphicsScale,page=2]{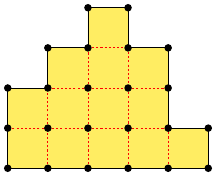}}
    \hspace*{\stretch{2}}
    \subcaptionbox{}{\includegraphics[scale=\graphicsScale,page=5]{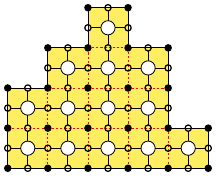}}
    \hspace*{\stretch{2}}
    \subcaptionbox{}{\includegraphics[scale=\graphicsScale,page=6]{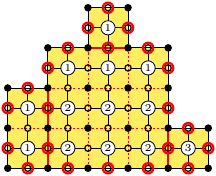}}
    \hspace*{\stretch{1}}
    \caption{
    (a) A polygon $ \polygon $ with its pixel graph $ \pixelGraph{ \polygon } $. The vertices of $ \pixelGraph{ \polygon } $ are solid black disks; the edges of $ \pixelGraph{ \polygon } $ are dotted red lines except when they lie on the boundary of $ \polygon $, in which case they are solid black lines.
    (b) The graph $ \extGraph{ \polygon } $. The vertices of $ \extGraph{ \polygon } $ that are edges of $ \pixelGraph{ \polygon } $ are small hollow black circles, while the vertices of $ \extGraph{ \polygon } $ that are pixels are big black circles filled with white. Edges of $ \extGraph{ \polygon } $ that are not edges of $ \pixelGraph{ \polygon } $ are thin solid black lines.
    (c) One possible solution to the MSO formula $ \formula $ for $ k = 3 $.
    A partition $ \partition $ is represented with plain bold red line segments. Vertices in the interpretation of $ \Variable $ corresponding to the partition $ \partition $ are circled in bold red. We label a pixel $ \pixel $ with an \raisebox{0.5pt}{\textcircled{\raisebox{-0.7pt}{$ i $}}} to indicate that $ \pixel $ belongs to $ \interpreted{ \extGraph{ \polygon } }{ \pixelSetVariable _ i ^ \hor } $. The idea is to encode a ``count'' of the horizontal stabbing number of each row of pixels from left to right.
    }
    \label{fig:MSO_ex}
\end{figure}

In the proof of Theorem~\ref{thm:treewidth} (stated next), we show how to exploit small treewidth of $ \extGraph{ \polygon } $ to compute both the stabbing number of $ \polygon $ and the conforming stabbing number of $ \polygon $.
To this end, we use Courcelle's theorem~\cite{courcelle1990monadic,Courcelle-book,ParamCompelxity-book},
which states that if a graph property can be expressed in monadic second-order logic (MSO) as a formula $ \formula $, then testing whether a graph $ \graph $ satisfies $ \formula $ is computable in time that is linear in the number of vertices of $ \graph $ and fixed-parameter tractable in $ | \formula | + \tw ( \graph ) $.   
The length of the two formulas (one for the stabbing number of $ \polygon $ and one for the conforming stabbing number of $ \polygon $) we exhibit in the proof of Theorem~\ref{thm:treewidth} is linear in $ k $ and independent of the size of the graph $ \extGraph{ \polygon } $. 
The pixel graph $ \pixelGraph{ \polygon } $ has $ \OO ( n ^ 2 ) $ vertices, and so does $ \extGraph{ \polygon } $.
Therefore, with Courcelle's theorem, we obtain FPT algorithms which are quadratic in $ n $ to decide each of $ \STAB{ k } $ and $ \CSTAB{ k } $.

\begin{theorem}
    \label{thm:treewidth}
    There exists an algorithm that, for a rectilinear $ n $-gon $ \polygon $ (possibly with holes) with treewidth $ \ell $, decides $ \STAB{ k } $ and $ \CSTAB{ k } $ in $ \OO ( f ( k + \ell ) \ n ^ 2 )$ time, for some function $ f $ that does not depend on $ n $.
\end{theorem}

\begin{proof}
As mentioned, Theorem~\ref{thm:treewidth} follows from Courcelle's theorem, which requires to exhibit two MSO formulas, say $ \MSOFormula{ k } $ for $ \STAB{ k } $ and $ \MSOFormulaConf{ k } $ for $ \CSTAB{ k } $, which do not dependent on $ \polygon $.

Recall that an MSO formula in graph theory may contain three types of variables:
\begin{itemize}
    \item First order variables, denoted using lower case letters such as $ \FOvariable $, that are interpreted as a vertex, denoted $ \interpreted{ \extGraph{ \polygon } }{ \FOvariable } $, in the graph $ \extGraph{ \polygon } $. 
    \item MSO$ _ 1 $ variables, denoted using upper case letters such as $ \Variable $, that are interpreted as a set of vertices, denoted $ \interpreted{ \extGraph{ \polygon } }{ \Variable } $, in the graph $ \extGraph{ \polygon } $ (the set $ \interpreted{ \extGraph{ \polygon } }{ \Variable } $ may be viewed as an \emph{unary} relation, hence the term ``monadic''). 
    \item MSO$ _ 2 $ variables that are interpreted as a set of edges in the graph $ \extGraph{ \polygon } $. Our formulas $ \MSOFormula{ k } $ and $ \MSOFormulaConf{ k } $ do not use those, which makes them formulas of MSO$ _ 1 $.\footnote{Since we only use MSO$ _ 1 $ and not MSO$ _ 2 $, we actually have a stronger result: These algorithms are FPT even in the \emph{cliquewidth} of $ \extGraph{ \polygon } $~\cite{Courcelle-book}.}
\end{itemize}
Both types of variable may be quantified with $ \forall , \exists $, the usual Boolean operators $ \neg , \wedge , \vee , \implies $ may be used, as well as the binary predicates $ \in $ (as in $ \FOvariable \in \Variable $) and $ \isAdjacent{ } $ (as in $ \isAdjacent{ } ( v _ 1 , v _ 2 ) $, interpreted as true iff the vertices $ v _ 1 , v _ 2 $ are adjacent in the graph $ \extGraph{ \polygon } $). 
We freely use the usual shorthands, such as $ \forall \FOvariable \in \Variable \ ( \formula ( \FOvariable , \dots ) ) $ in place of $ \forall \FOvariable \ ( \FOvariable \in \Variable \implies \formula ( \FOvariable , \dots ) ) $, and $ \FOvariable \not \in \Variable $ in place of $ \neg ( \FOvariable \in \Variable ) $, for any first order variable $ \FOvariable $, any second order variable or constant $ \Variable $, and any formula $ \formula $.

In fact, Courcelle's theorem also applies to \emph{labeled} graphs, where each vertex and each edge is labeled by a constant number of bits.
In other words, our formulas $ \MSOFormula{ k } $ and $ \MSOFormulaConf{ k } $ may use additional predicates as long as their interpretation only involves reading some labels on vertices or edges of the graph $ \extGraph{ \polygon } $.
In our case, we label each vertex $ v $ of $ \extGraph{ \polygon } $ to read which of $ \vertexSetOf{ \pixelGraph{ \polygon } } , \edgeSetOf{ \pixelGraph{ \polygon } } , \pixelSetOf{ \pixelGraph{ \polygon } } $ the vertex $ v $ belongs to, as well as whether $ v $ lies on the boundary of the polygon $ \polygon $.
We also duplicate each edge of $ \extGraph{ \polygon } $ into two directed edges with a label to read whether they are horizontal or vertical and whether they lie on the boundary of the polygon $ \polygon $.
In the following, we describe the additional predicates (readable from our labeling of $ \extGraph{ \polygon } $ is constant time) that we introduce.
\begin{itemize}
    \item The unary predicate $ v \in \vertexSetVariable $ is interpreted as true if and only if $ \interpreted{ \extGraph{ \polygon } }{ v } \in \vertexSetOf{ \pixelGraph{ \polygon } } $.
    \item The unary predicate $ e \in \edgeSetVariable $ is interpreted as true if and only if $ \interpreted{ \extGraph{ \polygon } }{ e } \in \edgeSetOf{ \pixelGraph{ \polygon } } $.
    \item The unary predicate $ \pixel \in \pixelSetVariable $ is interpreted as true if and only if $ \interpreted{ \extGraph{ \polygon } }{ \pixel } \in \pixelSetOf{ \pixelGraph{ \polygon } } $.
    \item The unary predicate \[ \isOnBoundary ( v ) \] is interpreted as true if and only if $ \interpreted{ \extGraph{ \polygon } }{ v } \in \vertexSetOf{ \pixelGraph{ \polygon } } \cup \edgeSetOf{ \pixelGraph{ \polygon } } $ and $ \interpreted{ \extGraph{ \polygon } }{ v } $ lies on the boundary of the polygon $ \polygon $.
    \item For each $ \alpha \in \{ \north , \south , \est , \west \} $, the binary predicate \[ \isAdjacent{ \alpha } ( v , e ) \] is interpreted as true if and only if $ \interpreted{ \extGraph{ \polygon } }{ e } \in \edgeSetOf{ \pixelGraph{ \polygon } } $ is the \emph{$ \alpha $ edge} from $ \interpreted{ \extGraph{ \polygon } }{ v } $, i.e., goes in the direction $ \alpha $ from the vertex $ \interpreted{ \extGraph{ \polygon } }{ v } \in \vertexSetOf{ \pixelGraph{ \polygon } } $.
    \item The ternary predicate \[ \areAdjacent ^ \hor ( \pixel , e , \pixel ' ) \quad ( \text{respectively } \areAdjacent ^ \ver ( \pixel , e , \pixel ' ) ) \] is interpreted as true if and only if $ \interpreted{ \extGraph{ \polygon } }{ \pixel } , \interpreted{ \extGraph{ \polygon } }{ \pixel ' } \in \pixelSetOf{ \pixelGraph{ \polygon } } $, $ \interpreted{ \extGraph{ \polygon } }{ e } \in \edgeSetOf{ \pixelGraph{ \polygon } } $, and $ \interpreted{ \extGraph{ \polygon } }{ \pixel } , \interpreted{ \extGraph{ \polygon } }{ e } , \interpreted{ \extGraph{ \polygon } }{ \pixel ' } $ forms a horizontal (respectively vertical) path directed from left to right (respectively from bottom to top) in $ \extGraph{ \polygon } $.
\end{itemize}

We now build the formula $ \MSOFormula{ k } $ (respectively $ \MSOFormulaConf{ k } $) that is satisfied if and only if polygon $ \polygon $ has stabbing number at most $ k $ (respectively has conforming stabbing number at most $ k $).
We define $ \MSOFormula{ k } $ and $ \MSOFormulaConf{ k } $ via subformulas with explicit names to ease the reading.
Both of $ \MSOFormula{ k } $ and $ \MSOFormulaConf{ k } $, begin with $ \exists \Variable $, where the MSO$ _ 1 $ variable $ \Variable $ is to be interpreted as a set $ \interpreted{ \extGraph{ \polygon } }{ \Variable } $ of vertices of the graph $ \extGraph{ \polygon } $.
This set $ \interpreted{ \extGraph{ \polygon } }{ \Variable } $ is constrained by the formulas to be a set of edges of the pixel graph $ \pixelGraph{ \polygon } $ that yields a partition of $ \polygon $ (in the sense that the edges of the rectangles of the partition are formed by concatenating the edges in $ \interpreted{ \extGraph{ \polygon } }{ \Variable } $) with stabbing number at most $ k $.
\begin{flalign*}
\MSOFormula{ k } \Equal \exists \Variable \ \big [ \ & \isPartition ( \Variable ) \ \wedge & \\
& \hasStabbingNumberAtMost{ k } ^ \hor ( \Variable ) \ \wedge & \\
& \hasStabbingNumberAtMost{ k } ^ \ver ( \Variable ) \quad \ \big ] &
\end{flalign*}
\begin{flalign*}
\MSOFormulaConf{ k } \Equal \exists \Variable \ \big [ \ & \isPartition ( \Variable ) \ \wedge & \\
& \isConforming ( \Variable ) \ \wedge & \\
& \hasStabbingNumberAtMost{ k } ^ \hor ( \Variable ) \ \wedge & \\
& \hasStabbingNumberAtMost{ k } ^ \ver ( \Variable ) \quad \ \big ] &
\end{flalign*}

We now define the subformulas, starting with $ \isPartition $.
\begin{flalign*}
    \isPartition ( \Variable ) \Equal & \containsOnlyEdgesOfThePixelGraph ( \Variable ) \ \wedge & \\
    & \coversAllEdgesOfThePolygon ( \Variable ) \ \wedge & \\
    & \yieldsOnlyRectangles ( \Variable ) &
\end{flalign*}
The first two subformulas of $ \isPartition $ have the following easy definitions.
\begin{flalign*}
    \containsOnlyEdgesOfThePixelGraph ( \Variable ) \Equal \forall e \in \Variable \ ( e \in \edgeSetVariable ) & &
\end{flalign*}
\begin{flalign*}
    \coversAllEdgesOfThePolygon ( \Variable ) \Equal \forall e \in \edgeSetVariable \ ( \isOnBoundary ( e ) \implies e \in \Variable ) & &
\end{flalign*}
The last subformula of $ \isPartition $ has the following definition.
It constrains the set $ \interpreted{ \extGraph{ \polygon } }{ \Variable } $ to yield only rectangles in the partition, that is, that the edges of $ \pixelGraph{ \polygon } $ which are in $ \interpreted{ \extGraph{ \polygon } }{ \Variable } $ and are incident to a common vertex $ v \in \vertexSetOf{ \pixelGraph{ \polygon } } $ do not form reflex vertex (of degree two or one).
Since the edges of the pixel graph $ \pixelGraph{ \polygon } $ which lie on the boundary of $ \polygon $ are all included in the partition defined by $ \interpreted{ \extGraph{ \polygon } }{ \Variable } $ (via $ \containsOnlyEdgesOfThePixelGraph $), such a reflex vertex $ v $ has degree four in the pixel graph.
We therefore enforce, for every vertex $ v $ of degree four in the pixel graph, that if the $ \north $ edge from $ v $ is in $ \interpreted{ \extGraph{ \polygon } }{ \Variable } $, then either the $ \south $ edge from $ v $ or both $ \est $ and $ \west $ edges from $ v $ are also in $ \interpreted{ \extGraph{ \polygon } }{ \Variable } $.
We also enforce the three other similar rules, one for each orientation of the same situation.
\begin{flalign*}
    \yieldsOnlyRectangles ( \Variable ) \Equal \forall v \in \vertexSetVariable \ \exists e _ \north , e _ \south , e _ \est , e _ \west \ \Bigg [ & &
\end{flalign*}\vspace{-2em}
\begin{flalign*}
    \left ( \vphantom{ \bigwedge _{ \alpha \in \{ \north , \south , \est , \west \} } } \right . & \left . \bigwedge _{ \alpha \in \{ \north , \south , \est , \west \} } \isAdjacent{ \alpha } ( v , e _ \alpha ) \right ) \implies & \\
    \Bigg ( & \Big ( \ e _ \north \in \Variable \ \implies \ \big ( e _ \south \in \Variable \ \vee \ ( e _ \est \in \Variable \ \wedge \ e _ \west \in \Variable ) \big ) \ \Big ) \ \wedge \\
    & \Big ( \ e _ \south \in \Variable \ \implies \ \big ( e _ \north \in \Variable \ \vee \ ( e _ \est \in \Variable \ \wedge \ e _ \west \in \Variable ) \big ) \ \Big ) \ \wedge \\
    & \Big ( \ e _ \est \in \Variable \ \implies \ \big ( e _ \west \in \Variable \ \vee \ ( e _ \north \in \Variable \ \wedge \ e _ \south \in \Variable ) \big ) \ \Big ) \ \wedge \\
    & \Big ( \ e _ \west \in \Variable \ \implies \ \big ( e _ \est \in \Variable \ \vee \ ( e _ \north \in \Variable \ \wedge \ e _ \south \in \Variable ) \big ) \ \Big ) \ \phantom{ \wedge \ } \Bigg ) \Bigg ] &
\end{flalign*}

We now define the subformula $ \isConforming $.
The idea is somewhat similar to the definition of $ \yieldsOnlyRectangles $.
For each vertex $ v $ of the pixel graph that has degree four in the pixel graph and that lies in the interior of $ \polygon $, we enforce that $ v $ is not a Steiner point. In other words, we enforce that if an edge of the pixel graph is in $ \interpreted{ \extGraph{ \polygon } }{ \Variable } $ and is incident to $ v $, then there is another such edge extending the first one in the same direction, but we exclude the situation where all four edges of the pixel graph incident to $ v $ are in $ \interpreted{ \extGraph{ \polygon } }{ \Variable } $.
\begin{flalign*}
    \isConforming ( \Variable ) \Equal \forall v \in \vertexSetVariable \ \exists e _ \north , e _ \south , e _ \est , e _ \west \Bigg [ & &
\end{flalign*}\vspace{-2em}
\begin{flalign*}
    & & \left ( \neg \isOnBoundary ( v ) \ \wedge \left ( \bigwedge _{ \alpha \in \{ \north , \south , \est , \west \} } \isAdjacent{ \alpha } ( v , e _ \alpha ) \right ) \right ) \implies \\
    & & \Bigg ( \ ( e _ \north \in \Variable \ \iff \ e _ \south \in \Variable ) \ \wedge \ \hphantom{ \Bigg ) \Bigg ] } \\
    & & ( e _ \est \in \Variable \ \iff \ e _ \west \in \Variable ) \ \wedge \ \hphantom{ \Bigg ) \Bigg ] } \\
    & & ( e _ \north \not \in \Variable \ \vee \ e _ \est \not \in \Variable ) \ \hphantom{ \wedge } \ \Bigg ) \Bigg ]
\end{flalign*}

Finally, we define the subformula $ \hasStabbingNumberAtMost{ k } ^ \hor $ (the subformula $ \hasStabbingNumberAtMost{ k } ^ \ver $ has exactly the same definition, except that all the occurrences of $ ^ \hor$ are replaced by $ ^ \ver $).
The idea is to define $ k $ MSO$ _ 1 $ variables $ \pixelSetVariable _ 1 ^ \hor , \dots , \pixelSetVariable _ k ^ \hor $, which are to be interpreted as $ k $ sets $ \interpreted{ \extGraph{ \polygon } }{ \pixelSetVariable _ 1 ^ \hor } , \dots , \interpreted{ \extGraph{ \polygon } }{ \pixelSetVariable _ k ^ \hor } $ of vertices of the graph $ \extGraph{ \polygon } $. These will be constrained by the formulas to be interpreted as a set of pixels of $ \polygon $ that yields a ``count'' of the horizontal stabbing number of each row of pixels in $ \polygon $ from left to right (Figure~\ref{fig:MSO_ex}(c)).
More specifically, in the interpretation:
\begin{thmEnumerate}
    \item\label{pixelSetVariablesCover} We have $ \interpreted{ \extGraph{ \polygon } }{ \pixelSetVariable _ 1 ^ \hor } \cup \dots \cup \interpreted{ \extGraph{ \polygon } }{ \pixelSetVariable _ k ^ \hor } = \pixelSetOf{ \pixelGraph{ \polygon } } $. Note that we do not impose $ \interpreted{ \extGraph{ \polygon } }{ \pixelSetVariable _ 1 ^ \hor } , \dots , \interpreted{ \extGraph{ \polygon } }{ \pixelSetVariable _ k ^ \hor } $ to be pairwise disjoint.
    \item\label{pixelSetVariablesCount} If we have a pixel $ \pixel \in \interpreted{ \extGraph{ \polygon } }{ \pixelSetVariable _ i ^ \hor } $ and call $ c $ its center, and if $ \segment $ is the horizontal stabbing segment containing $ c $, then the portion of $ \segment $ to the left of $ c $ stabs at most $ i $ rectangles of the partition defined by $ \interpreted{ \extGraph{ \polygon } }{ \Variable } $. Note that we do not impose the leftmost pixel of a row to belong to $ \interpreted{ \extGraph{ \polygon } }{ \pixelSetVariable _ 1 ^ \hor } $.
\end{thmEnumerate}
\begin{flalign*}
    \hasStabbingNumberAtMost{ k } ^ \hor ( \Variable ) \Equal & \exists \pixelSetVariable _ 1 ^ \hor , \dots , \pixelSetVariable _ k ^ \hor \ \bigg [ & &
\end{flalign*}\vspace{-2em}
\begin{flalign*}
    & & \coversAllPixels ( \pixelSetVariable _ 1 ^ \hor , \dots , \pixelSetVariable _ k ^ \hor ) \ \wedge \ \hphantom{ \bigg ] } \\
    & & \isCountIncrementedAtEachStab ^ \hor ( \Variable , \pixelSetVariable _ 1 ^ \hor , \dots , \pixelSetVariable _ k ^ \hor ) \ \wedge \ \hphantom{ \bigg ] } \\
    & & \isCountNotDecrementedAtEachNonStab ^ \hor ( \Variable , \pixelSetVariable _ 1 ^ \hor , \dots , \pixelSetVariable _ k ^ \hor ) \phantom{ \ \wedge \ } \bigg ] 
\end{flalign*}
Note that the formulas $ \hasStabbingNumberAtMost{ k } ^ \hor $ has a size linear in $ k $ while all the other subformulas defined thus far had constant size.   (Some of the subformulas for $ \hasStabbingNumberAtMost{ k } ^ \hor $ will have size linear in $ k $ as well, as will the formula for $ \hasStabbingNumberAtMost{ k } ^ \ver $ and its subformulas, but no formula will have bigger size.) 

The first subformula of $ \hasStabbingNumberAtMost{ k } ^ \hor $ is defined as follows without difficulty and implements~\ref{pixelSetVariablesCover}. 
\begin{flalign*}
    \coversAllPixels ( \pixelSetVariable _ 1 ^ \hor , \dots , \pixelSetVariable _ k ^ \hor ) \Equal \forall \pixel \in \pixelSetVariable \ \left ( \bigvee _{ i = 1 } ^{ k } ( \pixel \in \pixelSetVariable _ i ^ \hor ) \right ) & &
\end{flalign*}
The remaining two subformulas implement~\ref{pixelSetVariablesCount} locally: consider two pixels $ \pixel , \pixel ' $ sharing a common vertical edge $ e $ with $ \pixel $ on the left of $ e $.
With the subformula $ \isCountIncrementedAtEachStab ^ \hor $, we impose that, if $ e \in \interpreted{ \extGraph{ \polygon } }{ \Variable } $ and if $ \pixel \in \interpreted{ \extGraph{ \polygon } }{ \pixelSetVariable _ i ^ \hor } $, then $ \pixel ' \in \interpreted{ \extGraph{ \polygon } }{ \pixelSetVariable _{ i + 1 } ^ \hor } $.
With $ \isCountNotDecrementedAtEachNonStab ^ \hor $, we impose that, if $ e \not \in \interpreted{ \extGraph{ \polygon } }{ \Variable } $ and if $ \pixel \in \interpreted{ \extGraph{ \polygon } }{ \pixelSetVariable _ i ^ \hor } $, then $ \pixel ' \in \interpreted{ \extGraph{ \polygon } }{ \pixelSetVariable _ i ^ \hor } $.
A simple induction on $ i $ proves that this indeed implements~\ref{pixelSetVariablesCount}.
\begin{flalign*} 
    & \isCountIncrementedAtEachStab ^ \hor ( \Variable , \pixelSetVariable _ 1 ^ \hor , \dots , \pixelSetVariable _ k ^ \hor ) \Equal &
\end{flalign*}
\[
    ~ ~ \forall e \in \Variable \ \forall \pixel , \pixel ' \in \pixelSetVariable  \ \Bigg [ \areAdjacent ^ \hor ( \pixel , e , \pixel ' ) \ \implies \ \left ( \bigwedge _{ i = 1 } ^{ k - 1 } ( \pixel \in \pixelSetVariable _ i ^ \hor \implies \pixel ' \in \pixelSetVariable _{ i + 1 } ^ \hor ) \right ) \Bigg ]
\]
\begin{flalign*} 
    & \isCountNotDecrementedAtEachNonStab ^ \hor ( \Variable , \pixelSetVariable _ 1 ^ \hor , \dots , \pixelSetVariable _ k ^ \hor ) \Equal &
\end{flalign*}
\[
    ~ ~ \forall e \not \in \Variable \ \forall \pixel , \pixel ' \in \pixelSetVariable  \ \Bigg [ \areAdjacent ^ \hor ( \pixel , e , \pixel ' ) \ \implies \ \left ( \bigwedge _{ i = 1 } ^{ k } ( \pixel \in \pixelSetVariable _ i ^ \hor \implies \pixel ' \in \pixelSetVariable _{ i } ^ \hor ) \right ) \Bigg ]
\]
This concludes the proof of Theorem~\ref{thm:treewidth}.
\end{proof}

\section{Hole-Free Polygons in General Position} \label{sec:gate-free-hole-free}

In this section, we give a second FPT algorithm, which makes a different assumption on the input polygon $ \polygon $.   
We require $ \polygon $ to have no holes and no gates (the latter holds in particular if $ \polygon $ is in general position), but in exchange we no longer need to bound the treewidth.
The idea for this theorem is to distinguish two cases using the maximum number $ m $ of reflex segments intersected by a stabbing segment in $ \polygon $: if $ m $ is smaller than a well-chosen threshold $ \ell $, then we show that the treewidth of $ \polygon $ is small and Theorem~\ref{thm:treewidth} applies, and if $ m $ is at least $ \ell $, then we show that the conforming stabbing number is bigger than $ k $ (Lemma~\ref{lem:gridStabbing}).

\begin{theorem}
    \label{thm:simpleGeneral}
    There exists an algorithm that, for any rectilinear $ n $-gon $ \polygon $ that has no holes and no gates (the latter holds if $ \polygon $ is in general position), decides $ \CSTAB{ k } $ and provides a solution (if any) in $ \OO ( f '( k ) n ^ 2 ) $ time, for some function $ f ' $ that does not depend on $ n $.
\end{theorem}

\begin{proof}
We show Theorem~\ref{thm:simpleGeneral} by applying a dichotomy argument somewhat similar to bidimensionality~\cite{DFHT05}.
We choose the threshold $ \ell = k \big( 2 ^{ k + 2 } + 4 ( k + 1 ) \big) + ( k - 1 ) $ (the reason for this choice of $ \ell $ will be clear later; it has been chosen for ease of explanation and has not been optimized).
Let $ m $ be the maximum number of reflex segments that are intersected by a stabbing segment in $ \polygon $. We have the following two cases.
\begin{Cases}
    \item $ m < \ell $.
    Recall that the \emph{outerplanarity} of a plane graph $ \graph $ is:
    \begin{itemize}
        \item $ 1 $ if $ \graph $ is outerplanar (i.e. all the vertices of $ \graph $ lie on the outer face of $ \graph $), and is
        \item  $ j \geq 2 $ if the induced subgraph of $ \graph $ obtained by removing the vertices of the outer face of $ \graph $ has outerplanarity $ j - 1 $.
    \end{itemize}
    It is well-known that the treewidth of a planar graph $ \graph $ is at most three times the outerplanarity of $ \graph $ minus one~\cite{Bod88outerplanarityBoundsTreewidth,Bod98arboretum,Kat13outerplanarityBoundsTreewidth}.

    Applied to our case, if we prove that the outerplanarity of $ \pixelGraph{ \polygon } $ is $ \OO ( \ell ) $ (which is done in the next paragraph), then the treewidth of $ \pixelGraph{ \polygon } $ is also $ \OO ( \ell ) $.
    Theorem~\ref{thm:treewidth} then applies and ensures that $ \CSTAB{ k } $ is decidable in $ \OO ( f ( k , \OO ( \ell ) ) n ^ 2 ) $ time.
    Since $ \ell $ only depends on $ k $ and not on $ n $, this gives the result in this case.

    We now prove that the outerplanarity of $ \pixelGraph{ \polygon } $ is $ \OO ( \ell ) $.
    Since we have $ m < \ell $, every vertex $ v $ of the pixel graph $ \pixelGraph{ \polygon } $ is within distance $ \lceil \ell / 2 \rceil $ of a vertex that lies on the boundary of $ \polygon $, simply by walking along one of the two reflex segments (seen as a stabbing segment) intersecting at $ v $ to reach the boundary.
    Since $ \polygon $ has no holes, all vertices of $ \pixelGraph{ \polygon } $ are within distance $ \lceil \ell / 2 \rceil $ of some vertex on the outer face of $ \pixelGraph{ \polygon } $.
    In other words, the outerplanarity of $ \pixelGraph{ \polygon } $ is at most $ \lceil \ell / 2 \rceil + 1 = \OO ( \ell ) $.

    \item $ m \geq \ell $.
    In this case, Lemma~\ref{lem:gridStabbing} (proved hereafter) applies and ensures that $ \CSTAB{ k } $ has no solution.
    In other words, deciding $ \CSTAB{ k } $ in this case amounts to systematically reject.\qedhere
\end{Cases}
\end{proof}

\begin{lemma}
    \label{lem:gridStabbing}
    Let $ \polygon $ be a gate-free rectilinear polygon (possibly with holes) that contains a stabbing segment $ \segment $ intersecting at least $ \ell = k ( 2 ^{ k + 2 } + 4 ( k + 1 ) ) + ( k - 1 ) $ reflex segments of $ \polygon $.
    Then any conforming partition of $ \polygon $ has stabbing number at least $ k + 1 $.
\end{lemma}

\begin{figure}[ht]
    \centering
    \includegraphics[scale=\graphicsScale,page=1]{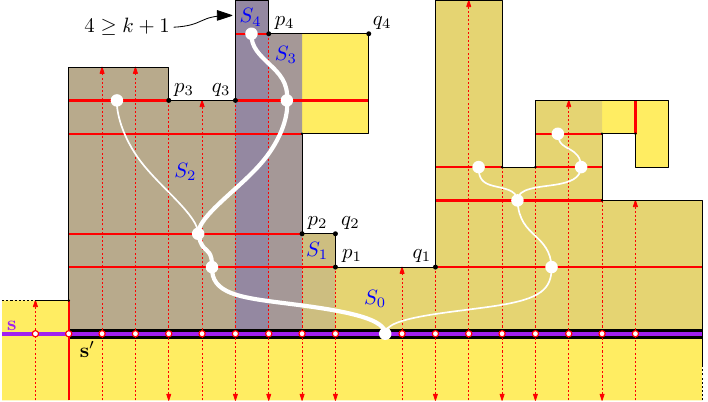}
    \caption{Illustration of the proof of Lemma~\ref{lem:gridStabbing}: If the segment $ \segment $ intersects a lot of vertical reflex segments (these intersections are hollow red circles), then the stabbing number is high (at least $ k + 1 $). The idea is a binary search in the tree drawn with white curvy edges. The tree is rooted at $ \segment ' $ and has a node for each horizontal reflex segment of the form $ \horOf{ p } $, where $ p $ is a reflex vertex above $ \segment ' $ and $ \verOf{ p } $ intersects $ \segment ' $ (note that not all horizontal reflex segments give rise to a node). The binary search recursively chooses the node of the tree that has the most reflex segments intersecting $ \segment ' $ from above (drawn with dotted red segments ending in a downward arrow). The branch of the tree resulting from this binary search is in bold white.}
    \label{fig:bigGrid}
\end{figure}

\begin{proof}
(This proof is illustrated in Figure~\ref{fig:bigGrid}.)
Assume without loss of generality that $ \segment $ is horizontal.
Assume for a contradiction that $ \polygon $ admits a conforming partition $ \partition $ with stabbing number at most $ k $.
The segments in $ \partition $ intersect $ \segment $ no more than $ k - 1 $ points, thereby inducing a partition of $ \segment $ into $ k $ sub-segments.
With our choice of $ \ell $, one of these sub-segments of $ \segment $, say $ \segment ' $, intersects at least $ 2 ^{ k + 2 } + 4 ( k + 1 ) $ vertical reflex segments of $ \polygon $ that are not in $ \partition $.
Up to mirroring vertically, at least half of these $ 2 ^{ k + 2 } + 4 ( k + 1 ) $ reflex segments are \emph{from above} $ \segment ' $, i.e., are of the form $ \verOf{ p } $, where $ p $ is a reflex vertex above $ \segment ' $ and $ \verOf{ p } $ intersects $ \segment ' $.
Note that, by construction of $ \segment ' $, none of these $ \verOf{ p } $ from above $ \segment ' $ are in the partition $ \partition $, which implies that the corresponding $ \horOf{ p } $ are in $ \partition $ instead.

The following construction exhibits a non-empty set of vertical stabbing segments each of which intersects at least $ k + 1 $ of these $ \horOf{ p } $.
The idea is to recursively refine a variable $ S $ initialized as the set $ S _ 0 $ of all the vertical stabbing segments of $ \polygon $ intersecting $ \segment ' $, using a binary search. (The underlying tree of this binary search is drawn in Figure~\ref{fig:bigGrid}.)

Among the horizontal edges of $ \polygon $ above $ \segment ' $, let $ \sgt{ p _ 1 }{ q _ 1 } $ be the \emph{unique} one which is the closest to $ \segment ' $; there is no tie because $ \polygon $ has no gates.
We first remove from $ S $ the stabbing segments which end on $ \sgt{ p _ 1 }{ q _ 1 } $.
The stabbing segments remaining in $ S $ are partitioned into two sets: the ones that stab $ \horOf{ p _ 1 } \in \partition $ (forming the \emph{left set}), and the ones that stab $ \horOf{ q _ 1 } \in \partition $ (forming the \emph{right set}).
One of these two sets, say the left set, contains at least about half ($ 2 ^{ k } + 2 k $, to be more precise) of the vertical reflex segments from above $ \segment ' $.
We keep the left set and discard the right set ($ S = S _ 1 $ in Figure~\ref{fig:bigGrid}).
We recursively repeat this binary search process until both the right set and the left set are empty ($ S = S _ 4 $ in Figure~\ref{fig:bigGrid}).
Each time, the stabbing segments we keep in $ S $ stab one more horizontal segment in the partition $ \partition $.
Our choice of $ \ell $ ensures a total of at least $ k + 1 $ iterations before running out of reflex segments from above $ \segment ' $.
At the end of the process, $ S $ has the desired property to raise a contradiction:
$ S $ is a non-empty set of vertical stabbing segments each of which intersects at least $ k + 1 $ reflex segments in $ \partition $.
\end{proof}

\begin{figure}[ht]
    \centering
    \includegraphics[scale=\graphicsScale,page=1]{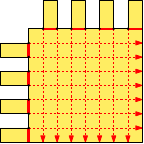}
    \caption{The polygon $ \comb $ (with two rows of~$ 4 $ teeth and $ n = 16 $ vertices).}
    \label{fig:comb}
\end{figure}

\medskip
Unfortunately, Lemma~\ref{lem:gridStabbing} does not carry over to polygons that have gates.
Consider the polygon $ \comb $ that is essentially a square, except that we attach $ \Theta ( n ) $ \emph{teeth} (i.e., small rectangles) on the left side and $ \Theta ( n ) $ teeth on the top side (Figure~\ref{fig:comb}).
The polygon $ \comb $ has conforming stabbing number~$ 2 $ (consider the conforming partition consisting of the reflex segments that ``cut off'' all the teeth). 
Yet, $ \comb $ contains a stabbing segment intersecting $ \Theta ( n ) $ reflex segments of $ \comb $ (i.e. $ m = \Omega ( n ) $), and in fact, the pixel graph of $ \comb $ contains a $ \Theta ( n ) \times \Theta ( n ) $-grid and has treewidth $ \Theta ( n ) $. 
Thus, a dichotomy based on the treewidth of $ \comb $ or on $ m $ does not work. 
We conjecture that some other approach would work and leave one important open problem.

\begin{conjecture}
    Testing whether a rectilinear polygon $ \polygon $ without holes has conforming stabbing number at most $ k $ is fixed-parameter tractable in $ k $. 
\end{conjecture}

\section{Conclusion}

In this paper, we show that computing a conforming partition of a rectilinear polygon with stabbing number $ k $ is $ \NP $-hard for all $ k \geq 4 $. Since the reduction uses only thin polygons, the hardness result follows even if we omit the conforming constraint. The polygons used in our reduction have holes. Therefore, determining the time complexity of computing an optimal (conforming) partition for polygons without holes remains open.

On the positive side, we provide an $ \OO ( n \log n ) $-time algorithm to decide whether a polygon admits a conforming partition with stabbing number~$ 2 $.
Since the problem is $ \NP $-hard already for conforming stabbing number~$ 4 $, only the case of conforming stabbing number~$ 3 $ remains open.
For polygons (possibly with holes) with bounded treewidth and bounded conforming stabbing number, we give a quadratic-time algorithm to compute the minimum stabbing number. 
An exciting direction would be to design fixed-parameter tractable algorithms for polygons without holes parameterized by the conforming stabbing number, which would complement the hardness result for polygons with holes. Interestingly, for polygons without holes that are in general position, we already gave such a fixed-parameter-tractable algorithm. 
But we also proved that general position does not help the case of polygons with holes: computing a conforming partition with stabbing number at most $ k $ (for $ k \geq 4 $) remains $ \NP $-hard for polygons in general position.

Extending all these results to higher dimensions would be interesting, even for the restricted class of orthogonal 3D-histograms where previous results focus on minimizing the number of partitions into rectangular boxes~\cite{biedl2018partitioning,floderus20183d}.

\ARXIVVERSIONONLY{
    \bibliographystyle{plainurl}
}
\ELSVERSIONONLY{
    \bibliographystyle{elsarticle-num}
    \biboptions{sort&compress}
}
\bibliography{bibl}

\appendix

\renewcommand{\thesection}{\Alph{section}}
\makeatletter
\def\@seccntformat#1{\@ifundefined{#1@cntformat}
   {\csname the#1\endcsname.\hspace{0.5em}}
   {\csname #1@cntformat\endcsname}}
\newcommand\section@cntformat{\appendixname\ \thesection.\hspace{0.5em}}
\makeatother

\section{Details of Section~\ref{sec:4+intratable}}
\label{sec:intractableProof}

In this section, we provide a detailed proof of Theorem~\ref{thm:NPHardStabbingN4+}.

The following definitions are used to ease the writing or shorten the writing of some frequent terms in the proofs.
Given a polygon $ \polygon $, a \emph{sub-polygon} $ \polygon ' $ of $ \polygon $ is a polygon included in $ \polygon $ such that an edge (seen as an open segment) of $ \polygon ' $ is either included in an edge of $ \polygon $ or disjoint with the boundary of $ \polygon $.
We call the latter type of edge a \emph{connection edge}.
We extend the definition of a sub-polygon $ \polygon ' $ to any polygon with these two types of edges without specifying a $ \polygon $ containing $ \polygon ' $.
The vertical (respectively horizontal) reflex segments of $ \polygon $ on the boundary of a pixel $ \pixel $ are called the \emph{verticals} (respectively \emph{horizontals}) of $ \pixel $.
We also use the following lemma.

\begin{lemma} \label{lem:consequences(thin)}
    The stabbing number of a rectilinear thin polygon $ \polygon $ (possibly with holes) is equal to the conforming stabbing number of $ \polygon $.
\end{lemma}

\begin{proof}
    Let $ \partition $ be a rectangular partition of $ \polygon $ with a Steiner point $ q $. We show that removing $ q $ and merging some adjacent rectangles yields a rectangular partition of $ \polygon $.

    By definition of a thin polygon, $ q $ is not the intersection of two reflex segments of $ \polygon $.
    Thus, $ q $ is incident to at least one segment $ \segment $ of $ \partition $ which is not a reflex segment of $ \polygon $.
    Now, $ \segment $ is necessarily a full edge (and not only portions of an edge) shared by two rectangles of $ \partition $, because otherwise $ \polygon $ is not thin. 
    Therefore, merging these two rectangles yields a rectangle. Repeating this process for all segments $ \segment $ eventually removes the $ q $ from the partition but does not increase the stabbing number of the partition.
\end{proof}

\subsection{The Forcer Gadget for \texorpdfstring{$ k = 4 $}{k=4} Using Thin Polygons}
\label{sec:forcer-gadget}

\begin{figure}[ht]
    \hspace*{\stretch{1}}
    \subcaptionbox{}{\includegraphics[scale=\graphicsScale,page=3]{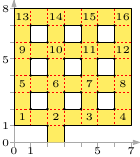}}
    \hspace*{\stretch{2}}
    \subcaptionbox{}{\includegraphics[scale=\graphicsScale,page=5]{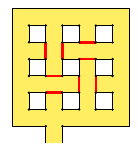}}
    \hspace*{\stretch{2}}
    \subcaptionbox{}{\includegraphics[scale=\graphicsScale,page=6]{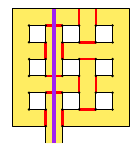}}
    \hspace*{\stretch{2}}
    \subcaptionbox{}{\includegraphics[scale=\graphicsScale,page=12]{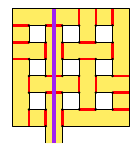}}
    \hspace*{\stretch{1}}
    \caption{
    (a) The generic forcer gadget $ F _ 0 $ is shaded (in yellow). 
    The connection pixel is adjacent to the $ \xcoor $-axis.
    Reflex vertices are small (black) disks, and reflex segments are dotted (in red). The pixels of the reflex vertices are numbered from~$ 1 $ to $ 16 $ as defined in the proof of Lemma~\ref{lem:forcerGadgetStabbingN4}.\\
    (b) The horizontals of pixels $ 6, 11 $ and verticals of pixels $ 7, 10 $ are solid (red) segments. They are one of the two possibilities for a minimal conforming partition of $ F _ 0 $ with stabbing number at most~$ 4 $ (proof of Lemma~\ref{lem:forcerGadgetStabbingN4}\ref{lem:forcerGadgetForall}).\\
    (c) The force-stab of $ F _ 0 $ (thick purple) ends with an arrow pointing outside $ F _ 0 $. The solid (red) verticals of pixels $ 2, 15 $ and horizontals of pixels $ 3, 14 $ are in any minimal conforming partition of $ F _ 0 $ with stabbing number at most~$ 4 $ (proof of Lemma~\ref{lem:forcerGadgetStabbingN4}\ref{lem:forcerGadgetForall}).\\
    (d) The solid (red) reflex segments form a minimal conforming partition of $ F _ 0 $ with stabbing number at most~$ 4 $ (proof of Lemma~\ref{lem:forcerGadgetStabbingN4}\ref{lem:forcerGadgetExists}).
    }
    \label{fig:forcerGadgetStabbingN4}
\end{figure}

Let $ F _ 0 $ be the sub-polygon with holes defined as follows (Figure~\ref{fig:forcerGadgetStabbingN4}(a)). 
The coordinates of the vertices of $ F _ 0 $ in counterclockwise order along the outer boundary are:
\[ ((3, 0), (3, 1), (7, 1), (7, 8), (0, 8), (0, 1), (2, 1), (2, 0)). \]
The set of holes of $ F _ 0 $ is composed of~$ 9 $ squares and is described next:
\[ \{((\mya, b), (\mya + 1, b), (\mya + 1, b + 1), (\mya, b + 1)) :
        (\mya, b) \in \{1, 3, 5 \} \times \{2, 4, 6 \} \}. \]
The only connection edge of $ F _ 0 $ is $ \sgt{ (2, 0) }{ (3, 0) } $.
An rectilinear sub-polygon $ F $ is a \emph{forcer gadget} if there exists a transformation $ \tau $ such that $ \tau $ is the composition of a translation with a rotation with angle in $ \{0, \frac{\pi}{2}, \pi, \frac{3 \pi}{2} \} $ and such that $ F = \tau ( F _ 0 ) $.
Next, we give names to some segments of interest of $ F $.
\begin{itemize}
    \item The pixel with corners $ \tau ( 2 , 0 ) $ and $ \tau ( 3 , 1 ) $ is the \emph{connection pixel} of $ F $.
    \item The stabbing segment containing $ \tau \left( \sgt{ (2.5, 8) }{ (2.5, 0) } \right) $ is the \emph{force-stab} of $ F $ (the thick segment drawn with an arrow pointing outside $ F _ 0 $ in Figure~\ref{fig:forcerGadgetStabbingN4}(c-d)).
\end{itemize}

\begin{lemma} \label{lem:forcerGadgetStabbingN4}
    Let $ F $ be an arbitrary forcer gadget. Then the following holds.
    \begin{thmEnumerate}
        \item\label{lem:forcerGadgetForall} 
        In any conforming partition with stabbing number at most~$ 4 $,
        the force-stab of $ F $ intersects at least~$ 3 $ reflex segments 
        within $ F $.
        \item\label{lem:forcerGadgetExists} The forcer gadget $ F $ admits a partition with stabbing number at most~$ 4 $. 
    \end{thmEnumerate}
\end{lemma}

\begin{proof}
    It is enough to prove that~\ref{lem:forcerGadgetForall} and~\ref{lem:forcerGadgetExists} hold for $ F _ 0 $ considering only minimal conforming partitions. 
    Let $ \partition $ be an arbitrary minimal conforming partition of $ F _ 0 $ with stabbing number at most~$ 4 $.

    We start by naming more parts of $ F _ 0 $.
    The polygon $ F _ 0 $ forms four rows and four columns that we number starting at row~$ 1 $ for the bottom row and at column~$ 1 $ is the left-most column.
    For each $ (\mya, b) \in \{0, 2, 4, 6 \} \times \{1, 3, 5, 7 \} $, the pixel (which is a wedge-pixel of some reflex vertex of $ F _ 0 $) consisting of a unit square with $ (\mya, b) $ as its lower left corner is numbered $ ( \frac{ \mya }{ 2 } + 1 ) + 4 ( \frac{ b - 1 }{ 2 } ) = \frac{ 1 }{ 2 } \mya + 2 b - 1 $ (Figure~\ref{fig:forcerGadgetStabbingN4}(a)).

    \ref{lem:forcerGadgetForall}:
    Pixels $ 6, 7, 10, 11 $ have~$ 2 $ verticals and~$ 2 $ horizontals each.
    Thus, among pixels $ 6, 7, 10, 11 $, $ \partition $ includes at most one pair of verticals per row (among rows $ 2, 3 $ of $ F _ 0 $) and one pair of horizontals per column (among columns $ 2, 3 $ of $ F _ 0 $).
    This leaves only two options: either $ \partition $ includes the horizontals of pixels $ 6, 11 $ and the verticals of pixels $ 7, 10 $ (Figure~\ref{fig:forcerGadgetStabbingN4}(b)), or $ \partition $ includes the horizontals of pixels $ 7, 10 $ and the verticals of pixels $ 6, 11 $.

    Thus, the force-stab of $ F _ 0 $ intersects either the horizontals of pixel~$ 6 $ or of pixel~$ 10 $; regardless, we can \emph{not} use the horizontals of pixel~$ 2 $ and therefore must use its verticals.
    This in turn means that we must use the horizontals of pixel 3, which (since one of pixels~$ 7 $ and~$ 11 $ uses the horizontals) means that we must use the verticals of pixel~$ 15 $ and the horizontals of pixel~$ 14 $ (Figure~\ref{fig:forcerGadgetStabbingN4}(c)).   
    Therefore, column~$ 2 $ of $ F _ 0 $ has at least~$ 3 $ horizontal reflex segments included in $ \partition $, which proves~\ref{lem:forcerGadgetForall}.

    \ref{lem:forcerGadgetExists}:
    Let $ \partition _ 0 $ be the minimal conforming partition of $ F _ 0 $ whose verticals are at pixels $ 1, 2, 5, 7, 10, 12, 15, 16 $ (Figure~\ref{fig:forcerGadgetStabbingN4}(d)).
    There are~$ 3 $ vertical (respectively horizontal) reflex segments of $ \partition _ 0 $ in each row (respectively column) of $ F _ 0 $.
    Thus, $ \partition _ 0 $ has stabbing number~$ 4 $, thereby proving the existence of a conforming partition of $ F _ 0 $ with stabbing number at most~$ 4 $, hence~\ref{lem:forcerGadgetExists}.
\end{proof}

\subsection{The Variable Gadget for \texorpdfstring{$ k = 4 $}{k=4} Using Thin Polygons}
\label{sec:variable-gadget}

\begin{figure}[ht]
    \hspace*{\stretch{1}}
    \subcaptionbox{}{\includegraphics[scale=\graphicsScale,page=1]{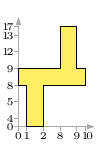}}
    \hspace*{\stretch{2}}
    \subcaptionbox{}{\includegraphics[scale=\graphicsScale,page=2]{variableGadgetStabbingN4.pdf}}
    \hspace*{\stretch{2}}
    \subcaptionbox{}{\includegraphics[scale=\graphicsScale,page=3]{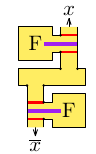}}
    \hspace*{\stretch{2}}
    \subcaptionbox{}{\includegraphics[scale=\graphicsScale,page=5]{variableGadgetStabbingN4.pdf}}
    \hspace*{\stretch{2}}
    \subcaptionbox{}{\includegraphics[scale=\graphicsScale,page=4]{variableGadgetStabbingN4.pdf}}
    \hspace*{\stretch{2}}
    \subcaptionbox{}{\includegraphics[scale=\graphicsScale,page=6]{variableGadgetStabbingN4.pdf}}
    \hspace*{\stretch{1}}
    \caption{
    (a) The polygon $ V _ 0 $ used in the construction of the variable gadget.
    (b) The generic variable gadget $ V _ 1 $ is shaded (in yellow). The squares labeled F are forcer gadgets. 
    Reflex vertices are small (black) disks, and reflex segments are dotted (red). The pixels are numbered from~$ 1 $ to~$ 4 $ as defined in the proof of Lemma~\ref{lem:variableGadgetStabbingN4}.\\
    (c) The horizontals of pixels $ 1, 4 $ are solid (red) segments. They are included in any minimal conforming partition of $ V _ 1 $ with stabbing number at most~$ 4 $.\\
    (d) The negative (respectively positive) out-stab of $ V _ 1$ ends with an arrow pointing downwards (respectively upwards) outside $ V _ 1 $. (A stabbing segment is green if it intersects~$ 2 $ reflex segments, purple if it intersects~$ 3 $ reflex segments.) The variable gadget $ V _ 1 $ is set to false: $ \partition _ 0 $ includes the solid (red) horizontals of pixel~$ 2 $ and verticals of pixel~$ 3 $.\\
    (e) The variable gadget $ V _ 1 $ is set to true: $ \partition _ 1 $ includes the solid (red) verticals of pixel~$ 2 $ and horizontals of pixel~$ 3 $.\\
    (f) The variable gadget $ V _ 1 $ is undetermined: $ \partition _ 2 $ includes the solid (red) horizontals of both pixel~$ 2 $ and pixel~$ 3 $.
    }
    \label{fig:variableGadgetStabbingN4}
\end{figure}

Let $ V _ 0 $ be the polygon without holes defined as follows (Figure~\ref{fig:variableGadgetStabbingN4}(a)).
The coordinates of the vertices of $ V _ 0 $ in counterclockwise order along the boundary are:
\begin{align*}
    (& (1, 0), (2, 0), (2, 8), (10, 8), (10, 9), (9, 9),\\
      & (9, 17), (8, 17), (8, 9), (0, 9), (0, 8), (1, 8)).
\end{align*}
Note that we use a local coordinate system and that the $ \xcoor $-axis of the global coordinate system intersects the variable gadget.
Let $ V _ 1 $ be a sub-polygon without holes defined as the union of $ V _ 0 $ with two forcer gadgets whose connection edges are $ \sgt{ (2, 4) }{ (2, 5) } $ and $ \sgt { (8, 12) }{ (8, 13) } $.
The connection edges of $ V _ 1 $ itself are defined as $ \sgt{ (1, 0) }{ (2, 0) } $ and $ \sgt{ (9, 17) }{ (8, 17) } $.
A rectilinear sub-polygon $ V $ is a \emph{variable gadget} if there exists a horizontal translation $ \tau $ such that $ V = \tau ( V _ 1 ) $.
Next, we give names to some segments of interest of $ V $.
\begin{itemize}
    \item The edge $ \tau \left( \sgt{ (1, 0) }{ (2, 0) } \right) $ is the \emph{negative connection edge} of $ V $ (the bottom segment of the outer boundary of $ V _ 1 $ which is not drawn in Figure~\ref{fig:variableGadgetStabbingN4}(b-f)).
    \item The edge $ \tau \left( \sgt{ (9, 17) }{ (8, 17) } \right) $ is the \emph{positive connection edge} of $ V $ (the top segment of the outer boundary of $ V _ 1 $ which is not drawn in Figure~\ref{fig:variableGadgetStabbingN4}(b-f)).
    \item The stabbing segment containing $ \tau \left( \sgt{ (1.5, 9) }{ (1.5, 0) } \right) $ is the \emph{negative out-stab} of $ V $ (the thick segment drawn with an arrow pointing downwards outside $ V _ 1 $ in Figure~\ref{fig:variableGadgetStabbingN4}(d-e)).
    \item The stabbing segment containing $ \tau \left( \sgt{ (8.5, 8) }{ (8.5, 17) } \right) $ is the \emph{positive out-stab} of $ V $ (the thick segment drawn with an arrow pointing upwards outside $ V _ 1 $ in Figure~\ref{fig:variableGadgetStabbingN4}(d-e)).
\end{itemize}

\begin{lemma} \label{lem:variableGadgetStabbingN4}
    Any variable gadget $ V $ admits exactly three minimal conforming partitions $ \partition _ 0 $, $ \partition _ 1 $, and $ \partition _ 2 $ with stabbing number at most~$ 4 $ such that the following holds (up to relabeling of $ \partition _ 0 $, $ \partition _ 1 $ and $ \partition _ 2 $).
    \begin{itemize}
        \item Within $ V $, exactly~$ 3 $ reflex segments of $ \partition _ 0 $ intersect the positive out-stab of $ V $, and exactly~$ 2 $ reflex segments of $ \partition _ 0 $ intersect the negative out-stab of $ V $. In this case, we say that $ V $ is \emph{set to false} (Figure~\ref{fig:variableGadgetStabbingN4}(d)).
        \item Within $ V $, exactly~$ 2 $ reflex segments of $ \partition _ 1 $ intersect the positive out-stab of $ V $, and exactly~$ 3 $ reflex segments of $ \partition _ 1 $ intersect the negative out-stab of $ V $. In this case, we say that $ V $ is \emph{set to true} (Figure~\ref{fig:variableGadgetStabbingN4}(e)).
        \item Within $ V $, exactly~$ 3 $ reflex segments of $ \partition _ 2 $ intersect the positive out-stab of $ V $, and exactly~$ 3 $ reflex segments of $ \partition _ 2 $ intersect the negative out-stab of $ V $. In this case, we say that $ V $ is \emph{undetermined} (Figure~\ref{fig:variableGadgetStabbingN4}(f)).
    \end{itemize}
\end{lemma}
\begin{proof}
    It is enough to prove Lemma~\ref{lem:variableGadgetStabbingN4} for $ V _ 1 $.
    We start by naming more parts of $ V _ 1 $.
    The segment $ \sgt{ (0, 8.5) }{ (10, 8.5) } $ is the \emph{inner stab} of $ V _ 1 $ (the horizontal thick (green) segment in Figure~\ref{fig:variableGadgetStabbingN4}(d-f)).
    The pixel consisting of a unit square is numbered $ k $ if its lower left corner is (Figure~\ref{fig:variableGadgetStabbingN4}(b)):
    \begin{itemize}
        \item $ (1, 4) $ and $ k = 1 $,
        \item $ (1, 8) $ and $ k = 2 $,
        \item $ (8, 8) $ and $ k = 3 $, or
        \item $ (8, 12) $ and $ k = 4 $.
    \end{itemize}

    Let $ \partition $ be an arbitrary minimal conforming partition of $ V _ 1 $ with stabbing number at most~$ 4 $.
    By Lemma~\ref{lem:forcerGadgetStabbingN4}, $ \partition $ includes the horizontals of pixels $ 1, 4 $ (Figure~\ref{fig:variableGadgetStabbingN4}(c)).

    The two remaining pixels are both intersected by the inner stab of $ V _ 1 $.
    Thus, there are three cases.
    \begin{Cases}
        \item $ \partition $ includes the verticals of pixel~$ 2 $ and the horizontals of pixel~$ 3 $ which corresponds to $ \partition = \partition _ 0 $ (Figure~\ref{fig:variableGadgetStabbingN4}(d)).
        \item $ \partition $ includes the horizontals of pixel~$ 2 $ and the verticals of pixel~$ 3 $ which corresponds to $ \partition = \partition _ 1 $ (Figure~\ref{fig:variableGadgetStabbingN4}(e)).
        \item $ \partition $ includes the horizontals of both pixel~$ 2 $ and pixel~$ 3 $ which corresponds to $ \partition = \partition _ 2 $ (Figure~\ref{fig:variableGadgetStabbingN4}(f)).
    \end{Cases}
\end{proof}

\subsection{The Split Gadget for \texorpdfstring{$ k = 4 $}{k=4} Using Thin Polygons}
\label{sec:split-gadget}

\begin{figure}[ht]
    \hspace*{\stretch{1}}
    \subcaptionbox{}{\includegraphics[scale=\graphicsScale,page=1]{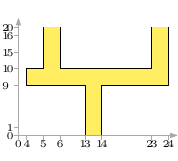}}
    \hspace*{\stretch{2}}
    \subcaptionbox{}{\includegraphics[scale=\graphicsScale,page=2]{splitGadgetStabbingN4.pdf}}
    \hspace*{\stretch{2}}
    \subcaptionbox{}{\includegraphics[scale=\graphicsScale,page=3]{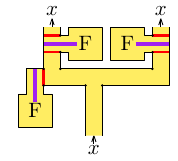}}
    \hspace*{\stretch{1}}\\
    \hspace*{\stretch{1}}
    \subcaptionbox{}{\includegraphics[scale=\graphicsScale,page=5]{splitGadgetStabbingN4.pdf}}
    \hspace*{\stretch{2}}
    \subcaptionbox{}{\includegraphics[scale=\graphicsScale,page=4]{splitGadgetStabbingN4.pdf}}
    \hspace*{\stretch{2}}
    \subcaptionbox{}{\includegraphics[scale=\graphicsScale,page=6]{splitGadgetStabbingN4.pdf}}
    \hspace*{\stretch{1}}
    \caption{The figure is not to scale.
    (a) The polygon $ S _ 0 $.
    (b) The generic split gadget $ S _ 1 $ is shaded (in yellow). The squares labeled F are forcer gadgets. The connection segments are the only segments of the boundary of $ S _ 1 $ which are not solid (black). The reflex vertices are small (black) disks and the reflex segments are dotted (in red). The pixels are numbered from~$ 1 $ to~$ 6 $ as defined in the proof of Lemma~\ref{lem:splitGadgetStabbingN4}.\\
    (c) The verticals of pixel~$ 1 $ and the horizontals of pixels~$ 5, 6 $ are solid (red) segments. They are included in any minimal conforming partition of $ S _ 1 $ with stabbing number at most~$ 4 $.\\
    (d) The out-stabs (respectively in-stabs) of $ S _ 1 $ ends with an arrow pointing outside (respectively inside) $ S _ 1 $. (A stabbing segment is green if it intersects~$ 2 $ reflex segments, purple if it intersects~$ 3 $ reflex segments.) The split gadget $ S _ 1 $ ``propagates false'': $ \partition _ 0 $ includes the solid (red) horizontals of pixels~$ 2, 4 $ and verticals of pixel~$ 3 $.\\
    (e) The split gadget $ S _ 1 $ ``propagates true'': $ \partition _ 1 $ includes the solid (red) verticals of pixel~$ 2, 4 $ and horizontals of pixel~$ 3 $.\\
    (f) One of the three cases where the value propagated by some of the out-stabs (here the right out-stab) is decreased compared to the value propagated by the in-stab.
    }
    \label{fig:splitGadgetStabbingN4}
\end{figure}

Let $ S_0 $ be the polygon without holes defined as follows (Figure~\ref{fig:splitGadgetStabbingN4}(a)).
The coordinates of the vertices of $ S_0 $ in counterclockwise order along the boundary are:
\begin{align*}
    (& (13, 0), (14, 0), (14, 9), (24, 9), (24, 20), (23, 20), (23, 10), \\
      & (6, 10), (6, 20), (5, 20), (5, 10), (4, 10), (4, 9), (13, 9)).
\end{align*} 
Let $ S_1 $ be the sub-polygon without holes defined as the union of $ S_0 $ with three forcer gadgets whose connection edges are $ \sgt{ (4, 9) }{ (5, 9) } $, $ \sgt{ (6, 15) }{ (6, 16) } $, and $ \sgt{ (23, 15) }{ (23, 16) } $.
The connection edges of $ S _ 1 $ itself are defined as $ \sgt{ (13, 0) }{ (14, 0) } $, $ \sgt{ (5, 20) }{ (6, 20) } $, and $ \sgt{ (23, 20) }{ (24, 20) } $.
A rectilinear polygon $ S $ is a \emph{split gadget} in the two following cases.
\begin{itemize}
    \item If there exists a translation $ \tau $ such that $ S = \tau (S_1) $, in which case $ S $ is called a \emph{positive} split gadget.
    \item If there exists a transformation $ \tau $ such that $ \tau $ is the composition of a horizontal reflection with a translation and such that $ S = \tau (S_1) $, in which case $ S $ is called a \emph{negative} split gadget.
\end{itemize}
Next, we give names to some segments of interest of $ S $.
\begin{itemize}
    \item The edge $ \tau \left( \sgt{ (13, 0) }{ (14, 0) } \right) $ is the \emph{in connection edge} of $ S $ (the bottom segment of the outer boundary of $ S_1 $ which is not drawn in Figure~\ref{fig:splitGadgetStabbingN4}(a)).
    \item The edge $ \tau \left( \sgt{ (5, 20) }{ (6, 20) } \right) $ is the \emph{left connection edge} of $ S $ (the top left segment of the outer boundary of $ S_1 $ which is not drawn in Figure~\ref{fig:splitGadgetStabbingN4}(a)).
    \item The edge $ \tau \left( \sgt{ (23, 20) }{ (24, 20) } \right) $ is the \emph{right connection edge} of $ S $ (the top right segment of the outer boundary of $ S_1 $ which is not drawn in Figure~\ref{fig:splitGadgetStabbingN4}(a)).
    \item The stabbing segment containing $ \tau \left( \sgt{ (13.5, 0) }{ (13.5, 10) } \right) $ is the \emph{in-stab} of $ S $ (the thick segment drawn with an arrow pointing upwards inside $ S_1 $ in Figure~\ref{fig:splitGadgetStabbingN4}(d-f)).
    \item The stabbing segment containing $ \tau \left( \sgt{ (5.5, 9) }{ (5.5, 20) } \right) $ is the \emph{left out-stab} of $ S $ (the leftmost thick segment drawn with an arrow pointing upwards outside $ S_1 $ in Figure~\ref{fig:splitGadgetStabbingN4}(d-f)).
    \item The stabbing segment containing $ \tau \left( \sgt{ (23.5, 9) }{ (23.5, 20) } \right) $ is the \emph{right out-stab} of $ S $ (the rightmost thick segment drawn with an arrow pointing upwards outside $ S_1 $ in Figure~\ref{fig:splitGadgetStabbingN4}(d-f)).
\end{itemize}

\begin{lemma} \label{lem:splitGadgetStabbingN4}
    Let $ S $ be an arbitrary split gadget.
    The following holds. 
    \begin{thmEnumerate}
        \item\label{item:1,lem:splitGadgetStabbingN4} Consider an arbitrary minimal conforming partition $ \partition $ of $ S $ with stabbing number at most~$ 4 $.  If~$ 0 $ reflex segments of $ \partition $ intersect the in-stab of $ S $ within $ S $,
        then exactly~$ 3 $ reflex segments of $ \partition _ 0 $ intersect the left out-stab of $ S $ within $ S $, and exactly~$ 3 $ reflex segments of $ \partition _ 0 $ intersect the right out-stab of $ S $ within $ S $ (Figure~\ref{fig:splitGadgetStabbingN4}(d)).
        \item\label{item:2,lem:splitGadgetStabbingN4} For any $ i \in \{ 0 , 1 \} $, 
        there exists a minimal conforming partition $ \partition _ i $ of $ S $ with stabbing number at most~$ 4 $ such that exactly~$ i $ reflex segments of $ \partition _ i $ intersect the in-stab of $ S $ within $ S $, and each out-stab of $ S $ intersects exactly $ 3 - i $ reflex segments of $ \partition _ i $ within $ S $ (Figure~\ref{fig:splitGadgetStabbingN4}(e)).
    \end{thmEnumerate}
\end{lemma}

Recall that a stabbing segment intersecting two gadgets propagates value~$ 0 $ exactly if it intersects three reflex segments in the gadget where it is the out-stab.
Such a segment therefore cannot intersect reflex segments in the gadget where it is the in-stab (in a partition with stabbing number at most~$ 4 $).
In this context, we rephrase Lemma~\ref{lem:splitGadgetStabbingN4}\ref{item:1,lem:splitGadgetStabbingN4} as:
\begin{quote}
    If the in-stab of a split gadget propagates~$ 0 $, then both of its out-stabs also propagate~$ 0 $.
\end{quote}

\begin{proof}
    It is enough to prove Lemma~\ref{lem:splitGadgetStabbingN4} for $ S_1 $.
    We start by naming more parts of $ S_1 $.
    The segment $ (4, 9.5) (24, 9.5) $ is the \emph{inner stab} of $ S_1 $ (the horizontal thick (purple) segment in Figure~\ref{fig:splitGadgetStabbingN4}(d-e)).
    The pixel consisting of a unit square is numbered $ k $ if its lower left corner is (Figure~\ref{fig:splitGadgetStabbingN4}(b)):
    \begin{itemize}
        \item $ (4, 9) $ and $ k = 1 $,
        \item $ (5, 9) $ and $ k = 2 $,
        \item $ (13, 9) $ and $ k = 3 $,
        \item $ (23, 9) $ and $ k = 4 $,
        \item $ (5, 15) $ and $ k = 5 $, or
        \item $ (23, 15) $ and $ k = 6 $.
    \end{itemize}

    Let $ \partition $ be an arbitrary minimal conforming partition of $ S_1 $ with stabbing number at most~$ 4 $.
    By Lemma~\ref{lem:forcerGadgetStabbingN4}, $ \partition $ includes the verticals of pixel~$ 1 $ and the horizontals of pixels~$ 5, 6 $ (Figure~\ref{fig:splitGadgetStabbingN4}(c)).

    Pixels~$ 2, 3, 4 $ are all intersected by the inner stab of $ S_1 $.
    Thus, we have the following two cases.
    \begin{Cases}
        \item If $ \partition $ includes the verticals of pixel~$ 3 $, then $ \partition $ includes the horizontals of pixel~$ 2, 4 $  (Figure~\ref{fig:splitGadgetStabbingN4}(d)).
        Thus, in this case, the partition $ \partition $ satisfies the assertion~\ref{item:1,lem:splitGadgetStabbingN4} of Lemma~\ref{lem:splitGadgetStabbingN4}, and can serve as $ \partition _ 0 $ for~\ref{item:2,lem:splitGadgetStabbingN4}. 
        \item If not, $ \partition $ includes the horizontals of pixel~$ 3 $.
        Therefore, the in-stab intersects a reflex segment and \ref{item:1,lem:splitGadgetStabbingN4} trivially holds. 
        As for\ref{item:2,lem:splitGadgetStabbingN4}, the partition including the horizontals of pixel~$ 3 $, and the verticals of pixel~$ 2, 4 $ is suitable to serve as $ \partition _ 1 $ (Figure~\ref{fig:splitGadgetStabbingN4}(e)).
    \end{Cases}
\end{proof}

\subsection{The Clause Gadget for \texorpdfstring{$ k = 4 $}{k=4} Using Thin Polygons}
\label{sec:clause-gadget}

\begin{figure}[ht]
    \subcaptionbox{}{\includegraphics[scale=\graphicsScale,page=1]{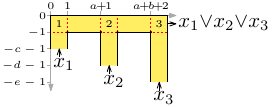}}
    \subcaptionbox{}{\includegraphics[scale=\graphicsScale,page=3]{clauseGadgetStabbingN4.pdf}}
    \subcaptionbox{}{\includegraphics[scale=\graphicsScale,page=4]{clauseGadgetStabbingN4.pdf}}
    \subcaptionbox{}{\includegraphics[scale=\graphicsScale,page=5]{clauseGadgetStabbingN4.pdf}}
    \caption{The figure is not to scale.
    (a) The generic clause gadget $ C _ 0 ( \mya , b , c , d , e ) $ is shaded (in yellow). The connection segments are the only segments of the boundary of $ C _ 0 ( \mya , b , c , d , e ) $ which are not solid (black). The reflex vertices are small (black) disks and the reflex segments are dotted (in red). The pixels are numbered from~$ 1 $ to~$ 3 $ as defined in the proof of Lemma~\ref{lem:clauseGadgetStabbingN4}.\\
    (b) The minimal conforming partition $ \partition _{ 0 0 0 } $ with stabbing number~$ 5 $. The verticals of pixels $ 1, 2, 3 $ are solid (red) segments. The in-stabs (respectively inner stab) of $ C _ 0 ( \mya , b , c , d , e )$ ends with an arrow pointing inside (respectively outside) $ C _ 0 ( \mya , b , c , d , e ) $. (An in-stab is green if it intersects~$ 2 $ reflex segments, purple if it intersects~$ 3 $ reflex segments. Yet, the inner stab is green if it intersects~$ 3 $ reflex segments or less, and purple otherwise.) The clause gadget $ C _ 0 ( \mya , b , c , d , e ) $ ``propagates false'': $ \partition _ 0 $ includes the solid (red) horizontals of pixels $ 2, 4 $ and verticals of pixel~$ 3 $.\\
    (c) The minimal conforming partition $ \partition _{ 0 0 1 } $ with stabbing number~$ 4 $.\\
    (d) The minimal conforming partition $ \partition _{ 0 1 0 } $ with stabbing number~$ 3 $.
    }
    \label{fig:clauseGadgetStabbingN4}
\end{figure}

Let $ \mya $, $ b $, $ c $, $ d $, and $ e $ be positive integers, and $ C _ 0 ( \mya , b , c , d , e ) $ be the sub-polygon without holes defined as follows (Figure~\ref{fig:clauseGadgetStabbingN4}).
The coordinates of the vertices of $ C _ 0 ( \mya , b , c , d , e ) $ in counterclockwise order along the boundary are:
\begin{align*}
    (& (0, 0), (0, - c - 1), (1, - c - 1), (1, - 1), (\mya + 1, - 1), (\mya + 1, - d - 1), \\
    & (\mya + 2, - d - 1), (\mya + 2, - 1), (\mya + b + 2, - 1), (\mya + b + 2, - e - 1), \\
    & (\mya + b + 3, - e - 1), (\mya + b + 3, 0)).
\end{align*}
The connection edges of $ C _ 0 ( \mya , b , c , d , e ) $ are:
\begin{align*}
    & \sgt{ (0, - c - 1) }{ (1, - c - 1) } , \\
    & \sgt{ (\mya + 1, - d - 1) }{ (\mya + 2, - d - 1) } , \\
    & \sgt{ (\mya + b + 2, - e - 1) }{ (\mya + b + 3, - e - 1) }.
\end{align*}
A rectilinear sub-polygon $ C ( \mya , b , c , d , e ) $ is a \emph{clause gadget} in the following two cases.
\begin{itemize}
    \item If there exists a translation $ \tau $ such that $ C ( \mya , b , c , d , e ) = \tau (C _ 0 ( \mya , b , c , d , e )) $, then $ C ( \mya , b , c , d , e ) $ is a called a \emph{positive} clause gadget.
    \item If there exists a transformation $ \tau $ such that $ \tau $ is the composition of a horizontal reflection with a translation and such that $ C ( \mya , b , c , d , e ) = \tau (C _ 0 ( \mya , b , c , d , e )) $, then $ C ( \mya , b , c , d , e ) $ is a called a \emph{negative} clause gadget.
\end{itemize}
Next, we give names to some segments of interest of $ C ( \mya , b , c , d , e ) $.
\begin{itemize}
    \item The edge $ \tau \left( \sgt{ (0, - c - 1) }{ (1, - c - 1) } \right) $ is the \emph{left connection edge} of $ C ( \mya , b , c , d , e ) $ (the bottom left segment of the outer boundary of $ C _ 0 ( \mya , b , c , d , e ) $ which is not drawn in Figure~\ref{fig:clauseGadgetStabbingN4}).
    \item The edge $ \tau \left( \sgt{ (\mya + 1, - d - 1) }{ (\mya + 2, - d - 1) } \right) $ is the \emph{center connection edge} of $ C ( \mya , b , c , d , e ) $ \ (the bottom center segment of the outer boundary of $ C _ 0 ( \mya , b , c , d , e ) $ which is not drawn in Figure~\ref{fig:clauseGadgetStabbingN4}).
    \item The edge $ \tau \left( \sgt{ (\mya + b + 2, - e - 1) }{ (\mya + b + 3, - e - 1) } \right) $ is the \emph{right connection edge} of $ C ( \mya , b , c , d , e ) $ (the bottom right segment of the outer boundary of $ C _ 0 ( \mya , b , c , d , e ) $ which is not drawn in Figure~\ref{fig:clauseGadgetStabbingN4}).
    \item The stabbing segment containing $ \tau \left( \sgt{ (0.5, 0) }{ (0.5, - c - 1) } \right) $ is the \emph{left in-stab} of the clause gadget $ C ( \mya , b , c , d , e ) $ (the leftmost thick segment drawn with an arrow pointing upwards inside $ C _ 0 ( \mya , b , c , d , e ) $ in Figure~\ref{fig:clauseGadgetStabbingN4}(b-e)).
    \item The stabbing segment containing $ \tau \left( \sgt{ (\mya + 1.5, 0) }{ (\mya + 1.5, - d - 1) } \right) $ is the \emph{center in-stab} of $ C ( \mya , b , c , d , e ) $ (the center thick segment drawn with an arrow pointing upwards inside $ C _ 0 ( \mya , b , c , d , e ) $ in Figure~\ref{fig:clauseGadgetStabbingN4}(b-e)).
    \item The stabbing segment containing $ \tau \left( \sgt{ (\mya + b + 2.5, 0) }{ (\mya + b + 2.5, - e - 1) } \right) $ is the \emph{right in-stab} of $ C ( \mya , b , c , d , e ) $ (the rightmost thick segment drawn with an arrow pointing upwards inside $ C _ 0 ( \mya , b , c , d , e ) $ in Figure~\ref{fig:clauseGadgetStabbingN4}(b-e)).
\end{itemize}

\begin{lemma} \label{lem:clauseGadgetStabbingN4}
    Any clause gadget $ C ( \mya , b , c , d , e ) $ admits exactly~$ 8 $ minimal conforming partitions.
    Specifically, these~$ 8 $ minimal conforming partitions are the $ \partition _{ \variable _ 1 \variable _ 2 \variable _ 3 } $ such that exactly $ \variable _ 1 \in \{ 0 , 1 \} $ (respectively $ \variable _ 2 $, $ \variable _ 3 $) reflex segments of $ \partition _{ \variable _ 1 \variable _ 2 \variable _ 3 } $ intersect the left (respectively center, right) in-stab of $ C ( \mya , b , c , d , e ) $ within $ C ( \mya , b , c , d , e ) $ (Figure~\ref{fig:clauseGadgetStabbingN4}(b-e) show respectively $ \partition _{ 0 0 0 }, \partition _{ 0 0 1 }, \partition _{ 0 1 1 }, \partition _{ 1 1 1 } $).
    Moreover, only $ \partition _{ 0 0 0 } $ has stabbing number greater than~$ 4 $.
\end{lemma}

\begin{proof}
    It is enough to prove Lemma~\ref{lem:clauseGadgetStabbingN4} for $ C _ 0 ( \mya , b , c , d , e ) $.
    We start by naming more parts of $ C _ 0 ( \mya , b , c , d , e ) $.
    The segment $ \sgt{ (0, - 0.5 ) }{ ( \mya + b + 3, - 0.5 ) } $ is the \emph{inner stab} of $ C _ 0 ( \mya , b , c , d , e ) $ (the horizontal thick (purple or green) segment in Figure~\ref{fig:clauseGadgetStabbingN4}(b-e)).
    The pixel consisting of a unit square is numbered $ k $ if its lower left corner is
    (Figure~\ref{fig:clauseGadgetStabbingN4}(a)):
    \begin{itemize}
        \item $ ( 0, - 1 ) $ and $ k = 1 $,
        \item $ ( \mya + 1, - 1 ) $ and $ k = 2 $, or
        \item $ ( \mya + b + 2, - 1 ) $ and $ k = 3 $.
    \end{itemize}

    Considering the horizontals versus the verticals in pixels $ 1, 2, 3 $ indeed shows that there exists exactly~$ 8 $ minimal conforming partition of $ C _ 0 ( \mya , b , c , d , e ) $ which are $ \{ \partition _{ \variable _ 1 \variable _ 2 \variable _ 3 } : \variable _ 1, \variable _ 2 , \variable _ 3 \in \{0, 1 \} \} $ as stated in Lemma~\ref{lem:clauseGadgetStabbingN4}.

    Given that the three pixels $ 1, 2, 3 $ all intersect the inner stab of $ C _ 0 ( \mya , b , c , d , e ) $, we check that the inner stab of $ C _ 0 ( \mya , b , c , d , e ) $ intersects at most~$ 4 $ reflex segments of the $ \partition _{ \variable _ 1 \variable _ 2 \variable _ 3 } $ except $ \partition _{ 0 0 0 } $.
\end{proof}

\subsection{Proof of Theorem~\ref{thm:NPHardStabbingN4+} When \texorpdfstring{$ k = 4 $}{k=4} Using Polygons in General Position} \label{sec:4intractableGeneralPosition}

To prove the $ \NP $-hardness for the polygons in general position, we use the same technique as for proving the hardness results for thin polygons but with modified gadgets. The functionalities of each gadget remain the same; therefore, we only give a high-level overview of the changes.

\paragraph{Forcer gadget\ARXIVVERSIONONLY{.}}
We use a staircase with~$ 6 $ reflex vertices as the force gadget, which is shown in Figure~\ref{fig:forcerGadgetGeneralN4}(a,c). We need the property that in any conforming partition of the gadget with stabbing number~$ 4 $, the maximal stabbing segment $ \segment $ perpendicular to the connection edge stabs~$ 3 $ reflex segments.  Suppose for a contradiction that $ \segment $ intersects smaller than three reflex segments. Then at least four vertical reflex segments in the partition would reach the topmost edge of the gadget, implying a stabbing number higher than~$ 4 $, a contradiction.  We also need the gadget to have a conforming partition with stabbing number at most~$ 4 $, which is shown in Figure~\ref{fig:forcerGadgetGeneralN4}(b).

\begin{figure}[ht]
    \hspace*{\stretch{1}}
    \subcaptionbox{}{\includegraphics[scale=\graphicsScale,page=1]{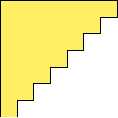}}
    \hspace*{\stretch{2}}
    \subcaptionbox{}{\includegraphics[scale=\graphicsScale,page=3]{forcerGadgetGeneralN4.pdf}}
    \hspace*{\stretch{2}}
    \subcaptionbox{}{\includegraphics[scale=\graphicsScale,page=4]{forcerGadgetGeneralN4.pdf}}
    \hspace*{\stretch{1}}
    \caption{
    (a) The generic gadget.
    (b) The unique conforming partition with stabbing number~$ 4 $.
    (c) A schematic drawing.
    }
    \label{fig:forcerGadgetGeneralN4}
\end{figure}

\paragraph{Variable gadget\ARXIVVERSIONONLY{.}}
The variable gadget (Figure~\ref{fig:variableGadgetGeneralN4}(a)) is a careful perturbation of the variable gadget that we used previously for thin polygons. Following the previously defined variable gadget, we define the positive and negative connection edges and their corresponding positive and negative out-stabs. This variable gadget in general position also has exactly three minimal conforming partitions with stabbing number at  most~$ 4 $. This can be verified by first  observing the stabbing segments imposed by the forcer gadgets, then using a case analysis on the four reflex vertices between the forcer gadgets. Two of these configurations are used to determine the truth values of the variable (Figure~\ref{fig:variableGadgetGeneralN4}(b-c)), while the third one (called \emph{undetermined}) is by convention interpreted as true (Figure~\ref{fig:variableGadgetGeneralN4}(d)).  In a false (true) configuration, exactly~$ 3 $ (exactly~$ 2 $) reflex segments intersect the positive out-stab, and exactly~$ 2 $ (exactly~$ 3 $) reflex segments intersect the negative out-stab. Similar to the previously defined variable gadget, if a positive (negative) out-stab intersects only two reflex segments, then it forces two vertical reflex segments inside the gadget, which enforces the negative (positive) out-stab to intersect three reflex segments.

\begin{figure}[ht]
    \hspace*{\stretch{1}}
    \subcaptionbox{}{\includegraphics[scale=\graphicsScale,page=1]{variablesGadgetGeneralN4.pdf}}
    \hspace*{\stretch{2}}
    \subcaptionbox{}{\includegraphics[scale=\graphicsScale,page=2]{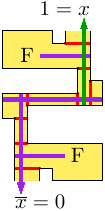}}
    \hspace*{\stretch{2}}
    \subcaptionbox{}{\includegraphics[scale=\graphicsScale,page=3]{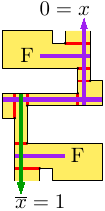}}
    \hspace*{\stretch{2}}
    \subcaptionbox{}{\includegraphics[scale=\graphicsScale,page=4]{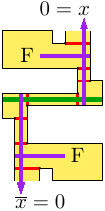}}
    \hspace*{\stretch{1}}
    \caption{
    (a) The variable gadget, with reflex segments that are forced by the forcer gadgets and by the variable gadget itself.
    (b) Reflex segments to encode $ \variable = 1 $.
    (c) Reflex segments to encode $ \variable = 0 $.
    (d) The variable gadget is undetermined.
    }
    \label{fig:variableGadgetGeneralN4}
\end{figure}

\paragraph{Split gadget\ARXIVVERSIONONLY{.}}
The split gadget in general position (Figure~\ref{fig:splitGadgetGeneralN4}(a)) has the same property as the split gadget that we built for thin polygon. This gadget is slightly different as it uses one less forcer gadget.  However, from the perspective of the split gadget, we can still define \emph{in connection edge}, \emph{left connection edge}, and \emph{right connection edge} that connect the split gadget to the rest of the polygon, and the corresponding perpendicular maximal stabbing segments as \emph{in-stab}, \emph{left out-stab} and \emph{right out-stab}. The property that we need for this gadget is that the value of the maximal stabbing segment entering from a variable gadget into the split gadget is propagated (either as it is, or with a decreased value) to the two stabbing segments leaving the split gadget.

Consider first the case when the in-stab does not intersect any reflex segment, i.e., corresponds to the value~$ 0 $ (Figure~\ref{fig:splitGadgetGeneralN4}(c)). We now show the left out-stab (similarly, the right out-stab) must propagate~$ 0 $, i.e., it will intersect~$ 3 $ reflex segments. Note that the forcer gadget near the left-out stab enforces two horizontal reflex segments. If the left-out stab does not intersect any more reflex segments (i.e., if it propagates~$ 1 $), then we must have two vertical reflex segments inside the gadget that are imposed by the out-stab. These two vertical reflex segments together with the vertical reflex segments imposed by the in-stab imply a stabbing number larger than~$ 4 $, a contradiction.

Consider now the case when the in-stab intersects only one reflex segment (corresponding to the value~$ 1 $). We are now free to propagate either~$ 0 $ or~$ 1 $ through the out-stabs. These are illustrated with the partition in Figure~\ref{fig:splitGadgetGeneralN4}(b,d).

\begin{figure}[ht]
    \hspace*{\stretch{1}}
    \subcaptionbox{}{\includegraphics[scale=\graphicsScale,page=1]{splitGadgetGeneralN4.pdf}}
    \hspace*{\stretch{2}}
    \subcaptionbox{}{\includegraphics[scale=\graphicsScale,page=3]{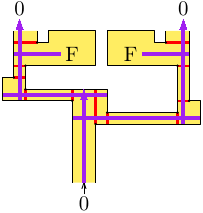}}
    \hspace*{\stretch{1}}\\
    \hspace*{\stretch{1}}
    \subcaptionbox{}{\includegraphics[scale=\graphicsScale,page=2]{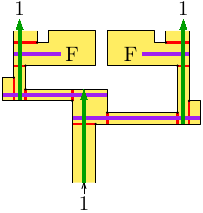}}
    \hspace*{\stretch{2}}
    \subcaptionbox{}{\includegraphics[scale=\graphicsScale,page=4]{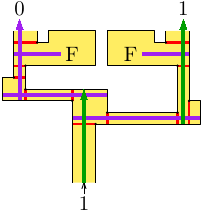}}
    \hspace*{\stretch{1}}
    \caption{
    (a) The split gadget, with reflex segments that are forced by the forcer gadgets and by the split gadget itself.
    (b) A set of reflex segments that propagates~$ 0 $.
    (c) A set of reflex segments that propagates~$ 1 $.
    (d) It is possible to reduce the propagated value (but it cannot increase).
    }
    \label{fig:splitGadgetGeneralN4}
\end{figure}

\paragraph{Clause gadget\ARXIVVERSIONONLY{.}}
A clause gadget in general position has the same properties as the one for thin polygons (Figure~\ref{fig:clauseGadgetGeneralN4}(a)).  The reflex vertices are perturbed such that a stabbing number greater than~$ 4 $ would require all in-stabs to propagate~$ 0 $ values (Figure~\ref{fig:clauseGadgetGeneralN4}(b)). For any other combination of values, there exists a partition with stabbing number at most~$ 4 $ where each in-stab that propagates~$ 1 $ intersects exactly one reflex segment, and each in-stab that propagates~$ 0 $ does not intersect any reflex segment. Figure~\ref{fig:clauseGadgetGeneralN4}(c-d) illustrates such choices when at least one incoming value is~$ 1 $.

\begin{figure}[ht]
    \hspace*{\stretch{1}}
    \subcaptionbox{}{\includegraphics[scale=\graphicsScale,page=1]{clauseGadgetGeneralN4.pdf}}
    \hspace*{\stretch{2}}
    \subcaptionbox{}{\includegraphics[scale=\graphicsScale,page=2]{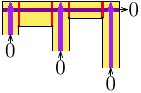}}
    \hspace*{\stretch{2}}
    \subcaptionbox{}{\includegraphics[scale=\graphicsScale,page=3]{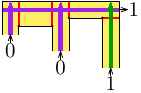}}
    \hspace*{\stretch{2}}
    \subcaptionbox{}{\includegraphics[scale=\graphicsScale,page=4]{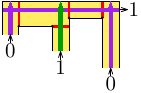}}
    \hspace*{\stretch{1}}
    \caption{
    (a) The generic clause gadget.
    (b) If all incoming values are~$ 0 $, then we require stabbing number~$ 5 $.
    (c-d) If at least one incoming value is~$ 1 $, then there is a choice of reflex segments with stabbing number at most~$ 4 $.
    }
    \label{fig:clauseGadgetGeneralN4}
\end{figure}

\subsection{Proof of Theorem~\ref{thm:NPHardStabbingN4+} When \texorpdfstring{$ k > 4 $}{k>4}} \label{sec:5+intractable}

In this section, we mention the modifications of the forcer gadgets to adapt them to higher values of $ k $.
There are two families of forcer gadgets to cover: the forcer gadgets for thin polygons and the forcer gadgets in general position.

\paragraph{Forcer gadgets for thin polygons\ARXIVVERSIONONLY{.}}
For the case when $ k = 5 $, we use the force gadget as shown in  (Figure~\ref{fig:forcerGadgetStabbingN5+67}(a-b)).
Similarly to the proof of Lemma~\ref{lem:forcerGadgetStabbingN4}, we show that the out-stab intersects at least four reflex segments in any conforming partition $ \mathcal{ R } $ with stabbing number at most~$ 5 $ of $ F $. 
Indeed, considering only the~$ 16 $ pixels which are the wedge-pixels of some reflex vertex, $ \mathcal{ R } $ includes the verticals of at most two wedge-pixels per row, that is, $ \mathcal{ R } $ includes the horizontals of at least two wedge-pixels per row.
Thus, in total, $ \mathcal{ R } $ includes the horizontals of at least eight wedge-pixels.
Since $ \mathcal{ R } $ includes the horizontals of at most two wedge-pixels per column, each column has exactly four horizontal reflex segments included in $ \mathcal{ R } $.
Furthermore, the gadget indeed admits a partition with stabbing number at most~$ 5 $ (Figure~\ref{fig:forcerGadgetStabbingN5+67}(b)).

\begin{figure}[ht]
    \hspace*{\stretch{1}}
    \subcaptionbox{}{\includegraphics[scale=\graphicsScale,page=3]{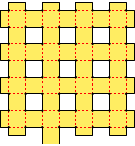}}
    \hspace*{\stretch{2}}
    \subcaptionbox{}{\includegraphics[scale=\graphicsScale,page=5]{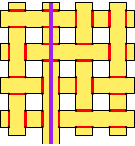}}
    \hspace*{\stretch{2}}
    \subcaptionbox{}{\includegraphics[scale=\graphicsScale,page=1]{forcerGadgetStabbingN6.pdf}}
    \hspace*{\stretch{2}}
    \subcaptionbox{}{\includegraphics[scale=\graphicsScale,page=1]{forcerGadgetStabbingN7.pdf}}
    \hspace*{\stretch{1}}
    \caption{(a) The forcer gadget in the context of $ \STAB{ 5 } $. (b) A partition with stabbing number at most~$ 5 $. The out-stab is thick (in purple). (c) The forcer gadget in the context of $ \STAB{ 6 } $. (d) The forcer gadget in the context of $ \STAB{ 7 } $.}
    \label{fig:forcerGadgetStabbingN5+67}
\end{figure}

\paragraph{Forcer gadgets in general position\ARXIVVERSIONONLY{.}}
For $ k > 4 $, we use as forcer gadget a staircase with $ 2 ( k - 1 ) $ reflex vertices, see Figure~\ref{fig:forcerGadgetStabbingN5+}(d).
Exactly as for $ k = 4 $, one argues that 
(in any conforming partition with stabbing number $ k $)
the stabbing segment $ \segment $ perpendicular to the connection edge stabs $ k - 1 $ reflex segments, since otherwise $ k $ reflex segments would reach the topmost edge of the staircase, which is impossible.

\end{document}